\documentclass[11pt]{article}
\usepackage{enumitem}
\usepackage{pdfpages}
\usepackage{multirow}
\usepackage{array}
\usepackage{graphicx}
\usepackage{subcaption}
\usepackage{latexsym}
\usepackage{pstricks}
\usepackage{pst-node}
\usepackage{pst-tree}
\usepackage{url}
\usepackage{times}
\usepackage{helvet}
\usepackage{courier}
\usepackage{algorithmic}
\usepackage[ruled,vlined,linesnumbered]{algorithm2e}
\usepackage{amsthm}
\usepackage{graphicx,amssymb,amsmath}
\usepackage{mathrsfs}
\usepackage{mathtools}
\usepackage{comment}
\usepackage{framed}
\usepackage{xspace}
\usepackage{comment}

\DeclareMathOperator*{\argmax}{arg\,max}

\DeclareMathOperator*{\E}{\mathbb{E}}
\DeclareMathOperator{\poly}{poly}

\usepackage[hypertexnames=false, colorlinks=true, citecolor=blue, linkcolor=red, urlcolor=black]{hyperref}

\usepackage{geometry}
\makeatletter
\renewcommand*{\verbatim@font}{\sffamily}
\title{Connected Components in Linear Work and Near-Optimal Time\thanks{A preliminary version of this paper appeared in the proceedings of SPAA 2024. This work is in part supported by NSF awards 2128519 and 2044679, an ONR grant, and a DARPA SIEVE grant under a subcontract from SRI.}}
\author{ 
Alireza Farhadi\thanks{Carnegie Mellon University, afarhad2@andrew.cmu.edu.}
\and
S. Cliff Liu\thanks{Carnegie Mellon University, cliffliu@andrew.cmu.edu.}
\and
Elaine Shi\thanks{Carnegie Mellon University, rshi@andrew.cmu.edu.}
}
\date{}

\begin{document}

\clearpage\maketitle

\newcounter{dummy} \numberwithin{dummy}{section}
\newtheorem{lemma}[dummy]{Lemma}
\newtheorem{definition}[dummy]{Definition}
\newtheorem{remark}[dummy]{Remark}
\newtheorem{corollary}[dummy]{Corollary}
\newtheorem{claim}[dummy]{Claim}
\newtheorem{observation}[dummy]{Observation}
\newtheorem{assumption}[dummy]{Assumption}
\newtheorem{conjecture}[dummy]{Conjecture}

\newcounter{dummy2}
\newtheorem{theorem}[dummy2]{Theorem}

\begin{abstract}
    Computing the connected components of a graph is a fundamental problem in algorithmic graph theory. A major question in this area is whether we can compute connected components in $o(\log n)$ parallel time. Recent works showed an affirmative answer in the Massively Parallel Computation (MPC) model for a wide class of graphs. Specifically, Behnezhad et al. (FOCS'19) showed that connected components can be computed in $O(\log d + \log \log n)$  rounds in the MPC model. More recently, Liu et al. (SPAA'20) showed that the same result can be achieved in the standard PRAM model but their result incurs $\Theta((m+n) \cdot (\log d + \log \log n))$ work which is sub-optimal.

    In this paper, we show that for graphs that contain \emph{well-connected} components, we can compute connected components on a PRAM in sub-logarithmic parallel time with \emph{optimal}, i.e., $O(m+n)$ total work. Specifically, our algorithm achieves $O(\log(1/\lambda) + \log \log n)$ parallel time with high probability, where $\lambda$ is the minimum spectral gap of any connected component in the input graph. The algorithm requires no prior knowledge on $\lambda$.
    
    Additionally, based on the \textsc{2-Cycle} Conjecture we provide a time lower bound of $\Omega(\log(1/\lambda))$ for solving connected components on a PRAM with $O(m+n)$ total memory when $\lambda \le (1/\log n)^c$, giving conditional optimality to the running time of our algorithm as a parameter of $\lambda$.
\end{abstract}

\newpage

%\pagenumbering{roman}

\newpage

%\pagenumbering{arabic}
%\setcounter{page}{1}

%\newpage

\newcommand{\Cliff}[1]{{\color{red}[Cliff: #1]}}
\newcommand{\AF}[1]{{\color{purple}[Alireza: #1]}}
\newcommand{\elaine}[1]{{\color{magenta}[elaine: #1]}}

\renewcommand{\Cliff}[1]{}
\renewcommand{\AF}[1]{}
\renewcommand{\elaine}[1]{}

\newcommand{\ignore}[1]{}

\newcommand{\graph}{\ensuremath{G}\xspace}
\newcommand{\mcal}[1]{\ensuremath{\mathcal{#1}\xspace}}

\newpage

\section{Introduction} \label{sec:intro}
Finding the {\it connected components} of an undirected graph (also called
{\it connectivity}) is the following problem.
Given an undirected graph with $n$ vertices and $m$ edges,
mark each vertex with the index of the connected component that it belongs to.
Connectivity is a fundamental problem 
in algorithmic graph theory and has numerous applications. 
It is long known that connected components can be computed 
in $O(m)$ time~\cite{DBLP:journals/siamcomp/Tarjan72}; 
moreover, $O(m)$ is also the best possible on a sequential machine.
Since modern computing architectures are parallel in nature, a line
of works %~\cite{} %\elaine{add cite} 
explored the parallel complexity of connectivity.
Shiloach and Vishkin~\cite{DBLP:journals/jal/ShiloachV82} 
gave a deterministic algorithm that achieved $O(\log n)$ parallel running time
and  $O(m \log n)$ work on an ARBITRARY CRCW PRAM (concurrent-read, concurrent-write with arbitrary write resolution), 
and their algorithm was simplified 
by subsequent works~\cite{DBLP:journals/tc/AwerbuchS87}.
Shiloach and Vishkin~\cite{DBLP:journals/jal/ShiloachV82}'s algorithm 
does not, however, achieve optimality in total work. 
Subsequent works~\cite{DBLP:journals/siamcomp/Gazit91,DBLP:journals/jcss/HalperinZ96,halperin2001optimal}
made further improvements and achieved
$O(\log n)$ parallel running time, and $\Theta(m)$ work (see details in later sections).

For a long time, an interesting question is whether we can beat
the $O(\log n)$-time barrier 
for computing connected components in parallel. 
In this respect, the initial breakthrough was made in the Massively
Parallel Computation (MPC) model \cite{DBLP:journals/jacm/BeameKS17}, a model that is stronger than the classic PRAM model \cite{DBLP:conf/soda/KarloffSV10, DBLP:conf/isaac/GoodrichSZ11, DBLP:conf/podc/AssadiSW19}. 
The work of \cite{DBLP:conf/focs/Andoni} showed an algorithm 
that completes in $O(\log d \cdot \log\log_{m/n} n )$
rounds of the MPC model, where $d$ denotes the maximum diameter over all components of the graph. 
Their result shows 
that for graphs with small diameter, we can indeed overcome the 
$\log n$ barrier, at least in more powerful parallel computation models
such as the MPC model.
The subsequent work of \cite{DBLP:conf/focs/BehnezhadDELM19}
further improved the result, by presenting an MPC algorithm  
that completes in $O(\log d + \log\log_{m/n} n)$ rounds. 
This line of work culminated in the recent work of 
Liu, Tarjan, and Zhong~\cite{liu2020connected}. 
Inspired by the MPC algorithms, \cite{liu2020connected} showed
that even on a PRAM, one can compute 
connectivity in $O(\log d + \log\log_{m/n} n)$ parallel time. 
Note that accomplishing the result on a PRAM 
is much more non-trivial than that in the MPC model, since in the MPC
model, primitives such as sorting 
and prefix sum can be accomplished in constant rounds, whereas
on a CRCW PRAM, these primitives require $\Omega(\log n / \log\log n)$ time even using $\text{poly}(n)$ processors \cite{DBLP:journals/jacm/BeameH89}.

Unfortunately, the work of Liu, Tarjan, and Zhong~\cite{liu2020connected} has two 
drawbacks. First, 
they incur $\Theta(m \cdot (\log d + \log\log_{m/n} n))$ work 
which is non-optimal. Second, 
their result holds only with {\it good probability}, i.e.,  
with probability $1- 1/\poly(m \log n /n)$, which is weaker
than the more commonly adopted {\it high probability} notion, i.e., with probability
$1-1/n^c$ for an arbitrary constant $c > 1$. 
%The work of Liu et al.~\cite{liu2020connected}
%leaves open 
We therefore ask the following natural question:
\begin{itemize}[leftmargin=5mm]
\item[]
{\it Can we overcome the $\log n$-time barrier for parallel connectivity on a PRAM,
while preserving  
optimality in total work?}
\end{itemize}

\subsection{Our Results and Contributions}
In this paper, we answer the question positively.
Specifically, we prove the following theorem:

\begin{theorem} \label{thm:our_main_result}
There is an ARBITRARY CRCW PRAM algorithm that computes connectivity in $O(\log (1 / \lambda) + \log \log n)$ time and $O(m + n)$ work with high probability, where $\lambda$ is 
the minimum spectral gap of any connected component in the input graph.\footnote{Our algorithm requires no prior knowledge on $\lambda$. Throughout the paper, we assume $m \le n^c$ for some constant $c > 0$ for a cleaner presentation in many places. This assumption can be removed by running the algorithm in Section 3 of \cite{liu2020connected} after our Stage $1$ which takes $O(\log d + \log\log n) \le O(\log (1 / \lambda) + \log \log n)$ time and $O(m+n)$ work with high probability (see \S{\ref{sec:overview}}).}
\end{theorem}
Recall that 
the spectral gap $\lambda(\graph) \in [0, 2]$ of a graph $\graph$ 
is the second smallest eigenvalue 
of the normalized Laplacian matrix of $\graph$. 
It is well-known that $\lambda(\graph) > 0$ iff $\graph$ is connected.
In general, the larger the $\lambda$, the more \emph{well-connected}  
the graph is. Therefore, our algorithm  
overcomes the $\log n$ barrier on graphs 
that are unions of well-connected components.
For example, for unions of expander graphs whose 
spectral gap $\lambda$ is a constant, our algorithm
achieves $O(\log \log n)$ parallel time and optimal work.
Such graphs are important in practice since  
several works~\cite{plrg,expansion-socialgraph}
showed that real-world communication and social graphs
have good expansion properties.
Further, for any graph whose  
spectral gap $\lambda \geq 1/n^{o(1)}$, our 
algorithm achieves sub-logarithmic parallel time.

%Gkantsidis et al.~\cite{plrg} showed that power-law random graphs
%have a core with large spectral gap --- such graphs 
%have important applications since 
%they are known to capture the topologies of real-world communication networks.

%While the round complexity of our algorithm is the same as 
%Finally, since we know that  
%the diameter $d = O(\frac{\log n}{\lambda})$, 
%a performance bound in terms of $\log \frac{1}{\lambda}$
%can be viewed as slightly weaker than $\log d$.

%Since the graph's diameter $d \leq O(\frac{\log n}{\lambda})$, \elaine{double check
%this expression} a direct corollary is that 
%connectivity can be computed in 
%$O(\log d + \log \log n)$
%parallel time and $O(m)$ total work on a PRAM.

To obtain our result, we developed several new techniques.
We highlight some of them here due to possible independent interests. 

First, we devise %a linear-work algorithm that contracts the input graph with $n$ vertices to $n/\beta$ vertices in $O((\log \beta)^2)$ time for \emph{any} $\beta \ge \log\log n$. Based on this, we obtain 
a linear-work algorithm that contracts the graph to $n / \beta$ vertices in $O(\log \beta)$ time for \emph{any} $\beta \ge \log n$. 
It is well-known that \emph{random leader election} contracts the graph to $n / \beta$ vertices in $O(\log \beta)$ time for any $\beta > 0$, but it requires either an MPC (see \cite{DBLP:conf/focs/BehnezhadDELM19} and the references therein) or PRAM with super-linear work \cite{reif1984optimal, DBLP:journals/siamcomp/Gazit91}. 
Note that for existing linear-work algorithms, contracting to $n/\beta$ vertices requires $\Omega((\log\beta)^2)$ time and only works for $\beta \ge \log n$ (see \S{\ref{sec:overview}} for details).
Graph contraction is a fundamental graph primitive and we hope our new algorithm can find more applications.

Second, we devise a new technique that densifies the graph without decreasing the component-wise spectral gap, and the total work is linear in the number of edges in the original graph. 
Graph densification is used in \cite{DBLP:conf/focs/Andoni,DBLP:conf/focs/BehnezhadDELM19,liu2020connected, soda23} to achieve $o(\log n)$ running time, but they either take super-linear work (on a PRAM) or can decrease the spectral gap of the graph (by adding extra edges to the graph). 
We hope our technique of graph densification that preserves the component-wise spectral gap can find further applications in graph algorithms and spectral graph theory.

On the lower bound side, the popular \textsc{2-Cycle} Conjecture \cite{DBLP:journals/jacm/RoughgardenVW18,DBLP:conf/podc/AssadiSW19,DBLP:conf/focs/Andoni,DBLP:conf/focs/BehnezhadDELM19} states that any MPC algorithm with $n^{1 - \Omega(1)}$ memory per processor cannot distinguish between one cycle of size $n$ with two cycles of size $n/2$ in $o(\log n)$ rounds. 
Based on this, \cite{DBLP:conf/focs/BehnezhadDELM19} shows that any MPC algorithm in this setting requires $\Omega(\log d)$ rounds to compute connected components of any given graph with diameter $d \ge \poly(\log(n))$ with high probability. 
In this paper, we use a similar approach to show that any PRAM algorithm with total memory $O(m+n)$ requires $\Omega(\log(1/\lambda))$ time to compute connected components of any given graph with minimum spectral gap $\lambda \le 1/\poly(\log(n))$ with high probability, unless the \textsc{2-Cycle} Conjecture is wrong. 
This shows that our algorithm is conditional optimal in the running time as a parameter of $\lambda$.

\subsection{Additional Related Work}
We now review some additional related work. 
A couple of classical works~\cite{DBLP:journals/siamcomp/CookDR86,DBLP:journals/jcss/DietzfelbingerKR94} have shown an $\Omega(\log n)$ lower bound on parallel time for any connectivity algorithm on a PRAM with exclusive writes. 
%The counter-example they used to prove
%this lower bound has small diameter.
%To the best of our knowledge, no such lower bound is known for a CRCW PRAM. 
%Note that the lower bounds of \cite{DBLP:journals/siamcomp/CookDR86,DBLP:journals/jcss/DietzfelbingerKR94} hold in the worst case and do not rule out the possible existence of (parametrized) sub-logarithmic-time algorithms for graphs with ``nice'' properties, even for PRAM with exclusive-writes.
%\elaine{double check this}
The algorithms in our paper require a CRCW PRAM. 
%Thus another interesting open question is how to get similar results for PRAM with exclusive writes or prove impossibility.

%As mentioned, the earlier work of \cite{DBLP:conf/podc/AssadiSW19} achieve the same round complexity on an MPC given the value of $\lambda$. 
It is also interesting to compare our result
with the earlier work of Assadi, Sun, and Weinsten~\cite{DBLP:conf/podc/AssadiSW19}, who
showed an MPC algorithm for computing connectivity
in $O(\log (1/\lambda) + \log \log n)$ rounds when $\lambda$ is \emph{known} to the algorithm beforehand. 
When their algorithm has no prior knowledge on $\lambda$, the running time degenerates to $O(\log (1/\lambda) + \log\log (1/\lambda) \cdot \log \log n)$ and requires slightly more processors.
Our connectivity algorithm does not require prior knowledge on $\lambda$ and achieves the same round complexity on a PRAM. 
As mentioned, the PRAM 
model is less powerful than MPC and thus our result 
requires more non-trivial techniques 
than those of \cite{DBLP:conf/podc/AssadiSW19}.
Indeed, our techniques have no overlaps with those of \cite{DBLP:conf/podc/AssadiSW19}.
In \cite{DBLP:conf/podc/AssadiSW19}, they have two novel ideas. First they show that in the MPC model, independent random walks for all vertices can be found using only small number of rounds. The other idea used in the paper is an $O(\log \log n)$ round algorithm for finding spanning forests in random graphs. They provide a method to transform each connected component of the input graph into a random graph, and later use the algorithm for random graphs to find connected components.

Coy and Czumaj~\cite{coy-czumaj} showed how to make the MPC algorithms of \cite{DBLP:conf/focs/Andoni,DBLP:conf/focs/BehnezhadDELM19}
{\it deterministic}, while preserving
the $O(\log d + \log \log_{m/n} n)$ round complexity.
One interesting future direction is how to {\it deterministically }compute connectivity in sub-logarithmic parallel time 
on a PRAM with $O(m+n)$ processors. 
Note that for this question, no answer is known even if we are willing to tolerate sub-optimality in total work.

In a very recent work, \cite{soda23} showed a deterministic MPC algorithm for computing connected components of any forest in $O(\log d)$ rounds.

\section{Preliminaries} \label{sec:pre}

\subsection{Definitions and Notations}
We assume the ARBITRARY CRCW PRAM model~\cite{DBLP:journals/jal/Vishkin83}, %, which is a weaker model than the COMBINING CRCW PRAM \cite{DBLP:books/daglib}.
where each processor has constant number of cells (i.e., words) in its private memory. All processors
share a common memory array.
The processors run synchronously.
In one step, each processor can read a cell in memory, write to a cell in memory, and do a constant amount of local computation.
Any number of processors can read from or write to the same common memory cell concurrently.
If more than one processor write to the same memory cell at the same time, an arbitrary one succeeds, and the algorithm should be correct no matter which one succeeds.

\paragraph{Notations for undirected graph.}
The input graph $\graph$ can contain self-loops and parallel edges. 
We use $m$ and $n$ to denote the number of edges and the number of vertices 
in the original input graph throughout the paper. 
We may assume that the vertices are named from $1$ to $n$.
We use $d$ to denote the {\it diameter} of the graph, i.e., the length 
of the longest shortest-path between any two vertices.
We often use the term {\it ends} to mean the endpoints of an edge.
We use the notation $V(E)$ as the set of ends of all edges in $E$ for any edge set $E$, and use $E(V)$ as the set of edges in the graph that are adjacent to any vertex set $V$. Given a graph $G$, we also use $V(G)$ and $E(G)$ to denote the set of vertices and the set of edges in $G$, respectively.

For any vertex $v$ in graph $G$, the \emph{degree} of $v$ is the number of edges 
adjacent to $v$ in $G$, denoted by $\deg_G(v)$, or $\deg(v)$ if the graph $G$ is clear
from the context.
Each self-loop should count only once towards the node's degree.
Define the \emph{minimum degree} of $G$ as $\deg(G) = \min_{v \text{ in } G} \deg_G(v)$. 
Define the \emph{neighbors} of $v$ in $G$ as $N_G(v) \coloneqq \{u \mid \text{edge~} (u, v) \in E(G)\}$. Again, we might omit the subscript $G$ if the graph $G$ is clear from the context. 
Note that $\deg_G(v) \ge |N_G(v)|$ as we allow parallel edges in $G$.

\paragraph{Labeled digraph.}
Following the same approach 
as in prior works~\cite{liu2020connected},
we formulate the 
connectivity problem as follows: label each vertex
$v$ with a unique vertex $v.p$ in its component. 
With such a labeling, we can test 
whether two vertices $v$ and $w$ are in the same
connected component in constant time by checking whether
$v.p = w.p$.

To compute connectivity, we shall
make use of a labeled digraph 
in the algorithm (which should not to confused
with the original graph $G$ itself).
We call directed edges in the 
labeled digraph {\it arcs} to differentiate
from {\it edges} in $G$. 
The parent field $.p$ is global for all our subroutines. 
Initially, every vertex's parent is itself, i.e., $v.p = v$.
During the algorithm, each vertex may change
its parent to other vertices in the same connected component. 
%Initially, each vertex $v$ has $v.p = v$. 
For a tree containing only $1$ vertex, we definite its \emph{height} to be $0$. 
For a tree containing at least $2$ vertices, its \emph{height} is the maximum number of arcs in any directed simple path from leaf to root.  
A tree in the labeled digraph is \emph{flat} if it has height at most $1$.
A \emph{(self-) loop} is an arc of the form $(v, v)$. 
We maintain the property that the only cycles in the labeled digraph 
are loops except in one of our subroutines, which restores 
this property at the end of that subroutine.

\paragraph{Additional notations.}
We use the short-hand {\it w.p.} to mean ``with probability'',
and use {\it w.h.p}
to mean ``with high probability''.
%w.p., w.h.p.
Unless otherwise noted, in this paper, ``with high probability''
means 
with probability $1-n^{-c}$
where $c > 1$ can be an arbitrary constant. 
We use the notation $\overline{O}(f(m, n))$ to mean that 
the performance bound 
$O(f(m, n))$ holds in expectation.
For any positive integer $k$, let $[k]$ be the set of all 
positive integers from $1$ to $k$ (inclusive).

\paragraph{Contraction algorithms.}
Given graph $G = (V, E)$, a \emph{contraction} on $G$ is formed by identifying two distinct vertices $u, v \in V$ into a single vertex $v^*$, i.e., replace vertices $u, v$ with $v^*$ in $V$, and replace any edge $(u, w), (v, w)$ with $(v^*, w)$ in $E$ for all $w$. 
A \emph{contraction algorithm} is an algorithm that performs a set of contractions on the input graph. 
To get a correct connected components algorithm, each contraction is limited to vertices from the same component. 
If $f, g$ are both contraction algorithms, then $f \circ g$ is also a contraction algorithm by the union of their sets of contractions.

%Diameter $d$, spectral gap $\lambda$ and conductance $\phi$. \Cliff{Might define conductance in Section 5 not here.}

%Consider an undirected graph \graph possibly with loops and parallel edges.
%Henceforth let  
%$w(u,v)$ be the number of edges
%between $u$ and $v$, and let $\deg(v)$ be the degree of $v$ 
%(where each loop contributes only once to the node's degree).
%\elaine{to expand} \Cliff{should be active roots but we can omit the details here.}
%A \emph{$\gamma$-Shrink algorithm} reduces the number of roots in the labeled digraph from $t$ to $t / \gamma$. 
%By repeating the same algorithm for $O(1)$ times, the number of roots will be $t / \gamma^c$ for any constant $c$ within asymptotically the same work.
%We use $\overline{O}(f)$ to denote $O(f)$ in expectation.

%\elaine{text to massage below:}
%\elaine{move this sentence to roadmap?}

%If all processors are indexed, then it is equivalent to think that each processor has $O(1)$ private memory, because the processor can write and read $O(1)$ cells in the public memory indexed by itself in $O(1)$ time.
%\elaine{i don't understand this  sentence}

\subsection{Preliminaries on Spectral Graph Theory}

\begin{definition}[Normalized Laplacian matrix]
Let $w(u,v)$ be the number of edges between $u$ and $v$.
The \emph{normalized Laplacian matrix} $\mcal{L}$ for \graph is defined as follows:
%\elaine{does loop count 2 towards degree}
\[
\mcal{L}(u, v) = \begin{cases}
1 - \frac{w(v, v)}{\deg(v)} & \text{if $u = v$, and $\deg(v)\neq 0$}\\
-\frac{w(u, v)}{\sqrt{\deg(u) \deg(v)}} & \text{if $u$ and $v$ are adjacent}\\
0 & \text{otherwise}.
\end{cases}
\]
\end{definition}

\begin{definition}[Spectral gap]
The spectral gap of the graph $\graph$, denoted $\lambda(\graph)$,  is defined
to be the second smallest eigenvalue of $\mcal{L}(G)$.
\end{definition}

\begin{definition}[Conductance] \label{def:conductance}
The conductance
of an undirected graph $G = (V, E)$ is defined as 
$$\phi(G) = \min_{\substack{S \subset V s.t.\\ {\sf vol}(S)\leq {\sf vol}(V)/2}} 
\frac{|E(S, \overline{S})|}{{\sf vol}(S)}$$
where $E(S, \overline{S})$ denotes the number
of edges in the cut $(S, \overline{S})$, and 
${\sf vol}(S) = \sum_{v \in S}\deg(v)$.
\end{definition}

\subsection{Previous Results}

\begin{theorem}[\cite{liu2020connected}] \label{thm:ltz_main}
    There is an ARBITRARY CRCW PRAM algorithm using $O(m + n)$ processors that computes the connected components of any given graph of $m$ edges and $n$ vertices.
    W.p. $1 - 1 / \poly(((m+n) \log n) / n)$, it runs in $O(\log d + \log \log_{(m + 2n)/n} n)$ time.\footnote{The running time is simplified to $O(\log d + \log\log_{m/n} n)$ in their presentation by assuming $m \ge 2n$.}
\end{theorem}

%\paragraph{Remarks.}
%The algorithm in Theorem~\ref{thm:ltz_main} proceeds in rounds, where each round updates the labeled digraph corresponding to $G$ and takes $O(1)$ time. 
%If the graph diameter $d = O(1)$ then the algorithm in Theorem~\ref{thm:ltz_main} runs in at most $10^4 \log\log n$ rounds w.p. at least $1 - 1 / \text{poly}(\log n)$.
%, and we shall call this as a subroutine in many places in our algorithm.

The algorithm of \cite{liu2020connected} succeeds w.p. $1 - 1/\poly(\log n)$ instead of $1 - 1/\poly(n)$. 
Moreover, it is not linear-work unless $d = O(1)$ and $m = n^{1 + \Omega(1)}$.

\section{Technical Overview} \label{sec:overview}

For simplicity in presenting the main ideas in our connectivity algorithm, we focus on the critical case when $m = n \poly(\log n)$.\footnote{The case of $m = O(n)$ requires more careful pre-processing that is \emph{not} covered in this overview (see \S{\ref{sec:stage1}}). The case of $m = n^{1 + \Omega(1)}$ is simpler.}
\cite{liu2020connected} showed how to compute connected components in $O(\log d + \log\log n)$ time and $(m (\log d + \log\log n))$ work for this case. 
Our goal is a connected components algorithm that only takes $O(m)$ work. 

A na\"ive idea is to sparsify the graph by deleting edges without affecting the connectivity, 
such that the resulting graph has only $O(m/\poly(\log n))$ edges, then running the algorithm of \cite{liu2020connected} achieves linear work. 
However, there are two challenges:
(\romannumeral1) sparsifying a graph may affect the connectivity, e.g., make 
a connected component disconnected, thus affecting correctness;
and 
(\romannumeral2) sparsifying a graph  
may significantly increase the diameter, and thus 
blow up the parallel time when running, e.g., \cite{liu2020connected}.
%However, it is not clear how to sparsify the graph without affecting connectivity. 
For example, 
consider the straightforward approach of \emph{random edge sampling}, i.e., preserving each edge with probability $1/\poly(\log n)$ independently.\footnote{Random edge sampling has been used frequently in many graph algorithms to achieve linear work, e.g., \cite{karger1995randomized, DBLP:journals/jcss/HalperinZ96, halperin2001optimal}.}  
Unfortunately, random edge sampling can cause a connected component in the graph to become disconnected (e.g., for a collection of path graphs). Furthermore,
it may significantly increase the diameter of the 
graph, e.g.,   
in \S{\ref{sec:diameter}} we construct a graph with $\poly(\log n)$ diameter, such that sampling each edge w.p. $1/\poly(\log n)$ gives a subgraph with diameter $n / \poly(\log n)$, which requires $\Omega(\log n)$ time to compute connected components by the algorithm of \cite{liu2020connected}.
%i.e., no improvement over the existing PRAM algorithm for connectivity \cite{DBLP:journals/siamcomp/Gazit91}.

One of our key ideas is to devise a new method to sparsify the graph to only $O(m/\poly(\log n))$ edges without affecting connectivity.
To achieve this, our algorithm first contracts and consequently densifies the graph to make each vertex (except those in very small connected components) sufficiently large (vertex-) degree. 
At this moment, we then use random edge sampling to sparsify the graph without breaking the connectivity of each connected component w.h.p. (provably). 
For comparison and highlight, recall that in \cite{DBLP:conf/focs/Andoni} (and also in \cite{liu2020connected}), they also densify the graph to minimum degree at least $b$ such that each vertex contracts to a leader w.h.p. after sampling each vertex as a leader w.p. $(\log n) / b$. 
However, their densification is by \emph{adding extra edges to the graph}, which might decrease the spectral gap of a component and thus increase the running time of our algorithm. (Our densification is by contractions so it preserves the spectral gaps of all components.)
More importantly, their densification requires $\Theta(m \log d)$ work since they need to use all the $m$ edges and repeat the broadcasting for $\Theta(\log d)$ times; by contrast, ours only takes $O(m)$ work.

%\subsection{Connectivity with Assumption in the Spectral Gap} \label{subsec:overview_simple}

To make the presentation simpler, we assume that 
%the input graph is connected with spectral gap at least $1/\log n$.
every connected component in the input graph is either small, i.e.,
less than $(\log n)^{10} $ in size (number of vertices), 
or has spectral gap at least $1/\log n$. 
This simplifying assumption can later be removed by several new ideas.
Under this assumption, our algorithm proceeds in the following stages, where each stage takes $O(\log\log n)$ time and $O(m)$ work.

%\newpage
\begin{framed}
\noindent \textbf{Connectivity with known $\lambda \ge 1/\log n$ (the assumption on $\lambda$ will be removed in \S{\ref{subsec:removing_assumption}}):}
\begin{itemize}
    \item \textbf{Stage $1$:} Contract the graph to $n / (\log n)^{20}$ vertices and assign each vertex $(\log n)^{20}$ processors.
    \item \textbf{Stage $2$:} Increase the degree of each vertex to $(\log n)^{5}$ using the $(\log n)^{20}$ processors per vertex.
    \item \textbf{Stage $3$:} Sample each edge w.p. $1/(\log n)$ then compute the connected components of the sampled subgraph.
\end{itemize}
\end{framed}

%As mentioned before, the assumption on $\lambda$ will be removed in \S{\ref{subsec:removing_assumption}} by introducing several new ideas.

We briefly explain the main ideas in each stage here. %A more in-depth description will be given later in this section. 

%\paragraph{Stage 1.}
\subsection{Stage $1$: Contract the Graph to Reduce Vertices}
Reif's algorithm \cite{reif1984optimal} contracts the graph to $n / (\log n)^{20}$ vertices in $O(\log\log n)$ time, however it requires $O(m \log\log n)$ work. 
Gazit's algorithm~\cite{DBLP:journals/siamcomp/Gazit91} contracts the graph to $n / (\log n)^{20}$ vertices in expected $O(m)$ work but takes $\Omega((\log\log n)^2)$ time. 
The EREW PRAM algorithm of Halperin and Zwick~\cite{DBLP:journals/jcss/HalperinZ96} contracts the graph to $n / (\log n)^{20}$ vertices in $O(m)$ work w.h.p., but takes $\Omega((\log\log n)^3)$ time \cite{DBLP:journals/jcss/HalperinZ96}. 
Our goal is $O(m)$ work and $O(\log\log n)$ running time, so none of them fits our requirement here.

In Stage $1$ we first present a linear-work algorithm that reduces the number of vertices in the graph to $n/ \log \log n$ in  $O(\poly(\log\log\log n))$ time. To achieve this reduction in the number of vertices, we borrow some ideas from Gazit's paper. The main tool that we use is a constant-time and linear-work algorithm that reduces the number of vertices by a constant factor. Our constant-shrink algorithm is inspired by the matching MPC algorithm of \cite{DBLP:conf/focs/BehnezhadDELM19}. In each call, our algorithm finds a large  matching between the vertices of the graph such that at least a constant fraction of the vertices are matched. We can then contract these matching edges to reduce the number of vertices by a constant factor in each call.

We then build on top of the constant-shrink algorithm and
devise an algorithm that reduces the number of vertices to $n / \log\log n$ in $O(\poly(\log\log\log n))$ time using one of the ideas from Gazit's algorithm (but boosted to $O(m)$ work w.h.p.): separate the contracted graph into a dense part (for vertices with high degree) and a sparse part. 
The separation is done by repeatedly deleting edges with constant probability in each round,\footnote{For comparison, Gazit's algorithm identifies the sparse part by choosing a random subset of edges with geometrically decreasing cardinality.} such that the vertex with low degree is more likely to be isolated. 
Gazit's algorithm repeats the above process for $O(\log\log n)$ times to get $n / \poly(\log n)$ vertices at the end, which gives $O((\log\log n)^2)$ total time since each application costs $O(\log\log n)$ time. 

On the contrary, we only repeat the above separation process for $O(\log\log\log n)$ times and each of them takes $O(\log\log\log n)$ time. 
One can show that the number of vertices with high degree in the original graph is reduced to $n / (\log n)^{21}$ with high probability. 
%The dense part can be proved to have at most $n / (\log n)^{11}$ vertices. 
We prove that the sparse part has $O(n / \log\log n)$ edges with high probability. 
This is done by showing that each high-degree vertex contributes no edge to the sparse part with high probability, and the total contribution of edges from low-degree vertices is $O(n / \log\log n)$ with high probability by the fact that the number of vertices is at most $n / \log\log n$ and with the use of concentration bounds for self-bounding functions. 
So we can contract the sparse part to at most $n / (\log n)^{21}$ vertices by $O(\log\log n)$ applications of the constant-shrink algorithm in $O(n)$ work since the number of edges is bounded from above by $O(n / \log\log n)$ with high probability.

\subsection{Stage $2$: Increase the Minimum Degree}

Recall that we have run Stage $1$ to contract the graph to $n/(\log n)^{20}$ vertices.
%(the degree on $\log n$ can be arbitrary large by paying a constant multiplicative factor in the running time of Stage $1$). 
In Stage $2$, our goal is to 
increase the degree 
of every vertex (except for those in small connected components)
to at least $(\log n)^{5}$.

\paragraph{Creating a skeleton graph.}
To do this, we first create a skeleton graph $H$ 
which sub-samples the original graph $G$
to at most $(m + n)/(\log n)^2$ edges,
while respecting the 
following property: every connected
component in $H$ either completely matches a connected 
component in the original graph $G$,
or has size at least $(\log n)^{10}$. 
In other words, the small connected components (i.e., of
size less than $(\log n)^{10}$)
in the original graph $G$ will be preserved 
precisely in $H$. Further, although $H$ may disconnect the larger connected 
components in $G$, it still guarantees that 
each new connected component is not too small, i.e., at least $(\log n)^{10}$ in size.
Later in our algorithm, we will simply ignore all vertices in small components since 
it can be shown that they are fully contracted (in $O(\log \log n)$ time).
For all other vertices, 
we can run a truncated version of \cite{liu2020connected} 
on the skeleton graph 
to increase the degree of each vertex to $(\log n)^{5}$
by finding connected vertices within its component.

To create the skeleton graph, 
we perform the following procedure:
\begin{itemize}[leftmargin=5mm]
\item 
Imagine each vertex has a hash table of size $(\log n)^{12}$ (which can be done by indexing the at most $n/(\log n)^{20}$ vertices).
For every edge adjacent to some vertex $v$, we throw the edge
into a random bin of the hash table.
We then tally the 
occupied bins of $v$'s hash table, to estimate  
how many edges $v$ has.
In this way, we can classify each vertex
as either {\it high} or {\it low}, such that with high probability, each low
vertex is guaranted to have at most $(\log n)^{10}$ edges
and each high vertex is guaranteed to have at least $(\log n)^{9}$ edges.
\item 
We then preserve all edges adjacent to low vertices. 
For edges that are adjacent on both ends to high vertices,
we down-sample 
each edge with probability 
$1/(\log n)^3$.
\end{itemize}

We prove that the above procedure produces a skeleton graph
satisfying the desired properties stated above. Moreover, we show that the algorithm can be accomplished in $O(m)$ work and $O(\log \log n)$ time with high probability.

Note that the skeleton graph is used only for increasing the vertex degrees using a 
truncated version of \cite{liu2020connected} as mentioned below, therefore, \emph{it does not matter that the skeleton graph may break apart connected components of the original graph}. The remainder of the algorithm (Stage $3$ and onward) will be operating on the original graph with the vertex degrees increased.

\paragraph{Using \cite{liu2020connected}
in a non-blackbox way to increase vertex degrees.} 
We observe that the potential argument in \cite{DBLP:conf/focs/BehnezhadDELM19} (and thus \cite{liu2020connected}) implies that after $\Theta(\log r + \log\log n)$ rounds of the algorithm, each vertex has edges to all vertices within distance $2^r$ in the original graph. 
This is because each vertex either increases its level or connects to all vertices within distance $2$ (for simplicity, we ignore the issue of altering graph edges).
This implies an algorithm that increases the minimum degree of the graph to $(\log n)^{5}$ in $O(\log\log n)$ time. 
However, such an algorithm requires $\Omega(m \log\log n)$ work since it uses all $m$ edges in each round. 
%Our main idea is to construct a subgraph of the original graph such that for each component in this subgraph, it is either the same component as in the original graph, or contains at least $(\log n)^{10}$. The construction of the subgraph can be done in $O(\log\log n)$ time and $O(m)$ work by classifying vertices according to their degrees by hashing, then preserve all edges adjacent to low-degree vertices and sample edges between high-degree vertices. 
Instead, we run $O(\log\log n)$ rounds of the algorithm in \cite{liu2020connected} on the skeleton graph described above. 
If the component has small size, then it is the same component as in the original graph and we contract it to $1$ vertex. 
Otherwise the component has size at least $(\log n)^{10}$ and there are two cases. 
If the component has diameter at most $(\log n)^{10}$ then calling \cite{liu2020connected} computes this component. 
This component can be different from the component in the input graph but the contracted vertex has at least $(\log n)^{10}$ children, giving at least $(\log n)^{10}/2$ adjacent edges in the contracted graph.  
Otherwise, there exists a shortest path of length at least $(\log n)^{10}$, and each vertex ends up with at least $(\log n)^{5}$ new neighbors by the potential argument, giving a contracted graph with minimum degree at least $(\log n)^{5}$.

\subsection{Stage $3$: Compute Connected Components of the Sampled Graph}
%\paragraph{Stage 3.}
After Stage $2$, we have a graph where each vertex has degree at least $(\log n)^5$.
In Stage $3$, we randomly sample each edge with probability $1/\log n$. 
In \S\ref{sec:spectral}, we use matrix concentration bounds, and prove that as long as each
vertex has sufficiently large degree (at least $(\log n)^{5}$) and the original graph has component-wise spectral gap $\lambda \ge 1/\log n$, the sampled graph has component-wise spectral gap at least $1/\log n - o(1/\log n) > 1/(2\log n)$ (so each component stays connected after sampling).
Moreover, by $d \le O(\log n / \lambda)$ we have that the sampled graph has diameter $O((\log n)^3)$. 
We then apply the technique of \cite{liu2020connected} to the resulting sparsified graph to compute connectivity in $O(\log\log n)$ time and $O(m)$ work.

\subsection{Removing the Assumption on Spectral Gap} \label{subsec:removing_assumption}

Our final algorithm works for graphs with any component-wise spectral gap and it does \emph{not} require prior knowledge on the spectral gap. 

\paragraph{Double-exponential progress on spectral gap assumptions.}
To achieve this, we start with assuming the graph (with $n/(\log n)^{20}$ vertices after Stage $1$) has minimum degree at least $b$ and component-wise spectral gap $\lambda \ge 1/b$ where parameter $b = (\log n)^{20}$ initially.
Then we call the three stages described before to \emph{try} to compute the connected components but limited to $O(\log b)$ running time. 
If the assumption is correct then we finish the computation of all connected components. 
Otherwise, we show that the number of vertices decreases from $n/(\log n)^{20}$ to $n/(\log n)^{40} = n/b^{2}$ such that we have enough processors per vertex to increase the minimum degree to $b^{3/2}$. 
Next, we update parameter $b$ to $b^{3/2}$ and proceed with the new assumption $\lambda \ge 1/b$ by running the three stages again limited to $O(\log b)$ running time in a new phase with the updated $b$ value, and repeat. 

Eventually, we are under the correct assumption $\lambda \in [1/b, 1/b^{2/3})$ in some phase and the algorithm computes the connected components in that phase in $O(\log b) = O(\log(1/\lambda))$ time. 
The total running time over all phases is 
$$O\left(\log b + \log (b^{2/3}) + \log (b^{(2/3)^2}) + \log (b^{(2/3)^3}) + \dots \right) = O(\log(1/\lambda)) .$$
It is easy to see that $O(\log\log n)$ phases are sufficient to obtain a correct connected components algorithm, because any graph with $m \le n^c$ has minimum spectral gap $\lambda \ge (1/\log n)^{(3/2)^{9c\log\log n}}$.

\paragraph{Reducing to linear work.}
One main obstacle is how to implement this algorithm efficiently without blowing up the total work. Specifically, the earlier connectivity algorithm with known $\lambda$ already takes $\Theta(m)$ work, resulting in at most $\Theta(m \log\log n)$ total work since we can now have $O(\log \log n)$ guesses of the spectral gap in the new algorithmic framework. 
One key observation is that Stage $1$ takes $O(m)$ work in total no matter how many steps we run it. 
Therefore, we want to \emph{interweave} the process of a complete execution of Stage $1$ with other stages such that each phase still runs $O(\log b)$ steps of Stages $1$, $2$, and $3$  but the total work from Stage $1$ over all phases is $O(m)$. 
It remains to reduce the work of Stage $2$ and Stage $3$ in each phase to $O(m /\log n)$.

Recall that earlier in Stage $2$, we construct a skeleton graph $H$ to run a truncated version of \cite{liu2020connected} to increase the minimum degree. The construction of $H$ relies on classifying vertices according to their estimated degrees.  The estimation is achieved by hashing edges
into hash tables of their adjacent vertices, and counting
the occupancy of the hash table.
%Observe that a high-degree vertex (with degree at least $b$) has degree at least $b/(\log n)^4$ left after sampling each edge with high probability. \elaine{changed to high prob}
%w.p. $1/(\log n)^3$. 
Earlier, we performed the estimation by hashing all edges
in the original graph, thus leading to $\Theta(m)$ work.
We observe that the estimation can be performed even after we randomly down-sample the edges in the original graph. 
%Therefore, instead of classifying vertices in the original graph, we use a
Specifically, we sample a random subgraph $H'$ of the original graph upfront, e.g., by sampling each edge independently w.p. $1/(\log n)^3$. 
This down-sampled subgraph $H'$
is used only for the purpose of
classifying vertices as high- or low-degree.
To create the skeleton graph $H$, observe that since the edges between high vertices were already sampled with the correct probability into $H'$, we simply copy them from $H'$
into the skeleton graph $H$. It remains to copy the edges adjacent to low vertices into $H$ as well.
This can be accomplished in $O(m')$ work using standard techniques from parallel algorithms~\cite{hagerup1992waste} where $m' = (m+n)/\poly(\log n)$ is the actual number of edges being copied. \elaine{can you read this?}
%The edges between high-degree vertices are already sampled with correct probability into $H$ (by copying $H'$ to $H$), so it remains to add all edges adjacent to low-degree vertices into $H'$. 
%This can be done in $O(m \log n)$ work w.h.p. using our new technique based on padded sorting. 

To reduce the work of Stage $3$, we modify the edge sampling probability in Stage $3$ to $1/(\log n)^3$ such that the sampled subgraph preserves the component-wise spectral gap (by tuning other parameters) and it has only $O(m / (\log n)^2)$ edges so that computing its connected components takes $O(m/\log n)$ work. 
However, sampling each edge in the original graph takes $\Theta(m)$ work since we need to sample from each of its edges, accumulating to $\Theta(m \log\log n)$ work over all phases. 
To overcome this, we build another random subgraph $H''$ with edge sampling probability $1/(\log n)^3$ before all phases, then perform the same contractions on $H''$ according to the contractions performed on $H$ in Stage $2$. 
We further isolate the randomness used in generating $H''$ to the randomness used in other parts of the algorithm such that $H''$ is an unbiased random subgraph to preserve the component-wise spectral gap. 
\newpage

\begin{framed}
\noindent \textbf{Connectivity with unknown $\lambda$:}
\begin{itemize}
    \item Parameter $b = (\log n)^{20}$.
    \item Build random subgraph $H', H''$ by sampling each edge in the original graph $G$ w.p. $1/(\log n)^3$.
    \item For \emph{phase} $i$ from $0$ to $O(\log\log n)$: 
    \begin{itemize}
        \item Run $O(\log b)$ steps of Stage $1$ to decrease the number of vertices to $O(n/b)$.
        \item To build $H$, preserve all edges in $G$ adjacent to vertices with degree at most $b$ in $H'$, then add all edges of $H'$ into $H$.
        \item Increase the minimum degree of $H$ to $b$ by contraction.
        \item Perform contractions on $H''$ as performed on $H$ in the previous step.
        \item Try to compute the connected components of $H''$ in $O(\log b)$ time.
        \item If all connected components of $H''$ are computed then return.
        \item Else update $b \leftarrow b^{3/2}$.
    \end{itemize}
\end{itemize}
\end{framed}

\paragraph{Corner case.}
To give a correct algorithm, we need to deal with the case that the algorithm terminates (computes all connected components of $H''$) in a phase under a \emph{wrong} assumption on the component-wise spectral gap. 
In this case, we show that after computing the connected components of the sampled subgraph in Stage $3$, the number of inter-component edges in the original graph is $O(n / \log n)$ by the sampling lemma of \cite{karger1995randomized}. 
So running the \cite{liu2020connected} only on these inter-component edges at the end finishes the computation in $O(m)$ work.

\subsection{Boosting to Linear Work with High Probability} \label{subsec:overview_boosting}

The procedure that increases the minimum degree in Stage $2$ originated from \cite{liu2020connected}, which only succeeds with probability $1 - 1/\poly(\log n)$ as we assume $m = n \poly(\log n)$. 
To achieve high probability for this part, observe that the skeleton graph $H$ we built for increasing minimum degree has bounded number of edges: each edge between high-degree vertices is sampled w.p. $1/(\log n)^3$, and there are at most $n/(\log n)^{20} \cdot (\log n)^{10}$ edges adjacent to low-degree vertices, since there are at most $n/(\log n)^{20}$ vertices after Stage $1$ and the low-degree vertex has at most $(\log n)^{10}$ adjacent edges. 
So the number of edges in $H$ is at most $m / (\log n)^2$. 
Running $\log n$ parallel instances of Stage $2$ on graph $H$ increases the minimum degree with high probability as desired.

Stage $3$ succeeds w.p. $1 - 1/\poly(\log n)$ because the algorithm of \cite{liu2020connected} is called. 
Note that the random subgraph $H''$ to solve connectivity on is equivalent to the original graph $G$ with contractions performed as on $H$ in Stage $2$ then edge-sampling w.p. $1/(\log n)^3$. 
Therefore, the number of edges in $H''$ is $m/(\log n)^2$ with high probability, and we run $\log n$ parallel instances of the algorithm of \cite{liu2020connected} on $H''$ to achieve a high-probability result.

\newpage

%\section*{\Large Detailed Description of Our Algorithms and Proofs}
%In the rest of the paper, we present the formal description of our algorithms and detailed proofs.
{\hypersetup{linkcolor=black, linktoc=all}\tableofcontents}

\newpage

\section{Stage 1: Contracting the Graph to $n/\poly(\log n)$ Vertices} \label{sec:stage1}

In this section, we describe the Stage $1$ algorithm in details. 
Recall that the goal is to contract and shrink
the graph to $n/\poly\log(n)$ vertices in $O(\log \log n)$ parallel time and linear work. 
As mentioned earlier in \S\ref{sec:overview}, existing schemes such as Reif's algorithm~\cite{reif1984optimal},
Gazit's algorithm~\cite{DBLP:journals/siamcomp/Gazit91}, and the algorithm of Halperin and Zwick~\cite{DBLP:journals/jcss/HalperinZ96} fail to achieve our purpose
since they either are suboptimal in total work or incur $\Omega((\log\log n)^2)$ parallel time.
We therefore devise a new algorithm for this purpose which can be of independent interest.

We begin with presenting a linear-work algorithm that reduces the number of vertices in the graph to $n/ \log \log n$ in  $O((\log\log\log n)^2)$ time. 
The main tool that we use is a constant-time and linear-work algorithm that reduces the number of vertices by a constant factor. 
Our constant-shrink algorithm is inspired by the matching MPC algorithm of \cite{DBLP:conf/focs/BehnezhadDELM19}. 
In each call, our algorithm finds a large matching between the vertices of the graph such that at least a constant fraction of the vertices are matched. We can then contract these matching edges to reduce the number of vertices by a constant factor in each call.

After that, we devise an algorithm that reduces the number of vertices to $n / \log\log n$ in $O((\log\log\log n)^2)$ time using one of the ideas from Gazit's algorithm (but boosted to $O(m)$ work w.h.p.): separate the contracted graph into a dense part (for vertices with high degree) and a sparse part. 
The separation is done by repeatedly deleting edges with constant probability in each round, such that the vertex with low degree has higher probability to be isolated (filtered out as the sparse part of the graph). 
We only repeat the above filter process for $O(\log\log\log n)$ times and prove that each application of filter takes $O(\log\log\log n)$ time, giving $O((\log\log\log n)^2)$ time for pre-processing.
Then, we prove that the number of vertices with high degree in the original graph is reduced to $n / (\log n)^{21}$, and the sparse part has $O(n / \log\log n)$ edges with high probability. 
This is done by showing that each high-degree vertex contributes no edge to the sparse part with high probability, and the total contribution of edges from low-degree vertices is $O(n / \log\log n)$ with high probability by the fact that the number of vertices is at most $n / \log\log n$ and with the use of concentration bounds for self-bounding functions. 
So we can contract the sparse part to at most $n / (\log n)^{21}$ vertices by $O(\log\log n)$ applications of the constant-shrink algorithm in $O(n)$ work since the number of edges is bounded from above by $O(n / \log n)$ with high probability.

Our algorithm will make use of a classical PRAM building block called \emph{approximate compaction}~\cite{DBLP:conf/focs/Goodrich91} which we define below.

\begin{definition}\label{def:apx_compaction}
    Given a length-$n$ array $A$ which contains $k$ \emph{distinguished} elements, \emph{approximate compaction} is to map all the distinguished elements in $A$ one-to-one to an array of length $2k$.
\end{definition}

\begin{lemma}[\cite{DBLP:conf/focs/Goodrich91}] \label{lem:apx_compaction}
    There is an ARBITRARY CRCW PRAM algorithm for approximate compaction that runs in $O(\log^* n)$ time and $O(n)$ work w.p. $1 - 1/2^{n^{1/25}}$.
    %Moreover, if using $n \log n$ processors, the algorithm runs in $O(1)$ time.\footnote{\cite{DBLP:conf/focs/Goodrich91} actually compacts the $k$ distinguished elements of $A$ to an array of length $(1 + \epsilon) k$ using $O(n / \log^* n)$ processors with probability $1 - 1 / c^{n^{1/25}}$ for any constants $\epsilon > 0$ and $c > 1$, but Lemma~\ref{lem_apx_compaction} suffices for our use. The second part of the lemma is a straightforward corollary (cf. \S{2.4.1} in \cite{DBLP:conf/focs/Goodrich91}).}
\end{lemma}

Our blueprint is the following.
First, in \S\ref{subsec:alg_constant}, we devise an algorithm that contracts and shrinks the graph by a constant factor (in the number of vertices) in constant time and linear work with high probability.
Next, building on top of the constant-shrink algorithm, in \S\ref{subsec:alg_loglog} we adapt ideas from Gazit~\cite{DBLP:journals/siamcomp/Gazit91} and devise an algorithm that reduces
the number of vertices by a $\log \log n$-factor in $o(\log \log n)$ time and linear work.
Finally, in \S\ref{subsec:alg_log}, we show how to achieve a $\poly(\log n)$-shrink algorithm in $O(\log\log n)$ parallel time and linear work.
In this section, we focus on bounding {\it expected} total work. Some extra work is required to achieve high probability guarantees which is deferred to \S\ref{sec:boosting}.

\subsection{A Constant-Shrink Algorithm} \label{subsec:alg_constant}

\Cliff{This algorithm has a long description cause it does 'dirty' work for other algorithms. This part might go to appendix. It makes our descriptions of other algorithms much cleaner.}
 
The following algorithm \textsc{Matching}$(E)$ shrinks the number of roots adjacent to $E$ by a constant fraction w.h.p. in $O(1)$ time and $O(|E|)$ work.
The algorithm either finds a large enough matching in $E$ so that each edge in the matching reduces $1$ root (by updating the parent of an end to another end), or a constant fraction of the roots are already non-roots.

The algorithm is inspired by the constant-shrink MPC algorithm from \cite{DBLP:conf/focs/BehnezhadDELM19}, which uses the $O(1)$-time sorting and prefix sum computing on an MPC. 
We could also modify Gazit's random-mate algorithm \cite{DBLP:journals/siamcomp/Gazit91} to get a constant-shrink algorithm. %, whose original version runs in $O(|V| + |E(V)|)$ work, instead of $O(|E|)$ work.
We include our version here tuned for our applications and for completeness.

\begin{framed}
\noindent \textsc{Matching}$(E)$:
\begin{enumerate}
    \item For each edge $(u, v) \in E$: if $u$ or $v$ is not a root or $u = v$ then delete $(u, v)$ from $E$. \label{alg:matching:s1}
    \item For each edge $e \in E$: orient it from the large end to the small end. This creates a digraph $D$. \label{alg:matching:s2}
    \item For each vertex $v \in V(E)$: if $v$ has more than $1$ outgoing arcs then keep an arbitrary one and delete the others. \label{alg:matching:s3}
    \item For each singleton $v$ created in Step~\ref{alg:matching:s3}: consider $D$ before Step~\ref{alg:matching:s3} and choose an arbitrary arc $(u, v)$ in $D$ then $v.p = u$. \label{alg:matching:s4}
    \item For each root $v \in V(E)$: if $v$ has more than $1$ incoming arcs then remove all outgoing arcs of $v$. \label{alg:matching:s5}
    \item For each root $v \in V(E)$: if $v$ has more than $1$ incoming arcs then update the parents of all vertices with arcs to $v$ as $v$ and delete all these vertices from $D$. \label{alg:matching:s6}
    \item For each arc $e$ in $D$: w.p. $1/2$ delete $e$ from $D$. \label{alg:matching:s7}
    \item For each arc $(u, v)$ in $D$: if $(u, v)$ is isolated (not sharing ends with other arcs), then $v.p = u$. \label{alg:matching:s8}
    \item For each vertex $v \in V(E)$: $v.p = v.p.p$. \label{alg:matching:s9}
    %\item Return the edge set $E$ after Step~\ref{alg:matching:s1}.  \label{alg:matching:s10}
\end{enumerate}
\end{framed}

At the beginning of the algorithm we assume each edge in $E$ as well as each vertex in $V(E)$ has a corresponding indexed processor.
\begin{lemma} \label{lem:constant_shrink_work_time}
    \textsc{Matching}$(E)$ uses $O(|E|)$ processors and runs in $O(1)$ time on an ARBITRARY CRCW PRAM.
\end{lemma}
\begin{proof} 
    In Step~\ref{alg:matching:s1}, each processor for each edge $(u, v) \in E$ queries the vertex processors $u$ and $v$ and if either $u \ne u.p$ or $v \ne v.p$ or $u = v$ then the edge processor notifies itself to be not in $E$. 
    In Step~\ref{alg:matching:s2}, each (processor corresponding to) edge $(u, v) \in E$ writes an arc $(u, v)$ into its private memory if $u > v$ or writes $(v, w)$ if otherwise. All arcs in the private memory consist of digraph $D$. Each arc processor also stores a copy of the arc in its private memory for future use in Step~\ref{alg:matching:s4}. 
    
    In Step~\ref{alg:matching:s3}, each arc $(v, u)$ writes itself to the private memory of (the processor corresponding to) $v$, so an arbitrary arc adjacent to $v$ wins the writing. Next, each arc $(v, u)$ checks the arc written to $v$, and if the arc does not equal to $(v, u)$ (so parallel edges in $E$ at the beginning of the algorithm does not effect the execution) then notifies itself to be not in $D$ but the arc processor is still active for the next step.
    
    In Step~\ref{alg:matching:s4}, each arc $(v, u)$ in $D$ notifies vertices $v$ and $u$ to be a non-singleton. Then, each active arc processor corresponding to $(u, v)$ checks whether $v$ is a singleton, and if so, updates $v.p$ to $u$.
    
    In Step~\ref{alg:matching:s5}, each arc $(u, v)$ checks whether $v$ is a root, and if so, writes $(u, v)$ to the private memory of $v$. Next, each arc $(u, v)$ checks whether the arc written to $v$ equals to $(u, v)$, and if not, then mark the vertex $v$. Next, each arc $(v, u)$ checks if $v$ is marked, and if so, notifies the arc itself to be not in $D$. \AF{If v is not root we should do that, since we are only deleting outgoing edges} \Cliff{If v is not a root then there is no adjacent edge by Step 1 and Step 4.}
    
    In Step~\ref{alg:matching:s6}, the algorithm uses the same operations in Step~\ref{alg:matching:s5} to mark all vertices with more than $1$ incoming arcs. Next, each arc $(u, v)$ in $D$ checks whether the vertex $v$ is marked, and if so, updates $u.p$ to $v$ then marks $u$ as deleted. Next, for each arc $(u, v)$ in $D$, if $u$ or $v$ is marked as deleted then notifies $(u, v)$ to be not in $D$.
    
    In Step~\ref{alg:matching:s7}, each arc in $D$ notifies itself to be not in $D$ with probability $1/2$.
    
    %In Step~\ref{alg:matching:s8}, each arc $(u, v)$ in $D$ writes $u$ to the private memory of $v$ and writes $v$ to the private memory of $u$. Next, each arc $(u, v)$ in $D$ checks whether the vertex written in the private memory of $v$ equals to $u$, and if not then updates $v.p$ to $u$ \AF{is this correct?}, else it checks whether the vertex written in the private memory of $u$ equals to $v$, and if not then updates $v.p$ to $u$. 
    
    In Step~\ref{alg:matching:s8}, each arc $(u, v)$ in $D$ writes $u$ to the private memory of $v$ and writes $v$ to the private memory of $u$. Next, each arc $(u, v)$ in $D$ checks whether the vertex written in the private memory of $v$ equals to $u$, and if not then mark vertices $v$ and $u$, then it checks whether the vertex written in the private memory of $u$ equals to $v$, and if not then mark vertices $v$ and $u$. Next, for each arc $(u, v)$ in $D$, if both $u$ and $v$ are not marked then updates $v.p$ to $u$. %\Cliff{Check this.}
    
    In Step~\ref{alg:matching:s9}, each edge (not arc) processor corresponding to $(u ,v)$ updates $u.p$ to $u.p.p$ by querying the private memory of $u$ and $u.p$ and writing to $u$, then do the same for $v$.
    
    Since each step takes $O(|E|)$ work and $O(1)$ time, the lemma follows.
\end{proof}

\begin{lemma} \label{lem:constant_shrink_reduce}
    Given graph $G(V(E), E)$, if its subgraph $G'$ induced on roots has $n'$ vertices and each component of $G'$ has at least $2$ vertices, then \textsc{Matching}$(E)$ reduces the number of roots to $0.999 n'$ w.p. $1 - 2^{-n'/10000}$.
\end{lemma}
\begin{proof}
    After Step~\ref{alg:matching:s2}, there are at least $n'/2$ arcs in $D$ since there is no singleton. 
    Suppose there are at most $n'/4$ arcs after Step~\ref{alg:matching:s3}. 
    Then Step~\ref{alg:matching:s3} creates at least $n'/2$ singletons. Because starting from an all-singleton graph, $1$ arc eliminates at most $2$ singletons. 
    Consider a singleton $v$ created in Step~\ref{alg:matching:s3}. In digraph $D$ before Step~\ref{alg:matching:s3}, $v$ must have no outgoing arcs, otherwise it cannot be a singleton in Step~\ref{alg:matching:s3}, so $v$ must have an incoming arc $(u, v)$ and will be a non-root in Step~\ref{alg:matching:s4}.
    Therefore, in Step~\ref{alg:matching:s4} we reduce at least $n'/2$ roots as desired.
    In the following, we assume the number of remaining arcs after Step~\ref{alg:matching:s3} (and thus Step~\ref{alg:matching:s4}) is at least $n'/4$. 

    After Step~\ref{alg:matching:s5}, we show that there are at least $n'/8$ arcs. After Step~\ref{alg:matching:s4}, put $1$ coin on each arc, then there are at most $2$ coins during the entire Step~\ref{alg:matching:s5}, because we can enumerate all vertices from small to large and distribute the at most $2$ coins on the only outgoing arc (guaranteed by Step~\ref{alg:matching:s3}) of that vertex to the at least $2$ incoming arcs. So at most half the arcs are removed in Step~\ref{alg:matching:s5}. 
    %\AF{ We put one coin on very single arc after Step~\ref{alg:matching:s4}, and we show that for every arc removal during the Step~\ref{alg:matching:s5} we can spend at least 2 coins. This immediately implies that at most half  of arcs get removed during the Step~\ref{alg:matching:s5}, thus we have at least $\frac{n'}{8}$ arcs by the end of this step. The argument is as follows. For a root $v$ with at least two incoming arcs we remove all of its outgoing arcs during the Step~\ref{alg:matching:s5}. It is however guaranteed by the Step~\ref{alg:matching:s3} that $v$ has at most one outgoing arc. Thus, we can spend all of the coins of on the incoming arcs to remove this outgoing arc.} 
    %\Cliff{I think you need to fix an order on vertices - how can you guarantee that every vertex still has at least two incoming arcs after some vertex removed its outgoing arc in Step 5? You can edit this if you think it's simpler, after fixing the order from small to large.}

    Vertex with no outgoing arc is called a \emph{head}. After Step~\ref{alg:matching:s5}, all the non-heads have in-degree at most $1$.
    In Step~\ref{alg:matching:s6}, if a head $v$ has $r > 1$ incoming arcs then we reduce $r$ roots and delete at most $2r$ arcs because each vertex $w$ with arc $(w, v)$ has at most $1$ outgoing arc and at most $1$ incoming arc.
    If at least $n'/16$ arcs are deleted in Step~\ref{alg:matching:s6}, then at least $n'/32$ roots are reduced as desired. 
    
    Now suppose there are at least $n'/16$ arcs after Step~\ref{alg:matching:s6}, which form a collection of disjoint directed paths since each vertex has in-degree and out-degree at most $1$ and there is no cycle dues to the orientations. 
    For each such path, starting from marking the outgoing arc on the largest vertex, we mark an arc in every $3$ arcs. There are at least $n'/48$ marked arcs. 
    After Step~\ref{alg:matching:s7}, each marked arc is isolated w.p. at least $1/8$ independently. 
    By a Chernoff bound, at least $n'/1000$ arcs are isolated w.p. at least $1 - 2^{-n'/10000}$, which reduces at least $n'/1000$ roots.
\end{proof}

\begin{lemma} \label{lem:constant_shrink_correctness}
    For any root at the beginning of \textsc{Matching}$(E)$, it must be a root or a child of a root at the end of \textsc{Matching}$(E)$.
\end{lemma}
\begin{proof}
    If a root $u$ at the beginning of Step~\ref{alg:matching:s4} becomes a parent of another vertex $v$ then $u$ cannot update its parent during Step~\ref{alg:matching:s4}, because $u$ cannot be a singleton in Step~\ref{alg:matching:s3} since $u$ has an outgoing arc $(u, v)$ before Step~\ref{alg:matching:s3}. 
    Therefore, any original root must be a root or a child of a root after Step~\ref{alg:matching:s4}.
    
    If a root $v$ at the beginning of Step~\ref{alg:matching:s6} becomes a parent of another vertex then $v$ cannot update its parent during Step~\ref{alg:matching:s6}, because $v$ has no outgoing arc dues to Step~\ref{alg:matching:s5}. Moreover, $v$ cannot update its parent in Step~\ref{alg:matching:s8} since $v$ has no incoming arc dues to the deletion of all its children from $D$ in Step~\ref{alg:matching:s6}. Therefore, any vertex who obtains a child in Step~\ref{alg:matching:s6} must be a root after Step~\ref{alg:matching:s8}.
    Similarly, any vertex $u$ who obtains a child $v$ in Step~\ref{alg:matching:s8} must be a root after Step~\ref{alg:matching:s8} because $(u, v)$ is isolated so there is no incoming arc on $u$.
    
    As a result, an original root who becomes a non-root in Step~\ref{alg:matching:s6} or Step~\ref{alg:matching:s8} must be a child of a root after Step~\ref{alg:matching:s8}. 
    Additionally, if an original root becomes a child of $u$ in Step~\ref{alg:matching:s4}, then $u$ must be a root by the first paragraph, and by the previous sentence, $u$ has to be a root or a child of a root after Step~\ref{alg:matching:s8}. Therefore, any original root must be a root or a child of a root or a grandchild of a root after Step~\ref{alg:matching:s8}. 
    Finally, Step~\ref{alg:matching:s9} makes each of them a root or a child of a root.
\end{proof}

\subsection{A $\log\log n$-Shrink Algorithm} \label{subsec:alg_loglog}

\Cliff{I will add a simplified version for the first 10 pages.}

In this section, we give a $\log\log n$-shrink algorithm. 
The algorithm runs in $o(\log\log n)$ time and $O(m) + \overline{O}(n)$ work and contracts the input graph to at most $n/\log\log n$ vertices w.h.p.
%\Cliff{Specify parameters later.}

The four subroutines \textsc{Alter}, \textsc{Filter}, \textsc{Reverse}, and \textsc{Extract} are listed below. 
We briefly explain the parameter and function of each algorithm, whose detailed proofs are presented after that in this section.

\begin{framed}
\noindent \textsc{Alter}$(E)$:
\begin{enumerate}
    \item For each edge $(u, v) \in E$: replace $(u, v)$ by $(u.p, v.p)$ in $E$. \label{alg:alter:s1}
    \item Remove loops from $E$.
\end{enumerate}
\end{framed}
\textsc{Alter}$(E)$ is a standard PRAM algorithm that takes an edge set $E$ as input and alters $E$ by moving each edge adjacent to a vertex to its parent and deleting loops on each vertex.

\begin{framed}
\noindent \textsc{Filter}$(E, k)$:
\begin{enumerate}
    \item For \emph{round} $j$ from $0$ to $k$: \textsc{Matching}$(E)$, \textsc{Alter}$(E)$, delete each edge from $E$ w.p. $10^{-4}$.\label{alg:filter:s1}
    \item For \emph{iteration} $j$ from $k$ to $0$: if a vertex $v$ updates $v.p$ in round $j$ then $v.p = v.p.p$.\label{alg:filter:s2}
    \item Return $V(E)$. \label{alg:filter:s3}
    %\item Return vertex set $V' = \{v \mid \text{$v$ is a root and }(u, v) \in E\}$.\label{alg:filter:s3}
\end{enumerate}
\end{framed}
\textsc{Filter}$(E, k)$ takes an edge set $E$ and an integer $k$ as input. In each of the $k$ rounds, it deletes a constant fraction of the edges to reduce the work and tries to contract vertices using the remaining edges. After $k$ rounds, the algorithm tries to flatten the tree in the labeled digraph in $k$ iterations. All roots that still adjacent with an edge are returned. The algorithm can be seen as a \emph{filter} as the vertices with higher degree are more likely to be filtered out to be contracted or returned. 
The parameter $E$ in \textsc{Filter}$(E, k)$ is local, i.e., the edge deletion and \textsc{Alter}$(E)$ in Step~\ref{alg:filter:s1} do not change the edge set originally passed to \textsc{Filter}.

\begin{framed}
\noindent \textsc{Reverse}$(V', E)$:
\begin{enumerate}
    \item For each non-root vertex $v \in V'$: $v.p.p = v$, $v.p = v.p.p$.\label{alg:reverse:s1}
    \item For each vertex $v$ in the labeled digraph: $v.p = v.p.p$.\label{alg:reverse:s2} 
    \item \textsc{Alter}$(E)$. \label{alg:reverse:s3}
\end{enumerate}
\end{framed}
\textsc{Reverse}$(V')$ is a helper subroutine that takes the set of high-degree vertices as input and reverses the root-child relation on a pair of vertices if the root has low degree and the child has high degree.

\begin{framed}
\noindent \textsc{Extract}$(E, k)$:
\begin{enumerate}
    \item Initialize vertex set $V' \coloneqq \emptyset$ and edge set $E' \coloneqq \{e \in E \mid e \text{ is a non-loop} \}$.\label{alg:extract:s1}
    \item For \emph{round} $i$ from $0$ to $k$: $V' = V' + \textsc{Filter}(E', k)$, \textsc{Alter}$(E')$, $E' = E' - \{(u, v) \mid u, v \in V'\}$.\label{alg:extract:s2}
    \item For \emph{iteration} $i$ from $k$ to $0$: if a vertex $v$ updates $v.p$ in round $i$ then $v.p = v.p.p$. \label{alg:extract:s3}
    \item \textsc{Reverse}$(V', E)$.\label{alg:extract:s4}
\end{enumerate}
\end{framed}
\textsc{Extract}$(E, k)$ takes an edge set $E$ and an integer $k$ as input. 
By calling $\textsc{Filter}(E', k)$ in each of the $k$ rounds, it tries to \emph{extract} more vertices with high degree from the subgraph induced on edges (the edge set $E'$) adjacent to those low-degree vertices identified (did not pass the filter) in the previous round. 
After $k$ rounds, the algorithm tries to flatten the trees in $k$ iterations. 
Finally, the algorithms calls \textsc{Reverse}$(V')$ to make the root of any tree contains a high-degree vertex (from $V'$) a high-degree vertex and \textsc{Alter}$(E)$.

For implementation, we note that algorithms \textsc{Filter}$(E, k)$ and \textsc{Matching}$(E)$ use pass-by-value, i.e., the algorithms copy the edge set $E$ before all steps and only operate on the copy. 
This does not affect the asymptotic running time and total work but simplifies the frameworks and proofs. 
All other algorithms in this paper use pass-by-reference, which operate on the variables passed to them.

\subsubsection{Correctness} \label{subsubsec:alg_loglog_correctness}

\begin{lemma} \label{lem:filter_root_child}
    In the execution of \textsc{Filter}$(E, k)$, for any vertex $v$ and integer $j \in [0, k]$, if $v.p$ is updated in round $j$, then $v.p$ is a root or a child of a root at the beginning of iteration $j$, and must be a root at the end of that iteration.
\end{lemma}
\begin{proof}
    The proof is by an induction on $j$ from $k$ to $0$. 
    In the base case, $v.p$ is updated to $u$ in round $k$. Observe that $u$ must be a root over all rounds, otherwise Step~\ref{alg:matching:s1} within \textsc{Matching}$(E)$ removes its adjacent edges and $u$ can never become a parent of other vertices afterwards (this observation will be used frequently). Since iteration $k$ is the first iteration, we have that $v.p = u$ is a root at the beginning of iteration $k$. Next, we prove that $v.p$ must be a root at the end of iteration $k$, completing the proof for the base case. Since $u$ did not update its parent in round $k$, we have that $u.p$ in not updated in iteration $k$. Since $u = u.p$ is a root during the entire iteration $k$, $v.p$ will be $v.p.p = u.p$ at the end of iteration $k$, and the base case holds.
    
    In round $j < k$, if $v.p$ is updated to $u$, then $u$ is a root and cannot update its parent in round $j$, which means $u.p$ either 
    (\romannumeral1) is updated in round $j' > j$ or
    (\romannumeral2) remains to be $u$ after round $k$. 
    It is easy to prove for case (\romannumeral2) because $v.p = u = u.p$ holds from the end of round $j$ to the end of round $k$, so it continues to hold from the beginning of iteration $k$ to the beginning of iteration $j$. Meanwhile, at the end of iteration $j$, $v.p$ becomes $v.p.p = u.p = u$ which is a root. So we proved case (\romannumeral2).
    
    It remains to prove for case (\romannumeral1). By the induction hypothesis, we know that $u.p$ is a root or a child of a root at the beginning of iteration $j'$, and is a root at the end of iteration $j'$. Note that for $u.p$ to be updated in round $j' > j$, $u$ must be a root before round $j'$ which means $u.p$ is not updated before round $j'$, so from the end of iteration $j'$ to the end of iteration $j$, $u.p$ is not updated. Therefore, $u.p$ is a root during the entire iteration $j$. Since $v.p$ is updated in round $j$, it cannot be updated in any round $j' > j$, which means it cannot be updated before iteration $j$. We obtain that at the beginning of iteration $j$, $v.p = u$ is a root or a child of $u.p$, which is a root, and this proves the first part of the statement. At the end of iteration $j$, $v.p$ is updated to $v.p.p = u.p$ which is a root, completing the induction.
\end{proof}

\begin{lemma} \label{lem:filter_tree_height}
    For any integer $h \ge 0$, if the labeled digraph at the beginning of an execution of \textsc{Filter}$(E, k)$ is a set of trees with height at most $h$, then it is a set of trees with height at most $h+1$ at the end of that execution.
\end{lemma}
\begin{proof}
    It is easy to see from \textsc{Matching} that the algorithm never creates cycles except loops on roots. For any vertex $v$ that is a root before round $0$, if $v$ is the only vertex in the tree, then the contribution to tree height from the tree path between $v$ and a leaf (might be $v$ itself) in this tree is $0$.
    Otherwise, let $w \ne v$ be a leaf descendant of $v$ in the tree before round $0$, then the distance between $v$ and $w$ in the tree is at most $h$ after round $k$, because for any vertex $u$ on the tree path from $v$ to $w$ (excluding $v$), there is no edge on $u$, so $u$ cannot change its parent and $w$ remains to be a leaf.
    
    If $v.p$ is updated in round $j$, then by Lemma~\ref{lem:filter_root_child}, $v.p = u$ is a root at the end of iteration $j$. Since $v$ must be a root before round $j$, $v.p = u$ continues to hold till the end of iteration $0$. Observe that any iteration does not make a root into a non-root (the only place that possibly turns a root into a non-root is within \textsc{Matching} when updating parents), therefore $u$ remains to be a root after iteration $0$, which means $v$ has distance at most $1$ from its root.
    Summing up, the distance from the root of any tree to any leaf in this tree is at most $h+1$ at the end of that execution.
\end{proof}

\begin{lemma} \label{lem:extract_round_root_adj}
    In the execution of \textsc{Extract}$(E, k)$, for any integer $i \in [0, k]$, if at the beginning of round $i$, both ends of any edge in $E'$ are  roots, then this still holds at the end of round $i$.
\end{lemma}
\begin{proof}
    Firstly, for any vertex $v$ with no adjacent edges in $E'$ at the beginning of round $i$, the only way to obtain an adjacent edge in this round is to have a vertex $w$ with an adjacent edge such that the \textsc{Filter}$(E', k)$ updates $w.p$ to $v$ and the \textsc{Alter} moves this edge to $v$. 
    So $w$ is a descendant of $v$ before Step~\ref{alg:filter:s2} of \textsc{Filter}$(E, k)$ called in round $i$.
    At this point, consider the tree path from $w$ to $v$ right before iteration $0$. Let $u$ be the child of $v$ on this path. We claim that $u$ must be a child of $v$ at the beginning of this call to \textsc{Filter}$(E, k)$. 
    Otherwise, $u$ becomes a child of $v$ during \textsc{Matching}$(E)$ where $E$ is the $E'$ passed to \textsc{Alter} in round $i$ within \textsc{Extract}$(E, k)$ then passed to the \textsc{Matching} within \textsc{Filter}, which leads to a contradiction since $v$ does not have an adjacent edge in $E'$ and \textsc{Matching} only updates the parent to a vertex with adjacent edges. 
    As a result, $u$ must be a non-root before any application of \textsc{Matching} and thus is a non-root before calling \textsc{Filter}$(E', k)$, which means $u$ cannot have an adjacent edge in $E'$ by the precondition of this lemma. 
    Continue this process by considering the child of $u$ on this path and so on, we get that $w$ is a non-root at the beginning of round $i$ of \textsc{Extract}$(E, k)$ and thus cannot have adjacent edges in $E'$, contradiction. 
    Therefore, a vertex with no adjacent edges in $E'$ at the beginning of round $i$ cannot have adjacent edges at the end of round $i$.
    
    Secondly, for any vertex $v$ with an adjacent edges in $E'$, $v$ must be a root before calling \textsc{Filter}$(E', k)$ in round $i$ by the precondition. If $v$ is still a root after calling \textsc{Filter}$(E', k)$, then the lemma holds. 
    Otherwise, $v$ must become a non-root in the call to \textsc{Matching}$(E)$ within \textsc{Filter}$(E', k)$ in round $i$ of \textsc{Extract}$(E, k)$. 
    Suppose that $v$ becomes a non-root in round $j$ of \textsc{Filter}$(E', k)$, then it remains to be a non-root afterwards. 
    By Lemma~\ref{lem:filter_root_child}, $v.p = u$ must be a root at the end of iteration $j$. 
    Since $u$ cannot be a non-root before iteration $j$, $u.p$ has not been updated over all rounds in \textsc{Filter}$(E', k)$. Therefore $u.p$ (which is $u$) is not updated during all iterations within \textsc{Filter}$(E', k)$. Recall that a root remains to be a root during all iterations, so $v.p = u$ is a root at the end of \textsc{Filter}$(E', k)$. Because $v.p$ is a root and $v$ is a non-root, it must be that $v.p \ne v$, then the \textsc{Alter}$(E')$ in round $i$ of \textsc{Extract}$(E, k)$ moves all adjacent edges on $v$ to $v.p$.
    
    It remains to show that no new edge is moved to $v$ in that same \textsc{Alter}$(E')$. For this situation to happen, there are two possibilities: 
    (\romannumeral1) a descendant $w$ of $v$ has an adjacent edge in $E'$ before calling \textsc{Filter}$(E', k)$ and $w.p = v$ after that in round $i$, which contradicts with the precondition since $w$ is a non-root;
    (\romannumeral2) a root $w$ with an adjacent edge before calling \textsc{Filter}$(E', k)$ updates $w.p$ to $v$ within \textsc{Filter}$(E', k)$. 
    Let $j$ be the round that $w.p$ is updated such that $w.p = v$ at the end of iteration $j$. 
    By Lemma~\ref{lem:filter_root_child}, $v$ is a root at the end of iteration $j$.
    Since a root never becomes a non-root without executing \textsc{Matching}, we have that $v$ is still a root at the end of that call to \textsc{Filter}$(E', k)$, contradicting with the assumption we made in the third sentence in the previous paragraph. 
    As a result, all adjacent edges on $v$ is moved to another vertex and there is no new edge moved from other vertices to $v$, so $v$ has no adjacent edges at the end of round $i$.
\end{proof}

\begin{lemma} \label{lem:extract_contraction_alg}
    If at the beginning of an execution of \textsc{Extract}$(E, k)$, both ends of any edge in $E$ are roots and all trees in the labeled digraph have height $0$, then at the end of that execution both ends of any edge in $E$ are still roots and all trees are flat.
\end{lemma}
\begin{proof}
    %By an induction on round $i$ from $0$ to $k$ in \textsc{Extract}$(E, k)$ and Lemma~\ref{lem:filter_tree_height}, all trees have height at most $k$ at the end of Step~\ref{alg:extract:s2} in \textsc{Extract}$(E, k)$ because \textsc{Filter}$(E, k)$ is called for $k$ times. 
    %In the following, we show by an induction on iteration $i$ from $k$ to $0$ (Step~\ref{alg:extract:s3} in \textsc{Extract}$(E, k)$) that any vertex $v$ which updated $v.p$ in round $i$ must be a root or a child of a root at the end 
    \Cliff{You can expand this part if you want.}
    Using the same proof for Lemma~\ref{lem:filter_root_child} (replacing \textsc{Matching}$(E)$ in that proof with \textsc{Filter}$(E, k)$), we can show that for any vertex $v$ and integer $i \in [0, k]$, if $v.p$ is updated in round $i$ of \textsc{Extract}$(E, k)$, then $v.p$ is a root or a child of a root at the beginning of iteration $i$, and must be a root at the end of that iteration. 
    This observation implies an analog of Lemma~\ref{lem:filter_tree_height} for \textsc{Extract}$(E, k)$ using the same technique in the proof of Lemma~\ref{lem:filter_tree_height} (by replacing \textsc{Matching}$(E)$ in that proof with \textsc{Filter}$(E, k)$)). That is, we can show that for any integer $h \ge 0$, if the labeled digraph at the beginning of \textsc{Extract}$(E, k)$ is a set of trees with height at most $h$, then it is a set of trees with height at most $h+1$ at the end of iteration $k$ in Step~\ref{alg:extract:s3} of \textsc{Extract}$(E, k)$, which means all tree are flat at this point.

    By an induction on round $i$ from $0$ to $k$ and Lemma~\ref{lem:extract_round_root_adj}, all edges are adjacent to roots after round $k$, which continues to hold after Step~\ref{alg:extract:s3}. 
    By the previous paragraph, all trees are flat. 
    In the \textsc{Reverse}$(V', E)$ in Step~\ref{alg:extract:s4}, consider a tree rooted at $u$ who has a child $v \in V'$ and $v$ wins the concurrent writing $v.p.p = v$ (arbitrarily) in Step~\ref{alg:reverse:s1}. 
    After this, we have $u.p = v$; after the following $v.p = v.p.p$ in Step~\ref{alg:reverse:s1}, we have $v.p = v$, giving a tree rooted at $v$ with height at most $2$.
    Then, Step~\ref{alg:reverse:s2} makes this tree flat. 
    Therefore, Step~\ref{alg:reverse:s3} (\textsc{Alter}$(E)$) makes all edges adjacent to roots.
\end{proof}

\subsubsection{Running Time and Work} \label{subsubsec:alg_loglog_time_work}

The goal of our main algorithm \textsc{Extract} is to reduce the number of roots (in the labeled digraph) to $n /\log\log n$, so we assume the number of roots is at least $n/\log n$ during the execution, since otherwise the algorithm succeeds earlier. 
Additionally, we assume the number of edges is at least $n/(\log n)^2$ during the execution, otherwise we achieve $O(m + n)$ work even more easily (see details below).
%cause otherwise running $\log n$ instances of LTZ'20 (\Cliff{will add theorem to be cited}) computes the connected components in $O(\log d + \log\log n)$ time and $O(m + n)$ work w.p. $1 - n^{-2}$. 

\begin{lemma} \label{lem:filter_time_work}
    \textsc{Filter}$(E, k)$ can be implemented on an ARBITRARY CRCW PRAM in $O(k \log^* n)$ time and $O(|E|)$ work w.p. $1 - k n^{-7}$.
\end{lemma}
\begin{proof}
    For all $j \in [0, k]$, let $E_j$ be the edge set $E$ at the beginning of round $j$. 
    At the beginning of \textsc{Filter}$(E, k)$, we assume all edges of $E$ are stored in an index array of length at most $2|E|$, which can be guaranteed by approximate compaction (Lemma~\ref{lem:apx_compaction}) before the call to \textsc{Filter}$(E, k)$ whenever we need. 
    Inductively, assuming at the beginning of each round $j$, all edges of $E_j$ are stored in an indexed array of length at most $2|E_j|$. 
    At the end of each round $j$, the algorithm uses approximate compaction with $2|E_j|$ processors to compact the edges in $E_j$ that are not deleted into an array of length at most $2|E_{j+1}|$ because the remaining edges are exactly $E_{j+1}$. The approximate compaction takes at most $O(\log^* m) \le O(\log^* n)$ time and $O(|E_j|)$ work w.p. $1 - 1/2^{(n/(\log n)^2)^{1/25}} \ge 1 - n^{-9}$ by Lemma~\ref{lem:apx_compaction} and the assumption at the beginning of \S{\ref{subsubsec:alg_loglog_time_work}}. 
    By a Chernoff bound, we have $|E_{j+1}| \le 0.99999 |E_j|$ w.p. $1 - n^{-9}$. 
    By a union bound over $k$ rounds, w.p. $1 - k n^{-8}$, it must be that $|E_j| \le 0.99999^j |E_0|$ and round $j$ takes $O(\log^* n)$ time and $O(|E_j|)$ work for all $j \in [0, k]$, giving $O(|E_0|)$ total work and $O(k \log^* n)$ running time for Step~\ref{alg:filter:s1}.
    
    To implement Step~\ref{alg:filter:s2}, we modify Step~\ref{alg:filter:s1} as follows. 
    In each round $j$, in the call to \textsc{Matching}$(E)$, whenever an edge $(u, v)$ updates the parent of $u$ to $v$, the edge marks itself with vertex $u$ (there is a unique edge processor for each edge dues to indexing); after the call to \textsc{Matching}$(E)$, each marked edge $e$ writes the vertex marked on $e$ to a cell in the public RAM with index $\textsf{id}(e) + \textsf{offset}_j$, where $\textsf{id}(e)$ is the index of $e$ in the compacted array, and $\textsf{offset}_{j'}$ is defined as $0$ if $j'=0$ and if $j' \in [k]$,
    \begin{equation} \label{eq:offset}
        \textsf{offset}_{j'+1} = \textsf{offset}_{j'} + 2 |E_0| \cdot 0.99999^{j'} . 
    \end{equation}
    We claim that w.p. $1 - k n^{-8}$, each vertex written by the marked edges over all $k$ rounds is in a unique cell. 
    Within a round $j$, each edge id is unique dues to approximate compaction. 
    It remains to prove that for all $j$, for each $e \in E_j$, $\textsf{id}(e) < \textsf{offset}_{j+1} - \textsf{offset}_j$, which holds w.p. $1 - k n^{-8}$ by the Chernoff bound in the previous paragraph, Equation~(\ref{eq:offset}), and the fact that all edges of $E_j$ are stored in an index array of length at most $2|E_j|$. 
    The work added to Step~\ref{alg:filter:s1} from this part is $O(|E_0|)$ w.p. $1 - k n^{-8}$ and Step~\ref{alg:filter:s1} still runs in $O(k \log^* n)$ time. 
    
    Now we are ready to prove Step~\ref{alg:filter:s2}. 
    For each $j \in [0, k]$, in iteration $j$, the algorithm uses $(\textsf{offset}_{j+1} - \textsf{offset}_j)$ indexed processors for each of the cell in the public RAM indexed from $\textsf{offset}_j$ to $\textsf{offset}_{j+1}$; if the cell contains a vertex $v$, then the processor corresponding to $v$ executes $v.p = v.p.p$. Since only the vertex with an adjacent edge can possibly update its parent, our implementation is correct. The running time for this iteration is $O(1)$. 
    The work in iteration $j$ is $O(\textsf{offset}_{j+1} - \textsf{offset}_j) = O(|E_0| \cdot 0.99999^j)$ by Equation~(\ref{eq:offset}). 
    So Step~\ref{alg:filter:s2} runs in $O(k)$ time and $O(|E_0|)$ work over all $k$ iterations w.p. $1 - k n^{-8}$.
    
    In Step~\ref{alg:filter:s3}, each edge $(u, v) \in E$ notifies the vertices $u$ and $v$ to be in $V'$. This takes $O(1)$ time and $O(|E_0|)$ work. 
    By a union bound, w.p. $1 - k n^{-7}$, \textsc{Filter}$(E, k)$ runs in $O(k \log^* n)$ time and $O(|E_0|)$ work.
\end{proof}

Before proving bounds on the running time and work on \textsc{Extract}, we need some helper lemmas below.

\begin{definition} \label{def:extract_vi_ei}
    In the execution of \textsc{Extract}$(E, k)$, for all $i \in [0, k]$, let $V_i$ be the vertex set returned by \textsc{Filter}$(E', k)$ in round $i$, let $E_i$ be the edge set $E'$ at the beginning of round $i$, and let $V'_i$ be $V(E_i)$.
\end{definition}
Assume $|V'_i| \ge n / \log n$ for all $i \in [0, k]$, since otherwise we will get the desired shrink and work bound even more easily (see details in the proofs for Lemma~\ref{lem:extract_time_work_helper_V} and Lemma~\ref{lem:extract_time_work}).
    
\begin{lemma} \label{lem:extract_time_work_helper_Vprime}
    $|V'_i| \le 0.999^i n$ for all $i \in [0, k]$ w.p. $1 - k n^{-9}$.
\end{lemma}
\begin{proof}
    We prove that $|V'_i| \le 0.999^i n$ for all $i \in [0, k]$ w.p. $1 - k n^{-9}$ by an induction on $i$.
    The base case of $i = 0$ is straightforward as the total number of vertices is at most $n$. 
    Consider the call to \textsc{Filter}$(E', k)$ in round $i$ of \textsc{Extract}$(E, k)$. 
    In round $0$ of this \textsc{Filter}$(E', k)$, the algorithms executes \textsc{Matching}$(E)$ where $E = E'$. 
    By Step~\ref{alg:extract:s1} of \textsc{Extract}$(E, k)$ or the \textsc{Alter}$(E')$ in each round of \textsc{Extract}$(E, k)$, there is no loop in $E'$. 
    So by Lemma~\ref{lem:extract_round_root_adj}, each vertex in $V'_i$ is a root and has a non-loop edge in $E'$ to another root. 
    Applying Lemma~\ref{lem:constant_shrink_reduce}, the number of roots adjacent to $E'$ at the end of this call to \textsc{Matching}$(E)$ is at most $0.999 |V'_i|$ w.p. $1 - 2^{-|V'_i| / 1000} \ge 1 - n^{-9}$. 
    %After this, the edge deletion in round $0$ of \textsc{Filter}$(E', k)$ only decreases the number of vertices adjacent to $E'$. 
    %In the next round (round $1$) of \textsc{Filter}$(E', k)$, in Step~\ref{alg:matching:s1} of \textsc{Matching}$(E)$, all edges in $E'$ adjacent to a non-root is deleted, so all vertices adjacent to $E'$ are roots, whose number is at most $0.999 |V'_i|$ w.p. $1 - n^{-9}$. 
    %The subsequent calls to \textsc{Matching}$(E)$ and the edge deletions in later rounds of \textsc{Filter}$(E', k)$ can only decrease the number of vertices adjacent to $E'$. 
    After this call to \textsc{Filter}$(E', k)$ in round $i$ of \textsc{Extract}$(E, k)$, the \textsc{Alter}$(E')$ and the following edge deletion can only decrease the number of vertices adjacent to $E'$. 
    By Lemma~\ref{lem:extract_round_root_adj}, all vertices adjacent to $E'$ must be roots, which means the number of vertices adjacent to $E'$ at the end of round $i$ is at most $0.999 |V'_i|$ w.p. $1 - n^{-9}$, giving $|V'_{i+1}| \le 0.999 |V'_i|$ w.p. $1 - n^{-9}$. 
    The lemma follows by a union bound.
\end{proof}

\begin{lemma} \label{lem:filter_root_adj}
    In the execution of \textsc{Filter}$(E, k)$, if both ends of any edge in $E$ are roots at the beginning of the algorithm, then both ends of any edge in $E$ are roots at the end of each round $j \in [0, k]$.
\end{lemma}
\begin{proof}
    Assume all edges in $E$ are adjacent to roots at the beginning of round $j$. 
    By Lemma~\ref{lem:constant_shrink_correctness}, each root is a root or a child of a root after the \textsc{Matching}$(E)$ in round $j$, so the following \textsc{Alter}$(E)$ moves edges in $E$ to the roots.
\end{proof}
Using Lemma~\ref{lem:extract_contraction_alg} and the fact that all edges are adjacent on roots initially, all edges are adjacent to roots at the end of each round of \textsc{Filter}$(E, k)$ for all its executions. 
Note that Lemma~\ref{lem:filter_root_adj} implies that Step~\ref{alg:matching:s1} of \textsc{Matching}$(E)$ does not need to delete edges adjacent to non-roots, but we keep it for simpler proofs in many other places.

\begin{lemma} \label{lem:extract_time_work_helper_V}
    $|V_i| \le 0.999^{k + i} n$ for all $i \in [0, k]$ w.p. $1 - k^2 n^{-8}$.
\end{lemma}
\begin{proof}
    For any $i \in [0, k]$, we have that $|V'_i|$, the number of vertices adjacent to $E'$ at the beginning of round $i$ of \textsc{Extract}$(E, k)$, is at most $0.999^i n$ w.p. $1 - k n^{-9}$ by Lemma~\ref{lem:extract_time_work_helper_Vprime}. In the following, We shall condition on this event happening, and apply a union bound at the end.
    
    Consider the call to \textsc{Filter}$(E', k)$ in round $i$ of \textsc{Extract}$(E, k)$. 
    In \textsc{Filter}$(E, k)$ (where $E = E'$ and we will use $E$ for simplicity in analyzing \textsc{Filter} locally), by an induction on $j$, we assume there are at most $0.999^{i+j} n$ roots adjacent to $E$ at the beginning of round $j$ w.p. at least $1 - j n^{-9}$. 
    Consider the call to \textsc{Matching}$(E)$ in round $j$.
    After Step~\ref{alg:matching:s1} of \textsc{Matching}$(E)$, let $E''$ be the edge set $E$ locally in \textsc{Matching}.
    There are at most $0.999^i n$ vertices adjacent to $E''$ by the first paragraph. Moreover, each vertex adjacent to $E''$ must be a root and has an edge in $E''$ to another root. 
    So by Lemma~\ref{lem:constant_shrink_reduce}, the number of roots adjacent to $E''$ at the end of \textsc{Matching}$(E)$ is at most $0.999^{i + j + 1} n$ w.p. $1 - (j+1) n^{-9}$.
    
    We claim that after the \textsc{Alter}$(E)$ following \textsc{Matching}$(E)$, the number of roots adjacent to $E$ is at most $0.999^{i + j + 1} n$ w.p. $1 - (j+1) n^{-9}$. 
    We prove the claim by showing that adding back the edges in $E \backslash E''$ then running \textsc{Alter}$(E)$ does not increase the number of adjacent roots. 
    
    Firstly, for any edge $e \in E \backslash E''$, it cannot be adjacent to a non-root before \textsc{Matching}$(E)$ (Lemma~\ref{lem:filter_root_adj}), so we only consider the case that $e$ is a loop. If $e$ is a loop before \textsc{Matching}$(E)$, then it is still a loop after \textsc{Matching}$(E)$ and Step~\ref{alg:alter:s1} of \textsc{Alter}$(E)$, which gets removed and thus does not increase the number of adjacent roots.
    Secondly, for any edge in $E''$, \textsc{Alter}$(E)$ cannot move it to a root that did not have an adjacent edge before \textsc{Matching}$(E)$, which proves the claim since the number of adjacent roots does not increase. 
    We show this by proving for any root $u$ with an adjacent edge on its child $v$ before \textsc{Alter}$(E)$, there is already an adjacent edge on $u$. 
    By Lemma~\ref{lem:filter_root_adj}, $v$ must be a root before \textsc{Matching}$(E)$, so $v$ becomes a child of $u$ during \textsc{Matching}$(E)$, which requires $u$ to have an adjacent edge or the child $w$ of $u$ to have an adjacent edge (if $v$ becomes a child of $u$ in Step~\ref{alg:matching:s9} of \textsc{Matching}$(E)$). For the latter case, $w$ must be a root before \textsc{Matching}$(E)$ by Lemma~\ref{lem:filter_root_adj} and $w$ becomes a child of $u$ during \textsc{Matching}$(E)$ before Step~\ref{alg:matching:s9}, which requires $u$ to have an adjacent edge. 
    So the claim in the previous paragraph holds. 
    
    The edge deletion following \textsc{Alter}$(E)$ in round $j$ can only decrease the number of roots adjacent to $E$. 
    So we complete the induction and the number of roots adjacent to $E$ at the end of round $j$ is at most $0.999^{i + j + 1} n$ w.p. $1 - (j+1) n^{-9}$ for all $j \in [0, k]$. 
    
    Finally, note from Lemma~\ref{lem:filter_root_adj} that all vertices adjacent to $E$ are roots at this point.
    As a result, the number of vertices adjacent to $E$ at the end of this call to \textsc{Filter}$(E, k)$ (which is returned as $V(E) = V_i$ by Definition~\ref{def:extract_vi_ei}) is at most $0.999^{k + i + 1} n$ w.p. $1 - (k+1) n^{-9}$. 
    By a union bound over all $i \in [0, k]$ and the first paragraph of this proof, the lemma follows.
\end{proof}

\begin{lemma} \label{lem:extract_time_work_helper_E}
     $|E_0| \le m$ and $\E\left[|E_i|\right] \le 10^4 n \cdot 0.9995^i + k^2 n^{-6}$ for all $i \in [k]$.
\end{lemma}
\begin{proof}
    \Cliff{I wrote many words in this proof to convince us.}
    By Definition~\ref{def:extract_vi_ei}, $E_i$ is the set of edges adjacent to $V'_i$ at the beginning of round $i$. 
    Since $m$ is the total number of edges in the input graph and we never add edges (but only delete or alter edges), it must be $|E_i| \le m$ for all $i \in [0, k]$. This simple observation will be used later in the proof.

    Consider round $i$ of \textsc{Extract}$(E, k)$. During the execution of the call to \textsc{Filter}$(E', k)$, if a vertex $v \in V'_i$ at the beginning of this call still has an adjacent edge $(v, u)$ in $E$ at the end of the call, then $v$ is returned by \textsc{Filter}$(E', k)$ and added into vertex set $V'$. 
    By Lemma~\ref{lem:filter_root_adj}, both $v$ and $u$ are roots, so edge $(v, u)$ is still in $E'$ after the \textsc{Alter}$(E')$ in round $i$, and edge $(v, u)$ is deleted from $E'$ at the end of round $i$ within \textsc{Extract}$(E, k)$. 
    Therefore, for an edge $(v, u)$ to be in $E_{i+1}$, there must be a round $j \in [0, k]$ such that in round $j$ of \textsc{Filter}$(E', k)$ (called in round $i$ of \textsc{Extract}$(E, k)$), all edge in $E$ adjacent to $v$ or $u$ are deleted. 

Now consider a vertex $v \in V'_i$, and let $E_{i,v}$ be the set of edges in $E_{i}$ that are adjacent to  $v$. Specifically, $E_{i,v} = \{ e \in E_i : v \in e \}$. We claim that for an edge $e=(v,u)$ to contribute to $E_{i+1}$ either all edges in $E_{i,v}$ or all edges in $E_{i,u}$ must be deleted by the end of execution of \textsc{Filter}$(E', k)$. %For the sake of contradiction suppose that %We claim that if any of edges in $E_{i,v}$ are not deleted during $k+1$ rounds of \textsc{Filter}$(E', k)$, then none of these edges will contribute to $E_{i+1}$.

Consider an edge $e=(v,u)$ and suppose that by the end of execution \textsc{Filter}$(E', k)$, at least one edge $(v,u') \in E_{i,v}$ is not deleted, and also at least one edge $(v',u) \in E_{i,u}$ is not deleted. In this case we show that $e$ cannot contribute to $E_{i+1}$. Considering the edge $(v,u')$, then in the round 0 of \textsc{Filter}$(E', k)$, the algorithm replaces $(v,u')$ with $(v.p,u'.p)$ in \textsc{Alter}$(E)$. It then replaces this new edge in the next rounds in a similar way. Suppose that $(v,u')$ is replaced by $(v_r, u'_r)$ by the end of Step \ref{alg:filter:s1} of \textsc{Filter}$(E', k)$. Then there should be a tree path $(v_0, v_1, \cdots, v_k)$ by the end of Step \ref{alg:filter:s1} of \textsc{Filter} such that $v_0=v$, $v_k=v_r$, and $v_i.p=v_{i+1}$ for each $0 \le i<k$. This implies that by the end of \ref{alg:filter:s2} of \textsc{Filter}, the algorithm sets $v.p=v_r$. Thus, when we call \textsc{Alter}$(E')$ in \textsc{Extract}$(E,k)$, the edge $(v,u)$ will be moved to $v_r$.  Similarly, if the edge $(v',u)$ gets replaced by $(v'_r,u_r)$ by the end of the \textsc{Filter}, then \textsc{Alter}$(E')$ also moves $(v,u)$ to $u_r$. This means that $(v,u)$ will be replaced by $(v_r,u_r)$.

Remember that we have assumed both $(v,u')$ and $(v',u)$ does not get deleted by \textsc{Filter}. This means that both $v_r$ and $u_r$ have an adjacent edge by the end of the \textsc{Filter}$(E', k)$, and they will be returned by \textsc{Filter}$(E', k)$. Thus, \textsc{Extract}$(E,k)$ adds both $v_r$ and $u_r$ to the vertex set $V'$. Since the edge $(v,u)$ will be replaced by $(v_r,u_r)$, and both $v_r$ and $u_r$ are in $V'$, the \textsc{Extract}$(E,k)$ removes this edge from $E'$. Thus,  $(v,u)$ does not contribute to $E_{i+1}$. This shows that for an edge $e=(v,u)$ to contribute to $E_{i+1}$ either all edges originally in $E_{i,v}$ or all edges originally in $E_{i,u}$ must be deleted by the end of execution of \textsc{Filter}$(E', k)$.

We say that $(v,u)$ contributes to $E_{i+1}$ through $v$ if all edges in $E_{i,v}$ are deleted by the end of execution of \textsc{Filter}$(E', k)$. It then follows that, for every edge in $E_{i+1}$, its contribution should be through at least one of its endpoints. Now for a vertex $v$ we give an upper bound on the expected number of edges that contribute to $E_{i+1}$ through $v$. For every edge $e \in E_{i,v}$, the probability that it gets deleted by \textsc{Filter} is  $1 - (1 - 10^{-4})^{k+1}$. Note that the set of edges that contribute through $v$ is a subset of $E_{i,v}$, and this contribution happens only if all the the edges in $E_{i,v}$ are deleted. Thus, the expected number of edges that contribute to $E_{i+1}$ through $v$ is at most
 
     \begin{equation*}
        |E_{i,v}| \left(1 - (1 - 10^{-4})^{k+1}\right)^{|E_{i,v}|} \le \frac{-1/e}{\ln\left(1 - (1 - 10^{-4})^{k+1}\right)} \le \frac{1/e}{(1 - 10^{-4})^{k+1}} \le (1 - 10^{-4})^{-k} \le 1 ,
    \end{equation*}
    where we used $xa^x \le -1/(e \ln a)$ for all $x$ when $a \in (0, 1)$ in the first inequality, and $\ln(1-x) \le -x$ for all $x \in (0,1)$ in the second inequality.

    Thus, the expected contribution to $E_{i+1}$ through  every vertex $v \in V'_i$ in round $j$ is at most $1$. Thus, we have $\E[|E_{i+1}|]=\E[|V'_i|]$.
    By Lemma~\ref{lem:extract_time_work_helper_Vprime}, $|V'_i| \le 0.999^{i} n$ w.p. $1 - k n^{-9}$. If we condition on $|V'_i| \le 0.999^{i}$, then we have $|E_{i+1}| \le 0.999^{i} n$.

    If the number of vertices in $V'_i$ exceeds $0.999^{i} n$, which happens w.p. at most $k n^{-9}$, then $\E[|E_{i+1}|] \le m$ by the first paragraph of this proof. The expected contribution from this part is at most $k n^{-9} m \le k n^{-7}$, so we obtain $\E\left[|E_{i+1}|\right] \le n \cdot 0.999^i + k n^{-7}$ which is better than the desired bound of lemma.
\end{proof}

\begin{lemma} \label{lem:extract_time_work}
    \textsc{Extract}$(E, k)$ can be implemented on an ARBITRARY CRCW PRAM in $O(k^2 \log^* n)$ time and $O(m) + \overline{O}(n + k^3 n^{-6})$ work w.p. $1 - k^3 n^{-6}$.
\end{lemma}
\begin{proof}
    \Cliff{You can expand this part if you want.}
    In \textsc{Extract}$(E, k)$, Step~\ref{alg:extract:s1} and Step~\ref{alg:extract:s4} (\textsc{Reverse}$(V')$) take $O(1)$ time and $O(m+n)$ work. 
    We shall use the same technique for implementing \textsc{Filter}$(E, k)$ in the proof of Lemma~\ref{lem:filter_time_work} to implement Step~\ref{alg:extract:s2} and Step~\ref{alg:extract:s3}.
    
    Round $0$ of Step~\ref{alg:extract:s2} takes $O(k \log^* n)$ time and $O(m)$ work w.p. $1 - k n^{-7}$ time by $|E_0| \le m$ (Lemma~\ref{lem:extract_time_work_helper_E}) and Lemma~\ref{lem:filter_time_work}. 
    Similar to the proof of Lemma~\ref{lem:filter_time_work}, at the end of each round $i \in [0, k]$, the algorithm runs approximate compaction to compact all edges that are still in $E'$ (which becomes the set $E_{i+1}$) to an indexed array of length at most $2|E_{i+1}|$ in $O(\log^* n)$ time w.p. $1 - n^{-9}$. Since we assume all edges of $E'$ are stored in an indexed array of length at most $2|E_i|$, the work for compaction is $O(m)$ for round $0$ and $\overline{O}(0.9995^i n + k^2 n^{-6})$ for round $i \in [k]$ w.p. $1 - n^{-9}$ by Lemma~\ref{lem:extract_time_work_helper_E}. 
    In each round $i \in [k]$, since all edges of $E_i$ are stored in an array of size $2 |E_i| = \overline{O}(0.9995^i n + k^2 n^{-6})$, the algorithm uses the same number of indexed edge processors (a processor corresponding to an empty cell in the compacted array corresponds to a dummy empty edge) to run $V_i = \textsc{Filter}(E_i, k)$, which takes $O(k \log^* n)$ time and $\overline{O}(0.9995^i n + k^2 n^{-6})$ work w.p. $1 - k n^{-7}$ by Lemma~\ref{lem:filter_time_work}. 
    After that, the $V' = V' + V_i$ takes $O(|V_i|)$ work and $O(1)$ time cause each vertex in $V_i$ notifies itself to be in $V'$. By Lemma~\ref{lem:extract_time_work_helper_V}, the work here is $O(0.999^{k+i} n)$ w.p. $1 - k^2 n^{-8}$. 
    Next, \textsc{Alter}$(E')$ uses the same $\overline{O}(0.9995^i n + k^2 n^{-6})$ processors and takes $O(1)$ time. 
    The following edge deletion on $E'$ also uses the same processors in $O(1)$ time since the algorithm let each edge in $E_i$ check whether both of its ends are in $V'$. 
    Summing up and by a union bound, w.p. at least $1 - k n^{-7} - k n^{-9} - k^2 n^{-8} \ge 1 - k^2 n^{-6}$, Step~\ref{alg:extract:s2} of round $i$ takes $O(k \log^* n)$ time, $O(m)$ work for round $0$, and $\overline{O}(0.9995^i n + k^2 n^{-6}) + O(0.999^{k+i} n) = \overline{O}(0.9995^i n + k^2 n^{-6})$ 
    work for round $i \in [k]$. 
    Summing over all $k+1$ rounds, we have that w.p. $1 - k^3 n^{-6}$, Step~\ref{alg:extract:s2} takes $O(k^2 \log^* n)$ time and the total work is
    \begin{equation*}
        O(m) + \sum_{i=1}^k \overline{O}(0.9995^i n + k^2 n^{-6}) = O(m) + \overline{O}(n + k^3 n^{-6}) .
    \end{equation*}
    
    To implement Step~\ref{alg:extract:s3}, we follow the idea used in the proof of Lemma~\ref{lem:filter_time_work}. Define $\textsf{offset}_0$ as $0$ and $\textsf{offset}_{i+1} = \textsf{offset}_i + |A_i|$ for $i \in [k]$, where $|A_i|$ is the length of the compacted array that stores $E_i$ which is at most $2|E_i|$ and is known to the algorithm at the beginning of round $i+1$. Then each vertex that updated its parent in round $i$ is written in a unique cell in the public RAM because the number of vertices that possibly updated parents is at most the number of edges in $E_i$. 
    As discussed before, the extra running time added to each round $i$ in Step~\ref{alg:extract:s2} and the running time of each iteration $i$ in Step~\ref{alg:extract:s3} is $O(1)$; 
    the total extra work added to Step~\ref{alg:extract:s2} and the total work of Step~\ref{alg:extract:s3} is
    \begin{align*}
        & \sum_{i = 0}^k O(\textsf{offset}_{i+1} - \textsf{offset}_i) = \sum_{i=0}^k O(|A_i|) = \sum_{i=0}^k O(|E_i|) \\
        =& O(m) + \sum_{i=1}^k \overline{O}(0.9995^i n + k^2 n^{-6}) = O(m) + \overline{O}(n + k^3 n^{-6}) ,
    \end{align*}
    where the second line is by Lemma~\ref{lem:extract_time_work_helper_E}. The lemma follows immediately.
\end{proof}

\subsubsection{Proof for $\log\log n$-Shrink} \label{subsubsec:alg_loglog_shrink}

In this section we prove that \textsc{Extract}$(E, k)$ contracts the input graph $G = (V, E)$ to $n / \log\log n$ vertices in expected linear work w.h.p. when $k = \Theta(\log\log\log n)$ (see parameter setup $k$ in \S{\ref{subsec:alg_log}}).

By Lemma~\ref{lem:extract_contraction_alg}, \textsc{Extract}$(E, k)$ is a contraction algorithm. 
If all vertices from a component in the original graph contract to the same vertices, i.e., all these vertices are in the same flat tree in the labeled digraph, then the computation for this component is done and the algorithm can ignore all these vertices in the later execution. 
If a root $v$ still has an adjacent edge to another root, then the computation for the component of $v$ is not done yet, and we call $v$ \emph{active}. 
The \emph{current graph} is defined to be the graph induced on active roots (using the current edges altered by the algorithm).

At the end of the algorithm, the vertices in the current graph (active roots in the labeled digraph) consist of three parts: 
(\romannumeral1) the active roots in $V'$, 
(\romannumeral2) the active roots in $V(E')$, and
(\romannumeral3) the active roots not in $V' \cup V(E')$. 
In the following, we bound each of them.

\begin{lemma} \label{lem:extract_bound_Vprime}
    At the end of \textsc{Extract}$(E, k)$, the number of active roots in $V'$ is $0.999^k 1000n$ w.p. $1 - k^2 n^{-8}$.
\end{lemma}
\begin{proof}
    By Definition~\ref{def:extract_vi_ei}, $V' = \bigcup_{i = 0}^k V_i$. 
    By Lemma~\ref{lem:extract_time_work_helper_V}, $|V_i| \le 0.999^{k+i} n$ for all $i \in [0, k]$ w.p. $1 - k^2 n^{-8}$. The lemma follows immediately.
\end{proof}

\begin{lemma} \label{lem:extract_bound_VEprime}
    At the end of \textsc{Extract}$(E, k)$, the number of active roots in $V(E')$ is $0.999^k n$ w.p. $1 - k n^{-8}$.
\end{lemma}
\begin{proof}
    Consider running \textsc{Extract}$(E, k)$ before its Step~\ref{alg:extract:s4}.
    By Definition~\ref{def:extract_vi_ei}, $V(E') = V(E_{k+1}) = V'_{k+1}$ where $E_{k+1}$ is the edge set $E'$ at the end of round $k$. 
    The proof of Lemma~\ref{lem:extract_time_work_helper_Vprime} implies that $|V'_{k+1}| \le 0.999^{k+1} n$ w.p. $1 - (k+1) n^{-9}$.
    It is easy to verify that if a flat tree does not contain a vertex from $V'$ then it is unchanged in the \textsc{Reverse}$(V', E)$ in Step~\ref{alg:extract:s4};
    otherwise the flat tree becomes a flat tree rooted at a vertex $u \in V'$ and all adjacent edges of this tree are adjacent to $u$ after \textsc{Reverse}$(V', E)$. 
    Therefore the number of active roots in $V(E')$ cannot increase, giving the lemma.
\end{proof}
Note that the above two lemmas bound the number of active roots by the number of vertices, which might not be tight but suffices for our proof.

\begin{lemma} \label{lem:extract_bound_no}
    At the end of \textsc{Extract}$(E, k)$, there is no active root in $V \backslash (V' \cup V(E'))$, where $V$ is the vertex set in the input graph.
\end{lemma}
\begin{proof}
    Consider running \textsc{Extract}$(E, k)$ before its Step~\ref{alg:extract:s4}.
    Observe that any edge in the graph at this point is either in $E'$ or has both ends in $V'$. 
    So, for any root $u \in V \backslash (V' \cup V(E'))$, $u$ must have no adjacent edges. 
    If the flat tree rooted at $u$ does not contain a vertex $v \in V'$, then this tree must contain all vertices in the component of $u$ in the input graph. This is because we never delete an edge (in Step~\ref{alg:extract:s2}) unless it is between vertices in $V'$. 
    Therefore, the computation for this component is done and $u$ is not an active root.
    
    Otherwise, if the flat tree rooted at $u$ does contain a vertex from $V'$, then 
    in the \textsc{Reverse}$(V', E)$ (Step~\ref{alg:extract:s4}), suppose $v \in V'$ and $v$ ((arbitrarily)) wins the concurrent writing $v.p.p = v$ in Step~\ref{alg:reverse:s1}. 
    After this, we have $u.p = v$; after the following $v.p = v.p.p$ in Step~\ref{alg:reverse:s1}, we have $v.p = v$, giving a tree rooted at $v$ with height at most $2$.
    Then, Step~\ref{alg:reverse:s2} makes this tree flat and $u$ is no longer a root.
\end{proof}

\begin{lemma} \label{lem:loglog_shrink_reduce}
    Let $k = 1000 \log\log\log n$. 
    \textsc{Extract}$(E, k)$ runs on an ARBITRARY CRCW PRAM in $O(m) + \overline{O}(n)$ work and $O((\log\log\log n)^2 \log^* n)$ time, and reduces the number of vertices in the current graph to $n / \log\log n$ w.p. $1 - n^{-4}$.
\end{lemma}
\begin{proof}
    By Lemma~\ref{lem:extract_time_work}, \textsc{Extract}$(E, k)$ runs in $O(m) + \overline{O}(n)$ work and $O((\log\log\log n)^2 \log^* n)$ time w.p. $1 - n^{-5}$. 
    By Lemma~\ref{lem:extract_bound_Vprime}, Lemma~\ref{lem:extract_bound_VEprime}, Lemma~\ref{lem:extract_bound_no}, and the discussion before them, the number of vertices in the current graph is at most 
    \begin{equation*}
        1001n \cdot 0.999^{1000 \log\log\log n} \le n / \log\log n
    \end{equation*}
    w.p. $1 - n^{-7}$. The lemma immediately follows from a union bound.
\end{proof}

\subsection{A $\poly(\log n)$-Shrink Algorithm} \label{subsec:alg_log} 

In this section, we give an algorithm that runs in $O(\log\log n)$ time and $O(m) + \overline{O}(n)$ work and contracts the graph to $n / \poly(\log n)$ vertices w.h.p.

\begin{framed}
\noindent \textsc{Reduce}$(V, E, k)$:
\begin{enumerate}
    \item \textsc{Extract}$(E, 1000 \log\log\log n)$.\label{alg:reduce:s1}
    \item $V' = \textsc{Filter}(E, k)$.\label{alg:reduce:s2}
    \item For each $v \in V$ do $v.p = v.p.p$. \textsc{Alter}$(E)$.\label{alg:reduce:s3}
    \item Edge set $E' = \{(u, v) \mid u \notin V' \text{ or } v \notin V'\}$.\label{alg:reduce:s4}
    \item For \emph{round} $i$ from $0$ to $k$: \textsc{Matching}$(E')$, for each $v \in V$ do $v.p = v.p.p$, \textsc{Alter}$(E')$.\label{alg:reduce:s5}
    \item \textsc{Reverse}$(V', E)$.\label{alg:reduce:s6}
\end{enumerate}
\end{framed}

\begin{lemma} \label{lem:alg:reduce_correctness}
    If at the beginning of an execution of \textsc{Reduce}$(V, E, k)$, both ends of any edge in $E$ are roots and all trees in the labeled digraph have height $0$, then at the end of that execution both ends of any edge in $E$ are still roots and all trees are flat.
\end{lemma}
\begin{proof}
    By Lemma~\ref{lem:extract_contraction_alg}, the statement holds after Step~\ref{alg:reduce:s1}. 
    By Lemma~\ref{lem:filter_tree_height}, after Step~\ref{alg:reduce:s2}, all trees have height at most $2$. 
    In Step~\ref{alg:reduce:s3}, all trees become flat and the \textsc{Alter}$(E)$ makes all edges adjacent to roots. 
    By an induction on round $i$ from $0$ to $k$ in Step~\ref{alg:reduce:s5} and  Lemma~\ref{lem:constant_shrink_correctness}, all trees have height at most $2$ after \textsc{Matching}$(E')$, and become flat at the end of that round and all edges are adjacent to roots. 
    Using the same proof for \textsc{Reverse} as in Lemma~\ref{lem:extract_contraction_alg}, the statements holds after Step~\ref{alg:reduce:s6}.
\end{proof}

Next, we prove that the algorithm runs in $O(m) + \overline{O}(n)$ work. 
We could not use the implementation for \textsc{Filter}$(E, k)$ in Lemma~\ref{lem:filter_time_work}, because here we want $O(k) = O(\log\log n)$ running time instead of $O(k \log^* n)$ time. 
Observe that the number of vertices in current graph is reduced, so we can afford more work proportional to the number of vertices in pursuit for an improved running time.

\begin{lemma} \label{lem:filter_time_work_naive}
    \textsc{Filter}$(E, k)$ can be implemented on an ARBITRARY CRCW PRAM in $O(k)$ time and $O(|E| + k \cdot |V(E)|)$ work.
\end{lemma}
\begin{proof}
    In Step~\ref{alg:filter:s1}, initially each edge knows itself to be in $E$. 
    In each round, the edge notifies itself to be not in $E$ w.p. $10^{-4}$, and only the edge corresponding to a non-deleted edge in $E$ performs work in \textsc{Matching}$(E)$, \textsc{Alter}$(E)$, and edge deletion. 
    Combining with Lemma~\ref{lem:constant_shrink_work_time}, round $j$ takes $O(1)$ time and $O(|E_j|)$ work where $E_j$ is the edge set $E$ at the beginning of round $j$.
    By a Chernoff bound, $\Pr[|E_j| \le 0.99999 |E_{j-1}|]$ w.p. $1 - n^{-9}$ for each $j \in [k]$. 
    By a union bound, $|E_j| \le 0.99999^j |E_0|$ w.p. $1 - n^{-8}$ for all $j \in [k]$. 
    As a result, Step~\ref{alg:filter:s1} takes $O(k)$ time and $O(|E|)$ work w.p. $1 - n^{-8}$. 
    
    To implement Step~\ref{alg:filter:s2}, each edge processor corresponding to edge $(u, v)$ in each round $j$ of Step~\ref{alg:filter:s1} also stores the ends $u$ (resp. $v$) into a cell indexed by vertex $u$ (resp. $v$) and $j$ in the public RAM\footnote{For example, the cell with index $j \cdot |A| + \textsf{id}(u)$ where $A$ is the array storing all vertices in $V(E)$ and $|A| \le 2 |V(E)|$ by approximate compaction (see the proof of Lemma~\ref{lem:reduce_time_work}) and $\textsf{id}(u)$ is the id of the processor corresponding to vertex $u$.}  
    if vertex $u$ (resp. $v$) updated its parent in round $j$, and this adds $O(|E|)$ total work to Step~\ref{alg:filter:s1}; in each iteration $j$ of Step~\ref{alg:filter:s2}, the algorithm uses all vertex processors to check all cells corresponding to round $j$ (by index range), and if the cell contains a vertex $v$ then executes $v.p = v.p.p$, so the iteration takes $O(|V(E)|)$ work and $O(1)$ time. 
    Since only vertex with an adjacent edge can update its parent, this implementation is correct, and Step~\ref{alg:filter:s2} takes $O(k)$ time and $O(k \cdot |V(E)|)$ work.
    
    In Step~\ref{alg:filter:s3}, each edge $(u, v) \in E$ notifies the vertices $u$ and $v$ to be in $V'$. This takes $O(1)$ time and $O(|E|)$ work, giving the lemma. 
\end{proof}

\begin{lemma} \label{lem:reduce_time_work}
    Let $k = 10^6 \log\log n$. 
    \textsc{Reduce}$(V, E, k)$ can be implemented on an ARBITRARY CRCW PRAM in $O(\log\log n)$ time and $O(m) + \overline{O}(n)$ work w.p. $1 - n^{-3}$.
\end{lemma}
\begin{proof}
    By Lemma~\ref{lem:loglog_shrink_reduce}, Step~\ref{alg:reduce:s1} takes $o(\log\log n)$ time and $O(m) + \overline{O}(n)$ work, and reduces the number of vertices in the current graph to $n / \log\log n$ w.p. $1 - n^{-4}$. In the following, we shall assume that this event holds and apply a union bound at the end. 
    
    After Step~\ref{alg:reduce:s1}, the algorithm runs approximate compaction in $O(\log^* n)$ time to rename each vertex to a unique id in $[2n / \log\log n]$ in $O(n)$ work w.p. $1 - 1/2^{n^{1/25}} \ge 1 - n^{-10}$ (Lemma~\ref{lem:apx_compaction}).
    
    By Lemma~\ref{lem:filter_time_work_naive}, Step~\ref{alg:reduce:s2} takes $O(\log\log n)$ time and $O(m + k \cdot n / \log\log n) = O(m+ n)$ work.
    
    Step~\ref{alg:reduce:s3} takes $O(1)$ time and $O(m + n)$ work.
    
    The edge set $E'$ defined in Step~\ref{alg:reduce:s4} is the same as the edge set $E'$ defined in round $0$ of \textsc{Extract}$(E, k)$, i.e., let $V' = \textsc{Filter}(E, k)$, then $E'$ is the set of edges such that at least one end of the edge is not in $V'$. 
    So, we can use the same technique used in the proof of Lemma~\ref{lem:extract_time_work_helper_E} to argue that the expected contribution to $|E'|$ from each vertex in $V(E)$ is at most 
    \begin{equation*}
        (1 - 10^{-4})^{-j} + 1 = 2,
    \end{equation*}
    where $j = 0$. 
    Therefore, $\E[|E'|] = O(n / \log\log n)$. 
    
    By Lemma~\ref{lem:constant_shrink_work_time}, the \textsc{Matching}$(E')$ and \textsc{Alter}$(E')$ in each round $i$ in Step~\ref{alg:reduce:s5} takes $O(1)$ time and $O(|E'|)$ work. 
    By the first paragraph, $|V| \le n / \log\log n$, and the $v.p = v.p.p$ takes $O(1)$ time. 
    Combining with the previous paragraph, Step~\ref{alg:reduce:s5} takes $O(k)$ time and $\overline{O}(k \cdot n / \log\log n) = \overline{O}(n)$ total work.
    
    Step~\ref{alg:reduce:s6} takes $O(1)$ time and $O(m + n)$ work. The lemma follows from a union bound.
\end{proof}

\begin{lemma} \label{lem:polylog_shrink_reduce}
     Let $k = 10^6 \log\log n$. 
     \textsc{Reduce}$(V, E, k)$ reduces the number of vertices in the current graph to $n / (\log n)^{1000}$ w.p. $1 - n^{-7}$.
\end{lemma}
\begin{proof}
    By Lemma~\ref{lem:alg:reduce_correctness}, the algorithm is a correct contraction algorithm. 
    Apply Definition~\ref{def:extract_vi_ei} and use the same technique in the proof of Lemma~\ref{lem:extract_time_work_helper_V}, we have that $V'$ defined in Step~\ref{alg:reduce:s2} has cardinality at most 
    \begin{equation*}
        0.999^{k + i} n = 0.999^{10^6 \log\log n} \le \frac{n}{(\log n)^{1400}}
    \end{equation*}
    w.p. $1 - n^{-8}$, where $i = 0$ as we only call \textsc{Filter}$(E, k)$ once.
    
    By the proof of Lemma~\ref{lem:alg:reduce_correctness} and Step~\ref{alg:reduce:s4}, all active roots not in $V'$ must have an adjacent edge in $E'$, so now we bound the number of active roots in $V(E')$. 
    (Recall that an active root is a vertex that is a root with adjacent edges to another root in the current graph.)
    A vertex $v$ is called \emph{alive} if $v$ is a root in the labeled digraph and has an adjacent edge in $E'$. 
    (Note that an active root might not be alive since $E'$ does not contain all the edges in the current graph.)
    By an induction on $i \in [0, k]$ in Step~\ref{alg:reduce:s5} and applying Lemma~\ref{lem:constant_shrink_reduce}, the number of alive vertices at the beginning of round $i$ is at most $0.999^i n$ w.p. $1 - i \cdot n^{-9}$. So, the number of alive vertices before Step~\ref{alg:reduce:s6} is at most $0.999^{10^6 \log\log n} \le n / (\log n)^{1400}$ w.p. $1 - n^{-8}$. 
    
    Consider a non-alive vertex $v$ in the subgraph $G'$ induced on $E'$ before Step~\ref{alg:reduce:s6}. 
    If the component of $v$ in $G'$ does not contain a vertex in $V'$, then the computation for this component is done as $G'$ contains all the edges between vertices in this component, so $v$ is not active since there is no edge on $v$ in the current graph by the fact that $v$ is the root of a flat tree and the execution of \textsc{Alter}$(E')$. 
    Otherwise, if the component of $v$ in $G'$ does contain a vertex in $V'$, then the \textsc{Reverse}$(V', E)$ in Step~\ref{alg:reduce:s6} makes $v$ a non-root and thus not an active root at the end of \textsc{Reduce}$(V, E, k)$ (and this root is counted through the first paragraph of this proof).
    
    As a result, the number of vertices in the current graph at the end of \textsc{Reduce}$(V, E, k)$ is at most $2n / (\log n)^{1400} \le n / (\log n)^{1000}$ w.p. $1 - n^{-7}$.
\end{proof}

\begin{lemma} \label{lem:polylog_shrink_reduce_all}
    Let $k = 10^6 \log\log n$. 
    \textsc{Reduce}$(V, E, k)$ runs on an ARBITRARY CRCW PRAM in $O(m) + \overline{O}(n)$ work and $O(\log\log n)$ time, and reduces the number of vertices in the current graph to $n / (\log n)^{1000}$ w.p. $1 - n^{-2}$.
\end{lemma}
\begin{proof}
    By Lemma~\ref{lem:reduce_time_work} and Lemma~\ref{lem:polylog_shrink_reduce}.
\end{proof}

We did not optimize the constants in this section. 
The setup of $k$ is a consequence of \S{\ref{sec:assumption}}.

\section{Stage 2: Increasing the Minimum Degree to $\poly(\log n)$} \label{sec:stage2}

After Stage $1$, by Lemma~\ref{lem:polylog_shrink_reduce_all} we assume the number of vertices is $n / b^{10}$ where $b = (\log n)^{100}$ is a parameter (we will gradually increase the value of $b$ in \S{\ref{sec:assumption}}). Let the current graph be $G = (V, E)$ and we have $|V| \le n/b^{10}$ and $|E| \le m$. 
In this section, we give a linear-work algorithm that increases $\deg(G)$, the minimum degree over all vertices in $V$, to at least $\text{poly}(b)$.

%The algorithm \textsc{Increase}$(V, E, b)$ will be analyzed in details in \S{\ref{subsec:increase}}.
Our algorithm, called \textsc{Increase}$(V, E, b)$, first builds a subgraph $H$ (called \emph{skeleton graph}) on the same vertex set $V$ as the input graph but with fewer edges ($\textsc{Build}(V, E, b)$), then it calls \textsc{Densify}$(H, b)$ to contract the skeleton graph such that some vertices obtains $\text{poly}(b)$ children, then in Step~\ref{alg:short_increase:s3} each vertex either 
(\romannumeral1) becomes a non-root, or
(\romannumeral2) becomes a root with $\text{poly}(b)$ children.  
At the end of the algorithm, any vertex in the current graph has degree at least $\text{poly}(b)$. 

\begin{framed}
\noindent \textsc{Increase}$(V, E, b)$:
\begin{enumerate}
    \item Skeleton graph $H = (V, E') = \textsc{Build}(V, E, b)$. \label{alg:short_increase:s1}
    \item $(V, E_{\text{close}}) = \textsc{Densify}(H, b)$. \label{alg:short_increase:s2}
    \item Use $E_{\text{close}}$ and the labeled digraph updated in Step~\ref{alg:short_increase:s2} to increase the degree of each root (vertex in the current graph) to be at least $\text{poly}(b)$ -- see \S{\ref{subsec:increase}} for a detailed description. \label{alg:short_increase:s3}
\end{enumerate}
\end{framed}

In this section, we state the algorithm's guarantees in terms of {\it good} rather than {\it high} probability, where good probability means with probability $1-1/\poly\log n$. Later in \S\ref{sec:boosting}, we will be able to boost the probability guarantees to high probability.
%The detailed description of Step~\ref{alg:short_increase:s3} will be presented in \S{\ref{subsec:increase}}.

\subsection{Building the Skeleton Graph} \label{subsec:build}

\begin{framed}
\noindent \textsc{Build}$(V, E, b)$:
\begin{enumerate}
    \item For each vertex $v \in V$: assign a block of $b^9$ processors, which will be used as an indexed hash table $\mathcal{H}(v)$ of size $b^9$. \label{alg:build:s1}
    \item Choose a random hash function $h: [n] \to [b^9]$. For each edge $(u, v) \in E$: write $u$ into the $h(u)$-th cell of $\mathcal{H}(v)$. \label{alg:build:s2}
    \item For each vertex $v \in V$: if the number of vertices in $\mathcal{H}(v)$ is more than $b^8$ then mark $v$ as \emph{high}, else mark $v$ as \emph{low}. \label{alg:build:s3}
    \item Let $E'$ be an empty edge set. For each edge $(u, v) \in E$: if $u$ or $v$ is low then add $(u, v)$ to $E'$, else add $(u, v)$ to $E'$ w.p. $b^{-1}$.  \label{alg:build:s4}
    \item Remove parallel edges and loops from $E'$. \label{alg:build:s5}
    \item Return graph $H = (V, E')$.  \label{alg:build:s6}
\end{enumerate}
\end{framed}

\begin{lemma} \label{lem:build_time_work}
    \textsc{Build}$(V, E, b)$ can be implemented on an ARBITRARY CRCW PRAM in $O(\log b)$ time and $O(m + n)$ work w.p. $1 - n^{-8}$.
\end{lemma}
\begin{proof}
    Recall that $|V| \le n / b^{10}$ and $b = (\log n)^{100}$.
    In Step~\ref{alg:build:s1}, using approximate compaction (Lemma~\ref{lem:apx_compaction}) we compact the vertices in $V$ to an index array then use their indices to assign a block of indexed $b^4$ processors, which takes $O(n)$ work and $O(\log^* n)$ time w.p. $1 - n^{-9}$.
    In Step~\ref{alg:build:s2}, each vertex $u$ computes $h(u)$ by choosing an index from $[b^9]$ uniformly at random and independently then stores $h(u)$ in the private memory of $u$. Next, each edge $(u, v)$ reads the value of $h(u)$ from $u$ then writes into $H(v)$. The total work is $O(m+n)$ and the running time is $O(1)$.
    In Step~\ref{alg:build:s3}, for each vertex $v$, use all the $b^9$ processors in its block to count the number of cells occupied by any vertex in $\mathcal{H}(v)$ (by a simple binary tree structure) in $O(\log(b^9)) = O(\log b)$ time, whose work is at most $|V| \cdot b^9 \cdot O(\log b) \le O(n)$. 
    In Step~\ref{alg:build:s4}, each edge decides whether itself is in $E'$ in $O(1)$ time, which takes $O(m)$ total work. 
    
    In Step~\ref{alg:build:s5}, loops can be removed in $O(m)$ work and $O(1)$ time. 
    To remove parallel edges, we use PRAM perfect hashing on (the indices of) all edges of $E'$ in $O(m)$ work and $O(\log^* n)$ time w.p. $1 - n^{-9}$ \cite{DBLP:conf/focs/GilMV91}.\footnote{Given a set $S \subseteq [m]$, the \emph{perfect hashing} problem is to find an one-to-one function $h: S \to [O(|S|)]$ such that $h$ can be represented in $O(|S|)$ space and $h(x)$ can be evaluated in $O(1)$ time for any $x \in S$. We use perfect hashing to distinguish from other hashings used in this paper which can cause collisions.} 
    The lemma follows from a union bound.
\end{proof}

\begin{lemma} \label{lem:build:high_vertex}
    Let $G = (V, E)$. 
    W.p. $1 - n^{-9}$, for any vertex $v \in V$ that is marked as high in \textsc{Build}$(V, E, b)$, $|N_G(v)| \ge b^8$;
    and for any vertex $v \in V$ that is marked as low in \textsc{Build}$(V, E, b)$, $|N_G(v)| \le b^9$.
\end{lemma}
\begin{proof}
    For any high vertex $v$, it must have at least $b^8$ neighbors by Step~\ref{alg:build:s3}, giving the first part of the lemma. 
    
    A vertex $v$ is low if at most $b^8$ cells in $\mathcal{H}(v)$ are occupied (Step~\ref{alg:build:s3}). 
    If $|N_{G}(v)| \ge b^9$, then $v$ is low w.p. at most
    %\footnote{The degree of $v$ includes all parallel edges, where each parallel edge still chooses an index independently and randomly, so the argument holds.}
    \begin{equation*}
        \binom{b^9}{b^8} \cdot \left( \frac{b^8}{b^9} \right)^{b^9} \le b^{9 b^8} \cdot b^{-b^9} \le b^{-b^8} \le n^{-10} ,
    \end{equation*}
    where we used $b = (\log n)^{100}$. So, w.p. $1 - n^{-9}$ all such vertices are marked as high by a union bound and we shall condition on this happening. 
    Since any vertex $v$ with $|N_{G}(v)| \ge b^9$ is marked as high, we have that any low vertex $v$ must have $|N_{G}(v)| \le b^9$. This gives the second part of the lemma.
\end{proof}
%Note that w.h.p. all vertices with degree at least $b^9$ are marked as high, and all vertices with degree less than $b^8$ must be marked as low, but vertices with degree within $[b^8, b^9)$ can be marked as either high or low with no guarantee on the probability. 
We are ready to give two key properties of the subgraph $H$ returned by \textsc{Build}$(V, E, b)$.

\begin{definition} \label{def:component_symbol}
    Given any graph $G = (V, E)$ and a vertex $v \in V$, let $\mathcal{C}_{G}(v)$ be the component in $G$ that contains $v$, i.e., $\mathcal{C}_{G}(v)$ is a vertex-induced subgraph of $G$ on all vertices with paths to $v$ in $G$.
\end{definition}

\begin{lemma} \label{lem:H_component_size}
    Given any graph $G = (V, E)$, let $H = \textsc{Build}(V, E, b)$. W.p. $1 - n^{-8}$, for any vertex $v \in V$, either $V(\mathcal{C}_{G}(v)) = V(\mathcal{C}_{H}(v))$ or $|V(\mathcal{C}_{H}(v))| \ge b^6$.
\end{lemma}
\begin{proof}
    For any vertex $v \in V$, consider $\mathcal{C}_{H}(v)$ (and thus $E'$ and $H$) before Step~\ref{alg:build:s5}. 
    
    For the purpose of analysis, perform breath-first searches on $\mathcal{C}_{H}(v)$ and $\mathcal{C}_{G}(v)$ simultaneously, starting at $v$. 
    At any time of the breath-first searches, let $u$ be the vertex that is currently exploring. If $u$ is low, then all of the adjacent edges of $u$ are added into $E'$ by Step~\ref{alg:build:s4}. 
    So, by an induction (on vertices in any breath-first search order), these two breath-first searches perform the same computations (the order of examined edges and explored vertices are identical) before exploring a high vertex. 
    Therefore, if there is no high vertex in $\mathcal{C}_{H}(v)$, all vertices in $\mathcal{C}_{G}(v)$ and their adjacent edges in $G$ are in $H$, giving $\mathcal{C}_{H}(v) = \mathcal{C}_{G}(v)$. 
    
    On the other hand, if there is a high vertex $u$ in $\mathcal{C}_{H}(v)$, then in Step~\ref{alg:build:s5}, each neighbor $w$ of $u$ in $G$ is added into $H$ w.p. at least $b^{-1}$ (if $w$ is low then it is added w.p. $1$). 
    By Lemma~\ref{lem:build:high_vertex}, w.p. $1 - n^{-9}$ there are at least $b^8$ such $w$. Conditioned on this and apply a Chernoff bound, w.p. $1 - n^{-10}$, at least $b^6$ neighbors $w'$ of $u$ are added into $\mathcal{C}_{H}(v)$ (by adding edge $(u, w')$ into $E'$). 
    By a union bound, w.p. $1 - n^{-8}$, we have $|V(\mathcal{C}_{H}(v))| \ge b^6$.
    
    Observe that Step~\ref{alg:build:s5} does not change $V(\mathcal{C}_{H}(v))$ (but only changes $E(\mathcal{C}_{H}(v))$). As a result, w.p. $1 - n^{-8}$, it must be either $V(\mathcal{C}_{G}(v)) = V(\mathcal{C}_{H}(v))$ (from the second paragraph of this proof) or $|V(\mathcal{C}_{H}(v))| \ge b^6$.
\end{proof}

\begin{lemma} \label{lem:H_is_sparse}
    Given any graph $G = (V, E)$, let $H = \textsc{Build}(V, E, b)$, then $|E(H)| \le (m+n)/(\log n)^5$ w.p. $1 - n^{-8}$.
\end{lemma}
\begin{proof}
    Each edge in $E'$ is either adjacent to a low vertex or between two high vertices. 
    
    By Lemma~\ref{lem:build:high_vertex}, w.p. $1 - n^{-9}$, any low vertex $v$ has at most $b^9$ neighbors in $G$, which means $v$ has at most $b^9$ neighbors in $H$ as the algorithm only deletes edges from $E$ to obtain $E'$ (Step~\ref{alg:build:s4}). By Step~\ref{alg:build:s5}, the number of edges adjacent on $v$ is exactly the number of neighbors in $H$. Combining, the number of edges in $H$ adjacent to low vertices is at most $|V| \cdot b^9 \le n/b$ w.p. $1 - n^{-9}$. 
    
    Assume the number of edges between high vertices is at least $n / (\log n)^6$ in $G$ (otherwise the proof is done). 
    After Step~\ref{alg:build:s4}, by a Chernoff bound, the number of edges between high vertices in $E'$ is at most $(|E'| \log n) / b \le (m \log n) / b$ w.p. $1 - n^{-9}$. Step~\ref{alg:build:s5} only removes edges from $E'$.
    
    By a union bound, w.p. $1 - n^{-8}$, the number of edges in $E'$ is at most $n/b + (m \log n) / b$, which is at most $(m+n)/(\log n)^5$ by $b = (\log n)^{100}$.
\end{proof}

\subsection{Running the Densify Algorithm on the Skeleton Graph} \label{subsec:truncate}

In this section, we introduce the algorithm \textsc{Densify}, which calls the subroutine \textsc{Expand-Maxlink} for certain number of rounds. 

\subsubsection{Algorithmic Framework of \textsc{Densify}} \label{subsubsec:truncate_AF}

The algorithm \textsc{Expand-Maxlink} (presented later) is from \cite{liu2020connected} with minor changes, which uses the following two subroutines.

\begin{framed}
\noindent \textsc{Shortcut}$(V)$: for each vertex $v \in V$: $v.p = v.p.p$.
\end{framed}

The \textsc{Shortcut} is a standard operation in PRAM algorithms for connected components (e.g., see \cite{liu_tarjan, liu2022simple}). 

\begin{framed}
\noindent \textsc{Maxlink}$(V)$: repeat \{for each vertex $v \in V$: let $u \coloneqq \argmax_{w \in N^*(v).p} \ell(w)$, if $\ell(u) > \ell(v)$ then $v.p = u$\} for $2$ iterations.

\end{framed}

For any vertex $v$, let $N^*(v) \coloneqq N(v) \cup \{v\}$, i.e., $N^*(v)$ is the set of neighbors of $v$ plus itself. 
For any vertex set $S$, define $N^*(S) \coloneqq \bigcup_{w \in S} N^*(w)$,
and define $S.p \coloneqq \{w.p \mid w \in S \}$.

The \emph{level} $\ell(v)$ of a vertex $v$ is a non-negative integer that monotonically increases during the algorithm, starting from $1$. 
Given $\beta \ge 0$, we say a vertex $v$ has \emph{budget} $\beta(v) \coloneqq \beta$ if the maximum-size block owned by $v$ has size $\beta$ (a vertex can own more than $1$ blocks, but only the block that is most recently assigned to $v$ has the maximum size; the non-maximum-size blocks contain the added edges and thus cannot be ignored, but when the algorithm does hashing it only uses the maximum-size block). 
%During a round, some roots become non-roots by updating their parents; for those vertex remains to be a root, its level might be increased by $1$;
A root with level $\ell$ is assigned to a block of size 
\begin{equation} \label{eq:setup_beta}
    \beta_{\ell} \coloneqq \beta_1^{1.01^{\ell - 1}} \text{~with~} \beta_1 = (\log n)^{80}
\end{equation}
in Step~\ref{alg:expandm:s8} of \textsc{Expand-Maxlink}. 
Each block of size $\beta$ is partitioned into $\sqrt{\beta}$ (indexed) tables, each with size $\sqrt{\beta}$. 
We assume each vertex in the current graph owns a block of size $\beta_1$ before any application of \textsc{Expand-Maxlink}, which can be guaranteed by approximate compaction in $O(\log^* n)$ time w.h.p. 
Therefore, the hash table of a vertex $v$ (denoted by $\mathcal{H}(v)$ and, without loss of generality, let it be the first indexed table in the block of $v$) has size $(\log n)^{40}$ initially.

The hashing in Step~\ref{alg:expandm:s3} is computed by letting the processor corresponding to edge $(v, w)$ choose a random index from $[0, |\mathcal{H}(v)| - 1]$. 
(Note that the edge $(v, w) \in E$ in the current graph can be an (altered) edge from the input graph, or be an edge added by \textsc{Expand-Maxlink} in the form of $(v, w)$ where $w$ is an item in $\mathcal{H}(v)$, so \textsc{Alter}$(E)$ also applies to those added edges.) 
Since each such edge owns a processor and $|\mathcal{H}(v)|$ can be computed in $O(1)$ time by the budget of $v$, the hashing can be done in $O(1)$ time. 
Step~\ref{alg:expandm:s5} follows the same manner. 

\begin{framed}
\noindent \textsc{Expand-Maxlink}$(H)$:
\begin{enumerate}
    \item Let $V = V(H)$ and $E = E(H)$.  \label{alg:expandm:s0}
    \item \textsc{Maxlink}$(V)$; \textsc{Alter}$(E)$. \label{alg:expandm:s1}
    \item For each root $v \in V$: increase $\ell(v)$ by $1$ w.p. $\beta(v)^{-0.06}$. \label{alg:expandm:s2}
    \item For each root $v \in V$: for each root $w \in N^*(v)$: if $\beta(w) = \beta(v)$ then hash $w$ into $\mathcal{H}(v)$. \label{alg:expandm:s3}
    \item For each root $v \in V$: if there is a collision in $\mathcal{H}(v)$ then mark $v$ as dormant.
        For each vertex $v \in V$: if there is a dormant vertex in $\mathcal{H}(v)$ then mark $v$ as dormant. \label{alg:expandm:s4}
    \item For each root $v \in V$: for each $w \in \mathcal{H}(v)$: for each $u \in \mathcal{H}(w)$: hash $u$ into $\mathcal{H}(v)$. For each root $v \in V$: if there is a collision in $\mathcal{H}(v)$ then mark $v$ as dormant. \label{alg:expandm:s5}
    \item \textsc{Maxlink}$(V)$; \textsc{Shortcut}$(V)$; \textsc{Alter}$(E(V))$. \label{alg:expandm:s6}
    \item For each root $v \in V$: if $v$ is dormant and did not increase its level in Step~\ref{alg:expandm:s2} then increase $\ell(v)$ by $1$. \label{alg:expandm:s7}
    \item For each root $v \in V$: assign a block of size $\beta_{\ell(v)}$ to $v$. \label{alg:expandm:s8}
    \item Return $H$ as the current graph with all added edges (in the form of items in all hash tables). \label{alg:expandm:s9}
\end{enumerate}

\end{framed}

\begin{framed}
\noindent \textsc{Densify}$(H, b)$:
\begin{enumerate}
    \item For \emph{round} $i$ from $0$ to $20 \log b$ \Cliff{Parameter}: $H = \textsc{Expand-Maxlink}(H)$. \label{alg:truncate:s1}
    \item Let $V$ be the set of roots in $H$. \label{alg:truncate:s2}
    \item Repeat for $2\log\log\log n$ times: \textsc{Shortcut}$(V(H))$. \textsc{Alter}$(E(H))$. \label{alg:truncate:s3}
    \item Let $E_{\text{close}}$ be the set of edges in the current graph $H$, including all (altered) edges from the input graph and all (altered) edges in all hash tables of all vertices. \label{alg:truncate:s4}
    \item Call the algorithm in Theorem~\ref{thm:ltz_main} on graph $(V(E_{\text{close}}), E_{\text{close}})$ to update the labeled digraph. \label{alg:truncate:s5}
    \item \textsc{Alter}$(E_{\text{close}})$. \label{alg:truncate:s6}
    \item Return $(V, E_{\text{close}})$.  \label{alg:truncate:s7}
\end{enumerate}

\end{framed}
The algorithm \textsc{Densify}$(H, b)$ calls $\textsc{Expand-Maxlink}(H)$ for $O(\log b)$ times and gradually changes the skeleton graph $H$ by adding edges (Step~\ref{alg:expandm:s9} in $\textsc{Expand-Maxlink}(H)$), where $b = (\log n)^{100}$ is fixed in this section. 
We will show that the total work of \textsc{Densify}$(H, b)$ is $O((m + n) / \poly(\log n))$ w.h.p., and the final skeleton graph $H$ (at the end of round $100 \log b$) is still sparse ($O((m + n) / \poly(\log n))$ edges) w.h.p.\footnote{The algorithm in Theorem~\ref{thm:ltz_main} proceeds in rounds, where each round updates the labeled digraph corresponding to $G$ and takes $O(1)$ time.}

\subsubsection{Properties of \textsc{Densify}} \label{subsubsec:truncate_properties}

In this section, we prove the key properties of the algorithm \textsc{Densify} that calls \textsc{Expand-Maxlink} repeatedly. 
Some properties of \textsc{Expand-Maxlink} are already presented in \cite{liu2020connected}.%, so we omit the proof or only provide a proof sketch.

\begin{lemma} \label{lem:expandm_processors}
    W.p. $1 - 1/(\log n)^9$, Step~\ref{alg:truncate:s1} of \textsc{Densify}$(H, b)$ uses $O((m + n) / (\log n)^4)$ processors.
\end{lemma}
\begin{proof}
    In \cite{liu2020connected}, it is proved that w.p. $1 - 1/(\log n)^{10}$, the first $O(\log n)$ rounds of $\textsc{Expand-Maxlink}(G')$ uses $O(|E(G')| + |V(G')|)$ processors\footnote{\cite{liu2020connected} assumes $|E(G')| \ge 2 |V(G')|$ for simplicity in presenting bounds.} (see their Lemma 5.3). 
    Since $b = (\log n)^{100}$ in this section, \textsc{Densify}$(H, b)$ calls $O(\log\log n)$ rounds of $H = \textsc{Expand-Maxlink}(H)$ and the number of processors used is $O(|E(H)| + |V(H)|)$. 
    By the assumption at the beginning of \S{\ref{sec:stage2}}, $|V(H)| \le n/b^{10}$. 
    By Lemma~\ref{lem:H_is_sparse}, $|E(H)| \le (m+n)/(\log n)^5$ w.p. $1 - n^{-8}$. 
    The lemma follows by $b = (\log n)^{100}$. 
\end{proof}

Later in \S{\ref{sec:assumption}}, we will gradually increase the value of $b$, but it will be at most $n$, so Lemma~\ref{lem:expandm_processors} can still be applied.

\begin{lemma} \label{lem:expandm_max_level}
    W.p. $1 - 1/(\log n)^9$, during the execution of \textsc{Densify}$(H, b)$, the level $\ell(v)$ of any vertex $v \in V(H)$ is at most $100 \log\log n$.
\end{lemma}
\begin{proof}
    Assume \textsc{Densify}$(H, b)$ uses $O((m + n) / (\log n)^4)$ processors. 
    By examining \textsc{Maxlink}$(V)$ -- the only place that updates parents in $\textsc{Expand-Maxlink}(H)$ in \textsc{Densify}$(H, b)$ -- we have that $\ell(v) \le \ell(v.p)$ for any vertex $v$, which holds at any time since the level increase only happens on roots (Step~\ref{alg:expandm:s2} and Step~\ref{alg:expandm:s7}) whose parents are themselves.
    Assume for contradiction that there exists a vertex $v \in V(H)$ such that $\ell(v) > 100 \log\log n$. 
    The root $u$ of the tree (in the labeled digraph) containing $v$ must have $\ell(u) > 100 \log\log n$. 
    In Step~\ref{alg:expandm:s8} of $\textsc{Expand-Maxlink}(H)$, root $u$ is assigned a block of size $\beta_{\ell(u)}$, which, by Equation~(\ref{eq:setup_beta}), implies that the number of processors in the block owned by $u$ is
    \begin{equation*}
        \beta(u) = \beta_{\ell(u)} \ge \beta_1^{1.01^{100 \log\log n}} \ge ((\log n)^{80})^{\log n} \ge n^{80} \ge m+n ,\footnote{In the last inequality, the degree on $n$ can be made arbitrarily large such that it is greater than $m + n$ for $m \le n^c$.} 
    \end{equation*} 
    contradicting with the assumption we made at the beginning of this proof. 
    The lemma then follows from Lemma~\ref{lem:expandm_processors}.
\end{proof}

\begin{lemma} \label{lem:expandm_constant_time}
    W.p. $1 - 1/(\log n)^6$, any round of \textsc{Densify}$(H, b)$ can be implemented on an ARBITRARY CRCW PRAM in $O(1)$ time.% and $(m + n) / (\log n)^3$ work.
\end{lemma}
\begin{proof}
    By an induction on round $i$ from $0$ to $20 \log b$ \Cliff{Parameter}, we assume each edge and each vertex in the current graph $H$ has a corresponding processor, which holds at the beginning of round $0$. 
    Consider the \textsc{Expand-Maxlink}$(H)$ in round $i$ of \textsc{Densify}$(H, b)$. 
    
    Firstly, we show that the \textsc{Maxlink}$(V)$ in Step~\ref{alg:expandm:s1} and Step~\ref{alg:expandm:s6} takes $O(1)$ time. 
    It suffices to compute $u \coloneqq \argmax_{w \in N^*(v).p} \ell(w)$ for each $v \in V$ in $O(1)$ time. 
    We could not compute the maximum level over the neighbors parents set directly in constant time.
    Observe from Lemma~\ref{lem:expandm_max_level} that w.p. $1 - 1/(\log n)^9$, all levels are in $[100 \log\log n]$. 
    Using the second indexed table in the block of $v$ (whose size is at least $\sqrt{\beta_1} = (\log n)^{40}$), the processor corresponding to edge $(w', v)$ reads $w = w'.p$ from $w'$ then reads $\ell(w)$ from $w$, for each $w' \in N^*(v)$. Next, this processor writes $w$ into the $\ell(w)$-th cell of the second indexed table in the block of $v$. Since such $w$'s are within the first $100 \log\log n$ cells, we can use $\log n$ processors from the third indexed table of $v$ to determine the $w$ with maximum level $\ell(w)$ by assigning each pair from $[100 \log\log n] \times [100 \log\log n]$ a corresponding processor, giving $u \coloneqq \argmax_{w \in N^*(v).p} \ell(w)$. 
    
    Secondly, we show that the hashings in Step~\ref{alg:expandm:s3} and Step~\ref{alg:expandm:s5} take $O(1)$ time. 
    Step~\ref{alg:expandm:s3} is straightforward since each $w \in N^*(v)$ has a corresponding processor, and the hashing is computed in $O(1)$ time as explained before the pseudocode. 
    To implement Step~\ref{alg:expandm:s5}, each vertex $u \in \mathcal{H}(w)$ such that $w \in \mathcal{H}(v)$ needs to own a corresponding processor. Note that by Step~\ref{alg:expandm:s3}, it must be $\beta(w) = \beta(v)$, so we have that for each $w \in \mathcal{H}(v)$, the table size of $w$ is exactly $\sqrt{\beta(w)} = \sqrt{\beta(v)}$. Since $v$ owns $\beta(v)$ processors in its block, we can assign a processor to each cell in the table of such $w$ for each $w \in \mathcal{H}(v)$ (the table size of $v$ is $\sqrt{\beta(v)}$) allocated by their indices. 
    Therefore, in $O(1)$ time we use the processor corresponding to $u$ to write $u$ into $\mathcal{H}(v)$. The added edge $(u, v)$ owns a corresponding processor -- the same processor in $\mathcal{H}(v)$. 
    
    Thirdly, we show that Step~\ref{alg:expandm:s8} can be implemented in $O(1)$ time. 
    By Lemma~\ref{lem:expandm_processors}, w.p. $1 - 1/(\log n)^9$, the total number of processors used is at most $(m + n) / (\log n)^3$. 
    The pool of $(m + n) / (\log n)^3$ processors is partitioned into $(\log n)^2$ zones such that the processor allocation in round $i \in [20 \log b]$ \Cliff{parameter} for vertices with level $\ell$ uses the zone indexed by $(i, \ell)$, where $\ell \in [100 \log\log n]$ w.p. $1 - 1/(\log n)^9$ by Lemma~\ref{lem:expandm_max_level}. 
    Since there are $(m + n) / (\log n)^3$ processors in total and all the vertex ids are in $[n/b^{10}]$ as assumed at the beginning of \S{\ref{sec:stage2}}, we can use $O(n/b^{10} \cdot \log n)$ processors (such that there are $\Theta(\log n)$ processors per vertex) for each different level and apply a tighten version of Lemma~\ref{lem:apx_compaction} (see \cite{liu2020connected} for details) to index each root in $O(1)$ time with probability $1 - n^{-9}$ such that the indices of vertices with the same level are distinct, then assign each of them a distinct block in the corresponding zone.\footnote{The number of processor used for this part is actually already counted in Lemma~\ref{lem:expandm_processors}.} 
    Therefore, Step~\ref{alg:expandm:s8} takes $O(1)$ time w.p. $1 - 1/(\log n)^8$ by a union bound over all $100\log\log n$ levels.%, and the work is at most $O(n / b^{10} \cdot \log n) \cdot 100 \log\log n \le n / (\log n)^3$ by $b = (\log n)^{100}$.
    
    All other computations in the \textsc{Expand-Maxlink}$(H)$ in round $i$ of \textsc{Densify}$(H, b)$ can be done in $O(1)$ time, and each edge and each vertex has a corresponding processor at the end of round $i$. This completes the induction on $i$ and gives the lemma by a union bound over all rounds.
\end{proof}

\begin{lemma} \label{lem:truncate:flat}
     W.p. $1 - 1/(\log n)^8$, at the end of \textsc{Densify}$(H, b)$, for any vertex $v \in V(H)$, $v$ is either a root or a child of a root,\footnote{We did not state that all trees in the labeled digraph are flat because the vertex set $V(H)$ does not contain all the vertices in the original input graph: $H$ is the current graph after Stage $1$ that induced on roots. The non-roots created in Stage $1$ can be as many as $\Theta(n)$ so that we cannot run \textsc{Shortcut} on all of them. Instead, we will prove that at the end of Stage $3$, all such non-roots are children or grandchildren of roots so that one application of \textsc{Shortcut} flattens all trees and thus gives a correct connected components algorithm, which adds $O(n)$ work in total.\label{fn:leaves_not_in_trees}}
     and both ends of any edge in $E_{\text{close}}$ are roots.
\end{lemma}
\begin{proof}
    The only place that updates parents is in \textsc{Maxlink}$(V)$, which gives $\ell(v) < \ell(v.p)$ for any non-root $v$.  
    Note that this relation holds at any time of the algorithm, because the level increase only happens on roots (Step~\ref{alg:expandm:s2} and Step~\ref{alg:expandm:s7} within \textsc{Expand-Maxlink}$(H)$). 
    As a result, the distance from any vertex $v \in V(H)$ to the root of the tree containing $v$ is at most the maximum level, which is at most $100 \log\log n$ w.p. $1 - 1/(\log n)^9$ by Lemma~\ref{lem:expandm_max_level}. 
    (The vertices in not in $V(H)$ must be leaves whose parents are in $V(H)$ since Stage $1$ is a contraction algorithm, also see footnote~\ref{fn:leaves_not_in_trees}.) 
    One iteration of \textsc{Shortcut}$(V(H))$ in Step~\ref{alg:truncate:s3} of \textsc{Densify}$(H, b)$ reduces the distance $\delta$ from a vertex $v$ to the root of the tree containing $v$ to at most $2\delta / 3$ if $\delta \ge 2$ (the worst case is reducing the distance from $3$ to $2$, see \cite{liu_tarjan} for details). 
    Therefore, after $2 \log\log n \ge \log_{3/2} (100 \log\log n)$ iterations of \textsc{Shortcut}$(V(H))$, the distance from any vertex $v \in V(H)$ to the root of the tree containing $v$ is at most $1$. 
    The following \textsc{Alter}$(E(H))$ in Step~\ref{alg:truncate:s3} makes all edges adjacent to roots. 
    Therefore, the contraction algorithm in Theorem~\ref{thm:ltz_main} operates on all roots with adjacent edges (in $E_{\text{close}}$) and maintains the invariant that each vertex in $V(H)$ is either a root or a child of a root at the end of Step~\ref{alg:truncate:s5} w.p. $1 - 1/(\log n)^9$, giving the first part of the lemma. 
    The \textsc{Alter}$(E_{\text{close}})$ in Step~\ref{alg:truncate:s6} moves all edges in $E_{\text{close}}$ adjacent to roots.
\end{proof}

\begin{lemma} \label{lem:truncate_time_work}
    W.p. $1 - 1/(\log n)^3$, \textsc{Densify}$(H, b)$ can be implemented on an ARBITRARY CRCW PRAM in $O(\log b)$ time and $(m + n) / (\log n)^3$ work.
\end{lemma}
\begin{proof}
    By Lemma~\ref{lem:expandm_processors} and Lemma~\ref{lem:expandm_constant_time}, w.p. $1 - 1/(\log n)^5$, Step~\ref{alg:truncate:s1} of \textsc{Densify}$(H, b)$ takes $O(\log b)$ time and $O(m + n) / (\log n)^4$ work. 
    Consider Step~\ref{alg:truncate:s3}. By the assumption at the beginning of \S{\ref{sec:stage2}}, $|V(H)| \le n/b^{10} \le n / (\log n)^5$ by $b = (\log n)^{100}$. 
    So the $2\log\log n$ iterations of \textsc{Shortcut}$(V(H))$ takes at most $n / (\log n)^4$ work and $\log b$ time. 
    By Lemma~\ref{lem:H_is_sparse}, $|E(H_0)| \le (m+n)/(\log n)^5$ w.p. $1 - n^{-8}$, where $H_0$ is the graph $H$ passed to \textsc{Expand-Maxlink} in round $0$. 
    The number of edges added into $H$ over all rounds in Step~\ref{alg:truncate:s1} is upper bounded by the total work performed over all rounds of \textsc{Expand-Maxlink}$(H)$, which is $O(m + n) / (\log n)^4$ w.p. $1 - 1/(\log n)^5$ by the first sentence of this proof. 
    By a union bound, the work in the first $4$ steps of \textsc{Densify}$(H, b)$ is at most $O(m + n) / (\log n)^4$ w.p. $1 - 1/(\log n)^4$. 
    Note that $E_{\text{close}}$ is exactly the $E(H)$ at the end of Step~\ref{alg:truncate:s4}, so w.p. $1 - 1/(\log n)^9$, Step~\ref{alg:truncate:s5} takes $O(m + n) / (\log n)^4 \cdot 100 \log\log n \le (m + n) / (\log n)^{3.5}$ work \Cliff{parameter} and  $100 \log\log n \le \log b$ time. 
    Step~\ref{alg:truncate:s6} takes $O(1)$ time and $O(E_{\text{close}}) \le O(m + n) / (\log n)^4$ work. 
    Summing up and by a union bound, w.p. $1 - 1/(\log n)^3$, all steps take $O(\log b)$ time and $(m + n) / (\log n)^3$ work in total.
\end{proof}

\Cliff{The above lemma can be boosted to high probability, maybe write in \S{\ref{sec:boosting}}.}

We use one crucial property of \textsc{Expand-Maxlink} from \cite{liu2020connected} (stated as Lemma 5.24 in their full version), which will be used later.

\begin{lemma}[\cite{liu2020connected}] \label{lem:connect_2}
    For any root $v \in V(H)$ and any vertex $u \in N^*(N^*(v))$ at the beginning of any round $i$ within \textsc{Densify}$(H, b)$, 
    let $u'$ be the parent of $u$ after Step~\ref{alg:expandm:s1} of the \textsc{Expand-Maxlink}$(H)$ in round $i$.
    If $v$ does not increase level and is a root during round $i$, then $u' \in \mathcal{H}(v)$ after Step~\ref{alg:expandm:s5} of the \textsc{Expand-Maxlink}$(H)$ within round $i$.
\end{lemma}

\subsubsection{Path Construction for \textsc{Densify}} \label{subsubsec:truncate_path_construction}

Fix any shortest path $P$ in graph $H$. 
The length of $P$ is not necessarily $d_G$ -- the diameter of the original input graph $G$, as shown in \S{\ref{sec:diameter}}.

Consider the changes on $P$ over rounds of \textsc{Densify}$(H, b)$. 
The \textsc{Alter}$(E)$ (Step~\ref{alg:expandm:s1} and Step~\ref{alg:expandm:s6}) in each round replaces an edge $(v, w)$ on $P$ by $(v.p, w.p)$. 
The path $P$ also changes when adding edges to the current graph:
for any vertices $v$ and $w$ on path $P$, if the current graph contains edge $(v, w)$, then all vertices between $v$ and $w$ (exclusively) can be removed from $P$, which still results in a valid path in the current graph from the first to the last vertex of $P$, reducing the length of $P$. 
We use $|P|$ as the number of vertices on $P$, which is $1$ plus the length of $P$. 
For any integer $j \in [0, |P| - 1]$, let $P(j)$ be the $j$-th vertex on $P$.
%If all such paths reduce their lengths to at most $d'$, the diameter of the current graph is at most $d'$.

Formally, we have the following inductive construction for path changes:

\begin{framed}
\noindent \textsc{PathConstruction}$(H, P)$:

    Let all vertices on $P$ be \emph{active} at the beginning of round $0$ of \textsc{Densify}$(H, b)$.
    
    For any round $i \in [0, 9 \log b)$, given path $P_i$ with at least $4$ active vertices at the beginning of round $i$, \textsc{Expand-Maxlink}$(H)$ constructs $P_{i+1}$ by the following $6$ phases:
    \begin{enumerate}
        \item The \textsc{Alter}$(E)$ in Step~\ref{alg:expandm:s1} replaces each vertex $v$ on $P_i$ by $v' \coloneqq v.p$ to get path $P_{i,1}$. Set $v'$ as active or passive according to $v$. For any $v'$ on $P_{i,1}$, let $\underline{v'}$ be on $P_{i}$ such that $\underline{v'}.p = v'$.\label{phase_1}
        \item After Step~\ref{alg:expandm:s5}, set $j$ as $1$, and repeat the following until $j \ge |P_{i,1}| - 1$ or $P_{i,1}(j + 1)$ is \emph{passive} (not active): 
            \begin{itemize}
                \item Let $v' \coloneqq P_{i,1}(j)$, if $\underline{v'}$ is a root and does not increase level during round $i$ then: 
                \begin{enumerate}
                    \item If the current graph contains edge $(v', P_{i,1}(j+2))$ then mark $P_{i,1}(j+1)$ as \emph{skipped} and set $j \coloneqq j+2$. \label{phase_2a}
                    \item Else set $j \coloneqq j+1$.
                \end{enumerate}
            \end{itemize} \label{phase_2}
        %\item If $j = |P_{i,1}|$, then mark $P_{i,1}(j-1)$ as passive. \label{phase_3}
        \item Let $k$ be the maximum integer such that $P_{i,1}(k)$ is active and $k \le j$. Mark $P_{i,1}(k)$ as passive. \label{phase_3}
        \item Remove all skipped and passive vertices from $P_{i,1}$ to get path $P_{i,4}$. \label{phase_4}
        \item Concatenate $P_{i,4}$ with all passive vertices after $P_{i,1}(|P_{i,4}|)$ (exclusively) on $P_{i,1}$ to get path $P_{i,5}$. \label{phase_5}
        \item The \textsc{Alter}$(E)$ in Step~\ref{alg:expandm:s6} replaces each vertex $v$ on $P_{i,5}$ by $v.p$ to get path $P_{i+1}$. Set $v.p$ as active or passive according to $v$. \label{phase_6}
    \end{enumerate}
    For any vertex $v$ on $P_i$ that is replaced by $v'$ in Phase~\ref{phase_1}, if $v'$ is not removed in Phase~\ref{phase_4}, then let $\overline{v}$ be the vertex replacing $v'$ in Phase~\ref{phase_6}, and call $\overline{v}$ the \emph{corresponding vertex} of $v$ in round $i+1$.
\end{framed}

\footnote{The definitions of $\underline{v'}$ and $\overline{v}$ are for our later potential argument. In Phase~\ref{phase_1}, if there are two vertices $P_i(x)$ and $P_i(y)$ with the same parent $v'$ where $x \ne y$, then $\underline{P_{i, 1}(x)} = P_i(x)$ and $\underline{P_{i, 1}(y)} = P_i(y)$, which resolves ambiguity by their indices. Phase~\ref{phase_6} follows the same manner. We could have defined potentials on indices on the path rather than on the vertices, but this involves too many parameters to parse.\label{fn:potential_index}}To obtain bounds on the length of $P$ over rounds, we need the following potential function.

\begin{definition}\label{def:path_potential}
    For any (active) vertex $v$ on path $P_0 \coloneqq P$, define its \emph{potential} $\Phi_0(v) \coloneqq 1$.
    For any integer $i \in [0, 9 \log b)$, given path $P_i$ with at least $4$ active vertices at the beginning of round $i$ of \textsc{Densify}$(H, b)$, define the \emph{potential} of each vertex on $P_{i+1}$ based on \textsc{PathConstruction}$(H, P)$ by the following steps: %\footnote{The second subscript of a potential indicates the phase to obtain that potential; the subscript for paths in Definition~\ref{def_path_construction} follows the same manner.}
    \begin{enumerate}
        \item For each vertex $v$ on $P_i$ that is replaced by $v.p$ in Phase~\ref{phase_1}, $\Phi_{i, 1}(v.p) \coloneqq \Phi_i(v)$. \label{potential:step1}
        \item After Phase~\ref{phase_3}, for each active vertex $v$ on $P_{i,1}$, if the successor $w$ of $v$ is skipped or passive then $\Phi_{i,3}(v) \coloneqq \Phi_{i,1}(v) + \Phi_{i,1}(w)$ and $\Phi_{i, 3}(w) \coloneqq 0$. \label{potential:step2}
        \item For each vertex $v$ on $P_{i,1}$ where $\Phi_{i,3}(v)$ is not updated in Step~\ref{potential:step2}, $\Phi_{i,3}(v) \coloneqq \Phi_{i,1}(v)$. \label{potential:step3}
        \item For each vertex $v$ on $P_{i,5}$ that is replaced by $v.p$ in Phase~\ref{phase_6}, $\Phi_{i+1}(v.p) \coloneqq \Phi_{i,3}(v)$. \label{potential:step4}
    \end{enumerate}
\end{definition}

The following lemma is immediate:
\begin{lemma} \label{lem:potential_sum_decrease}
For any integer $i \in [0, 9 \log b)$, let path $P_i$ with at least $4$ active vertices be defined in \textsc{PathConstruction}$(H, P)$, then
%\begin{equation} \label{eq:potential_sum}
    $\sum_{v \text{~is~on~} P_i} \Phi_i(v) \ge \sum_{v \text{~is~on~} P_{i+1}} \Phi_{i+1}(v) $.
%\end{equation}
\end{lemma}
\begin{proof}[Proof Sketch.]
    In Definition~\ref{def:path_potential}, 
(\romannumeral1) the mapping in Step~\ref{potential:step1} is one-to-one (see footnote~\ref{fn:potential_index});
(\romannumeral2) in Step~\ref{potential:step2}, the skipped or passive vertex $w$ passes its potential to $v$;
(\romannumeral3) in Step~\ref{potential:step3}, only a vertex not updates its potential in Step~\ref{potential:step2} inherits the potential;
(\romannumeral4) the mapping in Step~\ref{potential:step4} is one-to-one.
\end{proof}

\begin{lemma} \label{lem:active_vertex_potential_bound}
    For any integer $i \in [0, 9 \log b)$, let path $P_i$ with at least $4$ active vertices be defined in \textsc{PathConstruction}$(H, P)$. 
    For any active vertex $v$ on $P_i$, $\Phi_i(v) \ge 2^{i - \ell(v)}$.
\end{lemma}
\begin{proof}
    The proof is by an induction on $i$.
    The base case holds because for any vertex $v$ on $P_0 = P$, $\Phi_0(r) = 1 \ge 2^{0 - 1}$.
    Now we prove the inductive step from $i$ to $i+1$ given that the corresponding vertex $\overline{v}$ of $v \in P_i$ is on $P_{i+1}$ and active.

    Suppose $v$ is a non-root at the end of round $i$.
    If $v$ is a non-root before the \textsc{Alter}$(E(V))$ in Step~\ref{alg:expandm:s6} within \textsc{Expand-Maxlink}$(H)$ within round $i$ of \textsc{Densify}$(H,b)$, then $\ell(v.p) > \ell(v)$ by examining \textsc{Maxlink}$(V)$, and $\ell(\overline{v}) \ge \ell(v.p) > \ell(v)$ (Phase~\ref{phase_6} in \textsc{PathConstruction}$(H,P)$). 
    So, by the induction hypothesis, $\Phi_{i+1}(\overline{v}) \ge \Phi_{i}(v) \ge 2^{i - \ell(v)} \ge 2^{i+1 - \ell(\overline{v})}$.

    Suppose $v$ increases its level in round $i$.
    Let $\ell$ be the level of $v$ at the beginning of round $i$.
    If the increase happens in Step~\ref{alg:expandm:s2}, then $v$ is a root after Step~\ref{alg:expandm:s1}.
    Whether $v$ changes its parent in Step~\ref{alg:expandm:s6} or not, the level of $\overline{v} = v.p$ is at least $\ell+1$.
    Else, if the increase happens in Step~\ref{alg:expandm:s7}, then $v$ is a root after Step~\ref{alg:expandm:s6}.
    So $\overline{v} = v$ and its level is at least $\ell + 1$ at the end of round $i$.
    By the induction hypothesis, $\Phi_{i+1}(\overline{v}) \ge \Phi_{i}(v) \ge 2^{i - \ell} \ge 2^{i+1 - \ell(\overline{v})}$.

    Since a non-root never becomes a root, 
    it remains to assume that $v$ is a root and does not increase level during round $i$.
    Since we assumed that the corresponding vertex $\overline{v}$ of $v \in P_i$ is on $P_{i+1}$ and active, the vertex $v'$ that replaces $v$ on $P_i$ in Phase~\ref{phase_1} of \textsc{PathConstruction}$(H,P)$ (i.e., $\underline{v'}.p = v'$) is exactly $v$, and is not removed in Phase~\ref{phase_4} and remains to be active before Phase~\ref{phase_6}. Moreover, by Step~\ref{potential:step1} of Definition~\ref{def:path_potential}, $\Phi_{i,1}(v') = \Phi_{i,1}(v) = \Phi_i(v)$. 
    \begin{enumerate}
        \item Suppose the index $j$ of $v$ on $P_{i,1}$ satisfies that (\romannumeral1) $j \ge |P_{i,1}| - 1$, or (\romannumeral2) $j < |P_{i,1}| - 1$ and $P_{i,1}(j + 1)$ is passive. 
    We prove that $j$ is the maximum integer such that $P_{i,1}(j)$ is active. 
    Case (\romannumeral1) is easy since $j$ is the last vertex on $P_{i,1}$.
    It remains to prove case (\romannumeral2). 
    Assume for contradiction that there exists $j' > j$ such that $P_{i,1}(j')$ is active. Then we have $j' > j+1$ since $P_{i,1}(j + 1)$ is passive. Since $P_{i,1}(j + 1)$ is passive before Phase~\ref{phase_3} of round $i$ in \textsc{PathConstruction}$(H,P)$, in some round $i' < i$ the \textsc{PathConstruction}$(H,P)$ marks $j+1$ from active to passive, which means $j+1$ is maximum integer such that $P_{i',1}(j+1)$ is active, and $P_{i', 1}(j')$ must be passive. This implies that the $j'$-th vertex on this path remains to be passive starting from round $i'$, contradicting with the fact that $P_{i,1}(j')$ is active. 
    As a result, $j$ must be the maximum integer such that $P_{i,1}(j)$ is active, and $v = P_{i,1}(j)$ is marked as passive in Phase~\ref{phase_3} of round $i$ in \textsc{PathConstruction}$(H,P)$. 
    This contradicts with the fact that the corresponding vertex $\overline{v} $ of $v$ is active on $P_{i+1}$ (the first paragraph of this proof). \label{potential:case1}
    
        \item Suppose the index $j$ of $v$ on $P_{i,1}$ satisfies that $j < |P_{i,1}| - 1$ and $P_{i,1}(j + 1)$ is active. 
        \begin{enumerate}
            \item Suppose $P_{i,1}(j+2)$ is an active vertex on $P_{i,1}$. 
    Observe that $P_{i,1}(j+2)$ is the parent $u'$ of a vertex in $N^*(N^*(v))$ (where $N^*$ is defined at the beginning of round $i$) after Step~\ref{alg:expandm:s1} within \textsc{Expand-Maxlink}$(H)$ of round $i$, which must be in $\mathcal{H}(v)$ after Step~\ref{alg:expandm:s5} by Lemma~\ref{lem:connect_2}. 
    Therefore, the current graph contains edge $(v, u')$ and $w \coloneqq P_{i,1}(j+1)$ is skipped by Phase~\ref{phase_2a} in \textsc{PathConstruction}$(H,P)$, 
    thus 
    \begin{equation} \label{eq:potential_pass}
        \Phi_{i,3}(v) = \Phi_{i,1}(v) + \Phi_{i,1}(w)
    \end{equation}
    by Step~\ref{potential:step2} in Definition~\ref{def:path_potential}. 
    Next, we prove that $\ell(w) \le \ell(v)$ at the beginning of round $i$. 
    Assume for contradiction that $\ell(w) > \ell(v)$. 
    Then, $w \in N^*(v)$ by Step~\ref{alg:expandm:s3} and Step~\ref{alg:expandm:s5} of \textsc{Expand-Maxlink}$(H)$, where $N^*(v)$ is defined in Step~\ref{alg:expandm:s3}. 
    Since $v$ does not increase its level during round $i$, $v$ is not dormant in Step~\ref{alg:expandm:s7}. 
    Therefore, there is no collision in the hash tables of $v$ in Step~\ref{alg:expandm:s5} and Step~\ref{alg:expandm:s4}, 
    and there is no collision in the hash table of $w$ in Step~\ref{alg:expandm:s4}. 
    This means $w \in \mathcal{H}(w) = N^*(w)$ in Step~\ref{alg:expandm:s4}, and $w \in \mathcal{H}(v)$ in Step~\ref{alg:expandm:s5}. 
    Now we have $w \in N^*(v)$ before Step~\ref{alg:expandm:s6}. 
    Consider the \textsc{Maxlink}$(V)$ in Step~\ref{alg:expandm:s6}. 
    Since $v$ does not increase its level and the level of $w$ cannot decrease, $\ell(w) > \ell(v)$ still holds at this point. 
    In the first iteration of \textsc{Maxlink}$(V)$, there exists a $u' \in N^*(v).p$ such that $\ell(u') > \ell(v)$ due to the existence of $w$ and the fact that $\ell(w.p) \ge \ell(w)$. 
    As a result, $v$ will be a non-root at the end of that iteration. Since a non-root never becomes a root, we have that $v$ is a non-root after this \textsc{Maxlink}$(V)$, contradicting with the fact that $v$ is a root in round $i$. 
    Therefore, we have that $\ell(w) \le \ell(v)$ at the beginning of round $i$. 
    Note that vertex $\underline{w}$ (the vertex that is replaced by $w$ in the \textsc{Alter}$(E)$ in Step~\ref{alg:expandm:s1} of \textsc{Expand-Maxlink}$(H)$, see Phase~\ref{phase_1} in \textsc{PathConstruction}$(H,P)$) is on $P_i$ and $\ell(\underline{w}) \le \ell(w) \le \ell(v)$ since $w = \underline{w}.p$.
    Therefore, by the induction hypothesis, we obtain 
    \begin{equation} \label{eq:potential_w}
        \Phi_i(\underline{w}) \ge 2^{i - \ell(\underline{w})} \ge 2^{i - \ell(v)}.
    \end{equation}
    %Meanwhile, by Step~\ref{potential:step1} in Definition~\ref{def:path_potential}, we get
    %\begin{equation} \label{eq:potential_inherit}
    %    \Phi_{i,1}(w) = \Phi_i(\underline{w}) .
    %\end{equation}
    Combining, we obtain
    \begin{align*}
        \Phi_{i+1}(v) &= \Phi_{i,3}(v) \\
        &= \Phi_{i,1}(v) + \Phi_{i,1}(w) \\
        &= \Phi_i(v) + \Phi_i(\underline{w}) \\
        &\ge 2^{i - \ell(v)} + 2^{i - \ell(v)} = 2^{i + 1 - \ell(v)} ,
    \end{align*}
    where the first line is by the fact that $v$ is a root in round $i$ and Step~\ref{potential:step4} in Definition~\ref{def:path_potential}, 
    the second line is by Equation~(\ref{eq:potential_pass}),
    the third line is by Step~\ref{potential:step1} in Definition~\ref{def:path_potential}, 
    the fourth line is by the induction hypothesis on $v$ and Equation~(\ref{eq:potential_w}). 
    This finishes the inductive step for this case. 
    \label{potential:case2a}
    
            \item Suppose $j+2 \ge |P_{i,1}|$ or $P_{i,1}(j+2)$ is passive. Then, using the same argument in case~\ref{potential:case1} of this proof, we get that $w \coloneqq P_{i,1}(j+1)$ is the last active vertex on $P_{i,1}$, and $w$ is marked as passive in Phase~\ref{phase_3} of round $i$ in \textsc{PathConstruction}$(H,P)$. 
            Since $v$ is active, we have that $\Phi_{i,3}(v) = \Phi_{i,1}(v) + \Phi_{i,1}(w)$ by Step~\ref{potential:step2} in Definition~\ref{def:path_potential}. 
            Since $w$ is active before Phase~\ref{phase_3} of round $i$ in \textsc{PathConstruction}$(H,P)$, we can use the same argument in case~\ref{potential:case2a} of this proof to obtain that $\Phi_{i+1}(v) \ge 2^{i + 1 - \ell(v)}$. 
        \end{enumerate}
    \end{enumerate}
    As a result, the lemma holds for any active vertex $\overline{v}$ on $P_{i+1}$, finishing the induction and giving the lemma.
\end{proof}

\begin{lemma} \label{lem:active_passive_length}
    Given graph $H$ and integer $b \ge (\log n)^{100}$. 
    Let $r = 9 \log b - 1$. 
    W.p. $1 - 1/(\log n)^9$, for any shortest path $P$ in $H$ such that there are at most $b^7$ vertices on $P$, there are at most $10 \log b$  vertices on $P_r$, where $P_r$ is defined in \textsc{PathConstruction}$(H, P)$.
\end{lemma}
\begin{proof}
    W.p. $1 - 1/(\log n)^9$, suppose Lemma~\ref{lem:expandm_max_level} holds. 
    
    Firstly, we prove that there are at most $3$ active vertices on $P_r$.
    Assume for contradiction that there are at least $4$ active vertices on $P_r$. 
    By an induction on round $i \in [0, r]$, Lemma~\ref{lem:potential_sum_decrease}, and the fact that each vertex on $P_0 = P$ has potential $1$ and the potentials are non-negative (Definition~\ref{def:path_potential}), we get 
    \begin{equation} \label{eq:P_r_sum}
        \sum_{v \text{~is~active~and~on~} P_r} \Phi_r(v) \le \sum_{v \text{~is~on~} P_r} \Phi_r(v) \le b^7 .
    \end{equation}
    Since $P_r$ has at least $4$ active vertices, by Lemma~\ref{lem:active_vertex_potential_bound} we get for any active vertex $v$ on $P_r$,
    \begin{equation*}
        \Phi_r(v) \ge 2^{r - \ell(v)} \ge 2^{9 \log b - 1 - 100 \log\log n} = 2^{8 \log b - 1} > b^7 ,
    \end{equation*}
    where we have used $b \ge (\log n)^{100}$ and $\ell(v) \le 100 \log\log n$ for all $v \in V(H)$. 
    This contradicts with Equation~(\ref{eq:P_r_sum}). 
    
    Secondly, we show that there are at most $9 \log b$ passive vertices on $P_r$. This is because in each of the $9\log b$ rounds of \textsc{PathConstruction}$(H, P)$, we mark at most $1$ vertex from active to passive (Phase~\ref{phase_3}). Since all skipped vertices are removed from the path (Phase~\ref{phase_4}), we obtain that the number of vertices on $P_r$ is at most $9 \log b + 3 \le 10 \log b$.
\end{proof}

Now we reduce the length of any shortest path with length at most $b^7$ to $10 \log b$. 
To reduce its length even further, we need an addition path construction.

\begin{framed}
\noindent \textsc{PassivePathConstruction}$(H, P)$:
    
    %Let all vertices on $P$ be \emph{active} at the beginning of round $0$ of \textsc{Densify}$(H, b)$.
    
    Let $r= 9 \log b - 1$.
    For any integer $i \in [0, r]$, let path $P_i$ be defined as in \textsc{PathConstruction}$(H, P)$.
    
    For any integer $i \in [r, 20\log b]$, 
    given path $P_i$ with at least $3$ vertices at the beginning of round $i$, \textsc{Expand-Maxlink}$(H)$ constructs $P_{i+1}$ by the following $3$ phases:
    \begin{enumerate}
        \item The \textsc{Alter}$(E)$ in Step~\ref{alg:expandm:s1} replaces each vertex $v$ on $P_i$ by $v' \coloneqq v.p$ to get path $P_{i,1}$. For any $v'$ on $P_{i,1}$, let $\underline{v'}$ be on $P_{i}$ such that $\underline{v'}.p = v'$. \label{p_phase_1}
        \item After Step~\ref{alg:expandm:s5}, let $v' \coloneqq P_{i,1}(0)$, i.e., the first vertex on $P_{i, 1}$.
            If $\underline{v'}$ is a root at the end of round $i$ and does not increase level during round $i$ then: if the current graph contains edge $(v', P_{i,1}(2))$ then remove $P_{i,1}(1)$. After removing all such vertices from $P_{i,1}$ we obtain $P_{i,2}$. \label{p_phase_2}
        \item The \textsc{Alter}$(E)$ in Step~\ref{alg:expandm:s6} replaces each vertex $v$ on $P_{i,2}$ by $v.p$ to get path $P_{i+1}$. \label{p_phase_3}
    \end{enumerate}
    For any vertex $v$ on $P_i$ that is replaced by $v'$ in Phase~\ref{p_phase_1}, if $v'$ is not removed in Phase~\ref{p_phase_2}, then let $\overline{v}$ be the vertex replacing $v'$ in Phase~\ref{p_phase_3}, and call $\overline{v}$ the \emph{corresponding vertex} of $v$ in round $i+1$.
\end{framed}

\begin{lemma} \label{lem:expandm_path_length}
    Given graph $H$ and integer $b \ge (\log n)^{100}$. 
    Let $R = 20 \log b$. 
    W.p. $1 - 1/(\log n)^8$, for any shortest path $P$ in $H$ such that there are at most $b^7$ vertices on $P$, there are at most $2$ vertices on $P_R$, where $P_R$ is defined in \textsc{PassivePathConstruction}$(H, P)$.
\end{lemma}
\begin{proof}
    By Lemma~\ref{lem:active_passive_length}, at the beginning of round $r = 9 \log b - 1$, w.p. $1 - 1/(\log n)^9$, $P_r$ has length $10 \log b$. We shall condition on this happening and apply a union bound at the end of the proof. 
    
    In any round $i \ge r$, for any path $P_i$ with $|P_i| \ge 3$, consider the first vertex $v' = P_{i,1}(0)$ on $P_{i,1}$ (defined in \textsc{PassivePathConstruction}$(H, P)$). 
    If $\underline{v'}$ is a non-root or increases its level during round $i$, then by the first $4$ paragraphs in the proof of Lemma~\ref{lem:active_vertex_potential_bound}, it must be $\ell(\overline{v'}) \ge \ell(\underline{v'}) + 1$.
    If this is not the case, by Lemma~\ref{lem:connect_2}, there is an edge between $v'$ and $P_{i,1}(2)$ in the graph after Step~\ref{alg:expandm:s5}, which means the successor $P_{i,1}(1)$ of $v'$ on $P_{i,1}$ is removed in Phase~\ref{p_phase_2} of \textsc{PassivePathConstruction}$(H, P)$.
    Therefore, the number of vertices on $P_{i+1}$ is $1$ less than $P_i$ if $\ell(\overline{v'}) = \ell(\underline{v'})$ as the level of a corresponding vertex cannot be lower (by examining \textsc{Maxlink}$(V)$). 
    By Lemma~\ref{lem:expandm_max_level}, w.p. $1 - 1/(\log n)^9$, the level of any vertex cannot be higher than $100 \log\log n$.
    As $P_r$ has at most $10 \log b$ vertices, in round $R = 20 \log b \ge r + 10 \log b + 100 \log\log n$, the number of vertices on any $P_R$ is at most $2$.
    The lemma follows from a union bound. 
\end{proof}

\begin{lemma} \label{lem:truncate_path_length}
    Given graph $H$ and integer $b \ge (\log n)^{100}$. 
    For any shortest path $P$ in $H$ such that there are at most $b^7$ vertices on $P$, let $P_R$ be defined in \textsc{PassivePathConstruction}$(H, P)$. 
    Let path $P_{R+1}$ be defined by replacing each vertex $v$ on $P_R$ with $v.p$ in Step~\ref{alg:truncate:s3} of \textsc{Densify}$(H, b)$, and let path $P_{\text{final}}$ be defined by replacing each vertex $v$ on $P_{R+1}$ with $v.p$ in Step~\ref{alg:truncate:s6} of \textsc{Densify}$(H, b)$. 
    W.p. $1 - 1/(\log n)^8$, there are at most $2$ vertices on $P_{\text{final}}$.%\footnote{Formally, we need path $P_i$ to have at least $4$ active vertices in each round $i \in [9 \log b]$ to get a path $P_r$ to pass to \textsc{PassivePathConstruction}, but we can also just pass the path $P_i$ to \textsc{PassivePathConstruction} once it gets down to only $3$ active vertices. Similarly, we define the path $P_R$ to be the path $P_i$ in \textsc{PassivePathConstruction} when it contain at most $2$ vertices.}
\end{lemma}
\begin{proof}
    By Lemma~\ref{lem:expandm_path_length}, there are at most $2$ vertices on path $P_R$ w.p. $1 - 1/(\log n)^8$. 
    The \textsc{Alter}$(E(H))$ in Step~\ref{alg:truncate:s3} of \textsc{Densify}$(H, b)$ replaces vertices by their parents in each path, so $P_{R+1}$ has at most $2$ vertices. 
    Since Step~\ref{alg:truncate:s5} of \textsc{Densify}$(H, b)$ only alters edges locally without changing $E_{\text{close}}$, the path $P_{R+1}$ is unchanged. The  following \textsc{Alter}$(E_{\text{close}})$ in Step~\ref{alg:truncate:s6} replaces each vertex on $P_{R+1}$ by its parent. 
    As a result, $P_{\text{final}}$ has at most $2$ vertices.
\end{proof}

\subsection{Increasing Vertex Degrees} \label{subsec:increase}

In this section, we present the detailed algorithm \textsc{Increase}$(V, E, b)$ and prove that it increase the minimum degree of the current graph to at least $b$. 

\subsubsection{Algorithmic Framework of \textsc{Increase}} \label{subsubsec:increase_AF}

During the execution of \textsc{Increase}$(V, E, b)$, the subroutine \textsc{Densify}$(H,b)$ calls \textsc{Expand-Maxlink}$(H)$ (Step~\ref{alg:truncate:s1} of \textsc{Densify}$(H,b)$) for $R = 20 \log b$ times, where each \textsc{Expand-Maxlink}$(H)$ calls \textsc{Alter}$(E)$ for $2$ times (Step~\ref{alg:expandm:s1} and Step~\ref{alg:expandm:s6} of \textsc{Expand-Maxlink}$(H)$). 
\textsc{Increase}$(V, E, b)$ needs to use addition information of the execution, so we modify it as follows: 
\begin{itemize}
    \item In Step~\ref{alg:truncate:s1} of \textsc{Densify}$(H,b)$, pass the round id $i \in [0, R)$ to \textsc{Expand-Maxlink} as an additional parameter.
    \item Within \textsc{Expand-Maxlink}$(H, i)$, in Step~\ref{alg:expandm:s1}, pass the integer $2i$ to \textsc{Alter}$(E)$ as an additional parameter.
    \item Within \textsc{Expand-Maxlink}$(H, i)$, in Step~\ref{alg:expandm:s6}, pass the integer $2i+1$ to \textsc{Alter}$(E(V))$ as an additional parameter.
    \item For each \textsc{Alter}$(E, j)$ called within \textsc{Expand-Maxlink}$(H, i)$, execute $v.p_j = v.p$ for each vertex $v \in V(E)$.
    \item In the \textsc{Alter}$(E(H))$ within Step~\ref{alg:truncate:s3} of \textsc{Densify}$(H,b)$, execute $v.p_{2R} = v.p$ for each vertex $v \in V(H)$.
    \item In the \textsc{Alter}$(E_{\text{close}})$ within Step~\ref{alg:truncate:s6} of \textsc{Densify}$(H,b)$, execute $v.p_{2R+1} = v.p$ for each vertex $v \in V(E_{\text{close}})$. 
    The field $p_j$ for all $j \in [0, 2R+1]$ is global.
\end{itemize}

All $v.p_j$ for all $v \in V$ and all $j \in [0, 2R+1]$ can be written in a cell indexed by $(v, j)$ in the public RAM in $O(1)$ using the processor corresponding to $v$, and the total space (and work) used for this part is at most $|V| \cdot O(R) \le n/b^{10} \cdot O(\log n)$. 
The modified \textsc{Alter}$(E)$ still takes $O(|E|)$ work and $O(1)$ time, so Lemma~\ref{lem:truncate_time_work} still holds. 

\begin{definition} \label{def:iterative_parent}
    For any vertex $v \in V$, define $v.p^{(0)} \coloneqq v.p_0 = v$, 
    and for any integer $j \in [2R+1]$, define $v.p^{(j)} \coloneqq (v.p^{(j-1)}).p_j$.
\end{definition}

\begin{framed}
\noindent \textsc{Increase}$(V, E, b)$:
\begin{enumerate}
    \item Subgraph $H = (V, E') = \textsc{Build}(V, E, b)$. \label{alg:increase:s1}
    \item $(V, E_{\text{close}}) = \textsc{Densify}(H, b)$. \label{alg:increase:s2}
    \item For each vertex $v \in V$: assign a block of $b^9$ processors, which will be used as an indexed hash table $\mathcal{H}'(v)$ of size $b^9$. \label{alg:increase:s3}
    \item Choose a random hash function $h: [n] \to [b^9]$. For each vertex $v \in V$: compute $u \coloneqq v.p^{(2R+1)}$, then write $v$ into the $h(v)$-th cell of $\mathcal{H}'(u)$ and update $v.p = u$. \label{alg:increase:s4}
    %\item \textsc{Shortcut}$(V)$; \textsc{Alter}$(E_{\text{close}})$.  \label{alg:increase:s4.5}
    \item For each vertex $v \in V$: if the number of vertices in $\mathcal{H}'(v)$ is at least $2b$ then mark $v$ as a \emph{head}. \label{alg:increase:s5}
    \item For each non-loop edge $(v, w) \in E_{\text{close}}$: if $v$ is a head and $w$ is a non-head then $w.p = v$. \label{alg:increase:s6}
    \item \textsc{Shortcut}$(V)$.  \label{alg:increase:s7}
    \item For each vertex $v \in V$: sample $v$ as a \emph{leader} w.p. $1/2$. For each non-loop edge $(v, w) \in E_{\text{close}}$: if $v$ is a leader and $w$ is a non-leader then $w.p.p = v.p$.\label{alg:increase:s8}
    \item \textsc{Shortcut}$(V)$. \label{alg:increase:s9}
    \item \textsc{Alter}$(E)$. \Cliff{Insert edge sampling before Alter in Section 7 to get sublinear work such that we can run $\log n$ instances in parallel to get high probability.}  \label{alg:increase:s10}
\end{enumerate}
\end{framed}

%The hash table $\mathcal{H}'$ is to distinguish from the hash table $\mathcal{H}$ used in \S{\ref{subsec:build}} and \S{\ref{subsec:truncate}}. 

\begin{lemma} \label{lem:increase_work_time}
    W.p. $1 - 1/(\log n)^2$, \textsc{Increase}$(V, E, b)$ runs in $O(\log b)$ time and $O(m + n)$ work on an ARBITRARY CRCW PRAM.
\end{lemma}
\begin{proof}
    By Lemma~\ref{lem:build_time_work}, Step~\ref{alg:increase:s1} takes $O(\log b)$ time and $O(m + n)$ work w.p. $1 - n^{-8}$. 
    By Lemma~\ref{lem:truncate_time_work}, w.p. $1 - 1/(\log n)^3$, Step~\ref{alg:increase:s2} takes $O(\log b)$ time and $(m + n) / (\log n)^3$ work. 
    Using approximate compaction (Lemma~\ref{lem:apx_compaction}), Step~\ref{alg:increase:s3} takes $O(\log^* n)$ time and $O(|V| \cdot b^9)$ work (to notify each indexed processor the block it belongs to) w.p. $1 - n^{-9}$. 
    
    In Step~\ref{alg:increase:s4}, each vertex in $V$ chooses and stores a random number from $[b^9]$, giving function $h$, which takes $O(|V|)$ works and $O(1)$ time. 
    Next, the algorithm computes $u \coloneqq v.p^{(2R+1)}$ for each $v \in V$. Since each $v.p_j$ for each $v \in V$ and each $j \in [0, 2R+1]$ is written in a cell indexed by $(v, j)$ in the public RAM, we can compute $v.p^{(j)}$ ($j \ge 1$) (by Definition~\ref{def:iterative_parent}) for each $v$ using the processor of $v$ in $O(1)$ time, given that $v.p^{(j-1)}$ is stored in the private memory of $v$. Therefore, $u = v.p^{(2R+1)}$ is computed iteratively in $O(R) = O(\log b)$ time. The following hashing and parent update take $O(1)$ time and $O(|V|)$ work. 
    Therefore, Step~\ref{alg:increase:s4} takes $O(|V| \cdot \log b)$ works and $O(\log b)$ time.

    In Step~\ref{alg:increase:s5}, we use the binary tree technique in the proof of Lemma~\ref{lem:build_time_work} to count the number of occupied items in $\mathcal{H}(v)$ for each $v$ in $O(\log b)$ time, which takes $O(\log b)$ time and $O(|V| \cdot \log b)$ work.
    By the proof of Lemma~\ref{lem:truncate_time_work}, $O(|E_{\text{close}}|) \le O(m + n) / (\log n)^4$ w.p. $1 - 1/(\log n)^4$, so Step~\ref{alg:increase:s6} takes $O(m + n) / (\log n)^4$ work and $O(1)$ time w.p. $1 - 1/(\log n)^4$. 
    %Step~\ref{alg:increase:s7} takes $O(|V| + ||E_{\text{close}}||) \le O(m + n) / (\log n)^4$ work and $O(1)$ time w.p. $1 - 1/(\log n)^4$.
    Step~\ref{alg:increase:s8} takes $O(|V|) + O(m + n) / (\log n)^4$ work and $O(1)$ time w.p. $1 - 1/(\log n)^4$. 
    Step~\ref{alg:increase:s7} and Step~\ref{alg:increase:s9} take $O(|V|)$ work and $O(1)$ time.\footnote{Note that the total work performed from Step~\ref{alg:increase:s2} to Step~\ref{alg:increase:s9} is at most $O(m + n) / (\log n)^3 + O(|V| \cdot b^9) \le (m + n) / (\log n)^{2.5}$ w.p. $1 - 2/(\log n)^{3}$. In \S{\ref{sec:boosting} we will boost the success probability to $1 - 1/\text{poly}(n)$, which is obtained by running $\Theta(\log n)$ instances of Step~\ref{alg:increase:s2} to Step~\ref{alg:increase:s9} in parallel (along with other techniques), so the overall work will be $O(m+n)$.} \label{fn:increase_partial_work_sublinear}} 
    Step~\ref{alg:increase:s10} takes $O(m)$ work and $O(1)$ time.
    Summing up and by a union bound, the lemma follows.
\end{proof}

\subsubsection{Properties of \textsc{Increase}} \label{subsubsec:increase_property}

\begin{lemma} \label{lem:corresponded_non_roots}
    For any vertex $v \in V$, if $v.p^{(2R+1)} \ne v$ then $v$ is a non-root from Step~\ref{alg:increase:s3} to Step~\ref{alg:increase:s7} (inclusive) of \textsc{Increase}$(V, E, b)$.
\end{lemma}
\begin{proof}
    Let $j \in [0, 2R+1]$ be minimum integer such that $v.p^{(j)} = (v.p^{(j-1)}).p_j \ne v$ (Definition~\ref{def:iterative_parent}), 
    then $v.p$ is updated to a vertex $w \ne v$ within the call to \textsc{Expand-Maxlink}$(H)$ in round $\lfloor j/2 \rfloor$ of Step~\ref{alg:truncate:s1} of \textsc{Densify}$(H, b)$ (if $j \le 2R-1$) or within Step~\ref{alg:truncate:s5} of \textsc{Densify}$(H, b)$ (if $j = 2R+1$), all within Step~\ref{alg:increase:s2} of \textsc{Increase}$(V, E, b)$. 
    Recall that a non-root can never become a root during \textsc{Expand-Maxlink}$(H)$ or during Step~\ref{alg:truncate:s5} of \textsc{Densify}$(H, b)$ (\cite{liu2020connected}), so $v$ is a non-root right after Step~\ref{alg:increase:s2} of \textsc{Increase}$(V, E, b)$. 
    The parent updates in Step~\ref{alg:increase:s4} of \textsc{Increase}$(V, E, b)$ cannot make $v$ a root since $v.p^{(2R+1)} \ne v$. 
    The \textsc{Shortcut}$(V)$ in Step~\ref{alg:increase:s7} cannot make a non-root into a root. 
    Moreover, the parent updates in Step~\ref{alg:increase:s6} cannot make a non-root into a root either cause it only uses non-loop edges. 
    As a result, $v$ must be a non-root from Step~\ref{alg:increase:s3} to Step~\ref{alg:increase:s7} (inclusive) of \textsc{Increase}$(V, E, b)$.
\end{proof}

\begin{lemma} \label{lem:corresponding_roots}
    W.p. $1 - 1/(\log n)^8$, 
    for any vertex $u \in V$, if there exists vertex $v \in V$ such that $v.p^{(2R+1)} = u$ then $u$ is a root before Step~\ref{alg:increase:s5} of \textsc{Increase}$(V, E, b)$.
\end{lemma}
\begin{proof}
    By Definition~\ref{def:iterative_parent}, $u = (v.p^{2R}).p_{2R+1}$, where $p_{2R+1}$ is the parent filed in Step~\ref{alg:truncate:s6} of \textsc{Densify}$(H,b)$. 
    By Lemma~\ref{lem:truncate:flat}, w.p. $1 - 1/(\log n)^8$, at the end of \textsc{Densify}$(H, b)$, for any vertex $v \in V$, $v$ is either a root or a child of a root, which also holds in Step~\ref{alg:truncate:s6} of \textsc{Densify}$(H,b)$. 
    Since $u$ is a parent of another vertex at that time, $u$ must be a root. 
    So $u$ is a root right after \textsc{Densify}$(H, b)$ (Step~\ref{alg:increase:s2} of \textsc{Increase}$(V, E, b)$).
    Note that $u$ cannot update its parent to a vertex $w \ne u$ in Step~\ref{alg:increase:s4} of \textsc{Increase}$(V, E, b)$. 
    Because otherwise, $u.p^{(2R+1)} = w \ne u$, making $u$ a non-root after Step~\ref{alg:increase:s2} by Lemma~\ref{lem:corresponded_non_roots}, a contradiction. 
    Therefore, $u$ remains to be a root after Step~\ref{alg:increase:s4} of \textsc{Increase}$(V, E, b)$.
\end{proof}

\begin{lemma} \label{lem:increase_is_contraction}
    Given input graph $G = (V, E)$ and parameter $b \ge (\log n)^{100}$, at the end of \textsc{Increase}$(V, E, b)$, w.p. $1 - 1/(\log n)^7$, all vertices in $V$ must be roots or children of roots and both ends of any edge in $E$ must be roots.
\end{lemma}
\begin{proof}
    After Stage $1$, all trees are flat and both ends of any edge in $E$ are roots (Lemma~\ref{lem:alg:reduce_correctness}). 
    By Lemma~\ref{lem:truncate:flat},
    w.p. $1 - 1/(\log n)^8$, at the end of \textsc{Densify}$(H, b)$ (Step~\ref{alg:increase:s2} of \textsc{Increase}$(V, E, b)$), for any vertex $v \in V$, $v$ is either a root or a child of a root and both ends of any edge in $E_{\text{close}}$ are roots. We shall condition on this. 
    
    By Lemma~\ref{lem:corresponding_roots}, w.p. $1 - 1/(\log n)^8$, if $v.p^{(2R+1)} = v$ then $v$ is a root before Step~\ref{alg:increase:s4} of \textsc{Increase}$(V, E, b)$ and continues to be a root after this step.
    By Lemma~\ref{lem:corresponded_non_roots} and Lemma~\ref{lem:corresponding_roots}, w.p. $1 - 1/(\log n)^8$, if $v$ updates its parent to a vertex $u \ne v$ in Step~\ref{alg:increase:s4} of \textsc{Increase}$(V, E, b)$, then $v$ is a non-root and $u$ is a root before Step~\ref{alg:increase:s4}, which means $v$ is a child of a root and $u$ is a root before executing $v.p = u$ by the previous paragraph. 
    Therefore, each vertex in $V$ is either a root or a child of a root after Step~\ref{alg:increase:s4} of \textsc{Increase}$(V, E, b)$.
    
    In Step~\ref{alg:increase:s6} of \textsc{Increase}$(V, E, b)$, a non-head updates its parent to be a head if possible, which makes each vertex in $V$ a root, a child of a root, or a grandchild of a root. 
    The \textsc{Shortcut}$(V)$ in Step~\ref{alg:increase:s7} makes all vertices in $V$ roots or children of roots.
    
    In Step~\ref{alg:increase:s8}, $w.p.p$ is updated to $v.p$ if possible, which makes each vertex in $V$ a root, a child of a root, or a grandchild of a root, because $w.p$ and $v.p$ are both roots by the previous paragraph.  
    The \textsc{Shortcut}$(V)$ in Step~\ref{alg:increase:s9} makes all vertices in $V$ roots or children of roots. 
    
    Now both ends of any edges in $E$ are roots or children of roots. 
    So, the \textsc{Alter}$(E)$ in Step~\ref{alg:increase:s10} makes all ends of all edges in $E$ roots by the previous paragraph.
\end{proof}

\begin{lemma} \label{lem:corresponded_non_roots_final}
    W.p. $1 - 1/(\log n)^7$, 
    for any vertex $v \in V$, if $v.p^{(2R+1)} \ne v$ then $v$ is a non-root from Step~\ref{alg:increase:s3} till the end of \textsc{Increase}$(V, E, b)$.
\end{lemma}
\begin{proof}
    By Lemma~\ref{lem:corresponded_non_roots}, if $v.p^{(2R+1)} \ne v$ then $v$ is a non-root from Step~\ref{alg:increase:s3} to Step~\ref{alg:increase:s7} (inclusive) of \textsc{Increase}$(V, E, b)$. 
    By the last but two paragraph in the proof of Lemma~\ref{lem:increase_is_contraction}, w.p. $1 - 1/(\log n)^7$, all vertices in $V$ are roots or children of roots right after Step~\ref{alg:increase:s7}. 
    So, the parent updates in Step~\ref{alg:increase:s8} ($w.p.p = v.p$) can only update the parent of a root ($w.p$), which cannot turn a non-root into a root (which requires to update the parent of a non-root to itself).
\end{proof}

%Recall from Definition~\ref{def:component_symbol} that for any vertex $v$, $\mathcal{C}_{G}(v)$ is the component in $G$ that contains $v$.

\begin{lemma} \label{lem:H_small_component_done}
    Given input graph $G = (V, E)$ and parameter $b \ge (\log n)^{100}$, let $H = \textsc{Build}(V, E, b)$ be computed in Step~\ref{alg:increase:s1} of \textsc{Increase}$(V, E, b)$. 
    W.p. $1 - 1/(\log n)^6$, 
    for any component $\mathcal{C}$ in $H$ such that $|\mathcal{C}|$ (the number of vertices in $\mathcal{C}$) is less than $b^6$, at the end of \textsc{Increase}$(V, E, b)$, all vertices in $\mathcal{C}$ must be roots or children of roots and all edges in $E$ adjacent to these vertices must be loops.
\end{lemma}
\begin{proof}
    W.p. $1 - n^{-8}$, assuming the statement in Lemma~\ref{lem:H_component_size} always holds. 
    We get that $V(\mathcal{C}) = V(\mathcal{C}_G)$, where $\mathcal{C}_G$ is the component in $G$ induced on vertex set $V(\mathcal{C})$. 
    
    Note that for any shortest path $P$ in $\mathcal{C}$, there are at most $b^6$ vertices on $P$. 
    So by Lemma~\ref{lem:expandm_path_length}, w.p. $1 - 1/(\log n)^8$, any shortest path in $\mathcal{C}$ at the end of Step~\ref{alg:truncate:s1} in \textsc{Densify}$(H, b)$ has length at most $1$, so any vertex pair in $\mathcal{C}$ (with all edges in $E_{\text{close}}$) has distance at most $1$.
    By the proof of Lemma~\ref{lem:truncate:flat}, w.p. $1 - 1/(\log n)^8$, both ends of any edge in $E_{\text{close}}$ are roots. 
    Also note that the \textsc{Alter}$(E(H))$ in Step~\ref{alg:truncate:s3} of \textsc{Densify}$(H, b)$ never increases the diameter of any component since it only replaces edges on any shortest path. 
    So, the component in graph $(V(E_{\text{close}}), E_{\text{close}})$ corresponding to $\mathcal{C}$ has diameter at most $1$ and Theorem~\ref{thm:ltz_main} is applicable. 
    After Step~\ref{alg:truncate:s5} of \textsc{Densify}$(H, b)$, w.p. $1-1/(\log n)^9$, for any such component $\mathcal{C}$, 
    the vertices in $\mathcal{C}$ must be in the same tree (the algorithm in Theorem~\ref{thm:ltz_main} runs independently on each component, so the success probability is not for all vertices in a fix $\mathcal{C}$, but for all vertices in all such $\mathcal{C}$), and these vertices must be roots or children of roots.
    
    Consider the parent updates in Step~\ref{alg:increase:s4} of \textsc{Increase}$(V, E, b)$. 
    By Lemma~\ref{lem:corresponded_non_roots} and Lemma~\ref{lem:corresponding_roots}, w.p. $1 - 1/(\log n)^8$, if $v$ updates its parent to a vertex $u \ne v$ in Step~\ref{alg:increase:s4} of \textsc{Increase}$(V, E, b)$, then $v$ is a non-root and $u$ is a root before Step~\ref{alg:increase:s4}. 
    By an induction on $j$ from $0$ to $2R+1$, one can show that $u = v.p^{(2R+1)}$ is in the same component $\mathcal{C}$ of $v$. 
    By the previous paragraph, $u, v$ are in the same tree rooted at $u$ and $v.p = u$ before Step~\ref{alg:increase:s4}, so the parent update does not change the parent of $v$. 
    On the other hand, if $v.p^{(2R+1)} = v$ then $v$ is a root before Step~\ref{alg:increase:s4} of \textsc{Increase}$(V, E, b)$ w.p. $1 - 1/(\log n)^8$ by Lemma~\ref{lem:corresponding_roots} and continues to be a root. 
    Therefore, the tree corresponding to any such component $\mathcal{C}$ does not change during Step~\ref{alg:increase:s4}. 
    
    By Lemma~\ref{lem:truncate:flat}, w.p. $1 - 1/(\log n)^8$, the edge in $E_{\text{close}}$ are only adjacent to roots before Step~\ref{alg:increase:s4}. 
    By the previous paragraph, w.p. $1 - 1/(\log n)^8$, Step~\ref{alg:increase:s4} never turns a root into a non-root. 
    So, w.p. $1 - 1/(\log n)^7$, the edge in $E_{\text{close}}$ are only adjacent to roots before Step~\ref{alg:increase:s6}. 
    Since any vertex pair $v,w$ in $\mathcal{C}$ are in the same tree, if there is an edge in $E_{\text{close}}$ between $v,w$, the edge must be a loop. 
    So the parent updates in Step~\ref{alg:increase:s6} of \textsc{Increase}$(V, E, b)$ do not change the tree corresponding to any such component $\mathcal{C}$, because only non-loop edges are used. 
    
    The \textsc{Shortcut}$(v)$ in Step~\ref{alg:increase:s7} does not change the tree since all vertices in $\mathcal{C}$ must be roots or children of roots. 
    Step~\ref{alg:increase:s8} does not change the tree corresponding to any $\mathcal{C}$ because all edges between vertices in $\mathcal{C}$ remain to be loops. 
    The \textsc{Shortcut}$(v)$ in Step~\ref{alg:increase:s9} does not change these trees either. 
    We shall prove the lemma by contradiction based on Lemma~\ref{lem:increase_is_contraction}: w.p. $1 - 1/(\log n)^7$, both ends of any edge in $E$ are roots at the end of \textsc{Increase}$(V, E, b)$. 
    
    Assume for contradiction that there exists a component $\mathcal{C}$ in $H$ such that there exists an edge $(v, w) \in E$ (where the edge set $E$ here is already altered by the \textsc{Alter}$(E)$ in Step~\ref{alg:increase:s10} of \textsc{Increase}$(V, E, b)$) such that $v \ne w$ at the end of \textsc{Increase}$(V, E, b)$. 
    Then by the previous paragraph we know that both $v, w$ are roots. Since $(v, w) \in E$, we have that $v, w$ are in the same component $\mathcal{C}_G$.\footnote{We omit the proof that during the algorithm, $v, v.p$ are in the same component of $G$ and so do $v, w$ if edge $(v,w)$ presents in the current graph, which is by a standard induction by steps, e.g., see \cite{liu2020connected} or \cite{liu2022simple} for a proof.}
    By the second sentence of this proof, $V(\mathcal{C}) = V(\mathcal{C}_G)$. 
    So $v,w$ must be in the same component $\mathcal{C}$. 
    By the previous paragraph, $v,w$ must be within the same tree in the labeled digraph, contradicting with the fact that $v \ne w$ and they are both roots. 
    As a result, $v = w$ and thus all edges adjacent to vertices in $\mathcal{C}$ must be loops. 
    The lemma follows from a union bound.
\end{proof}

Lemma~\ref{lem:H_small_component_done} means that such small component (with size less than $b^6$) in $H$ can be \emph{ignored} in the rest of the entire algorithm as the connectivity computation on them is already \emph{finished}. 
In our presentation, we do not detect and delete all loops from the current graph, because although the loop-deletion does not affect the components that are already finished, it might decrease the spectral gap of other unfinished components, hurting the running time of our algorithm. 

Our edge sampling in Stage $3$ (\S{\ref{sec:stage3}}) will operate on all edge including loops. Since each small component is contracted to $1$ vertex,\footnote{Strictly speaking, such components are not fully contracted due to the existence of non-roots at the end of Stage $1$. The contraction will be done at the end of the algorithm, see footnote~\ref{fn:leaves_not_in_trees}.}
the execution of our algorithm after Stage $2$ does not affect such components, and it is equivalent to ignore all such components in the graph. 

\begin{lemma} \label{lem:H_big_component_large_degree}
    Given input graph $G = (V, E)$ and parameter $b \ge (\log n)^{100}$,
    w.p. $1 - 1/(\log n)^5$, 
    for any root $v \in V$ at the end of \textsc{Increase}$(V, E, b)$, the degree of $v$ in the current graph is at least $b$.
\end{lemma}
\begin{proof}
    By Lemma~\ref{lem:H_small_component_done} and the discussion following it, w.p. $1 - 1/(\log n)^6$, any vertex in a component of $H$ with size less than $b^6$ is ignored in the rest of the algorithm since the computation of its component is finished. 
    Therefore, we assume that any vertex $v \in V$ is in a component of $H$ with size at least $b^6$. 
    
    Suppose any shortest path in $\mathcal{C}_H(v)$ (Definition~\ref{def:component_symbol}) has less than $b^6$ vertices on it. 
    Then, using the same argument in the proof of Lemma~\ref{lem:H_small_component_done}, w.p. $1 - 1/(\log n)^6$, all shortest paths in $\mathcal{C}_H(v)$ have length at most $1$ at the end of Step~\ref{alg:truncate:s3} of \textsc{Densify}$(H, b)$, and the following Step~\ref{alg:truncate:s5} makes all vertices in $\mathcal{C}_H(v)$ the root of a tree or the children of the root of this same tree. The parent updates after Step~\ref{alg:truncate:s2} of \textsc{Increase}$(V, E, b)$ do not change the tree corresponding to $\mathcal{C}_H(v)$ by the same argument in the proof of Lemma~\ref{lem:H_small_component_done}. 
    Let $u$ be the root of the tree containing $v$. 
    By $|\mathcal{C}_H(v)| \ge b^6$, we get that the number of vertices $v \in V$ with $v.p = u$ is at least $b^6$ ($u.p = u$ also contributes $1$). 
    Since each vertex in $\mathcal{C}_H(v)$ has at least one adjacent edge (otherwise it is a singleton and cannot be presented in a component of size at least $b^6$), it must have at least one adjacent edge in $G$, so the number of edges in $G$ adjacent to vertices in $\mathcal{C}_H(v)$ is at least $b^6 / 2 \ge b$. 
    By the proof of Lemma~\ref{lem:increase_is_contraction}, all these edges are altered to be adjacent to root $u$ at the end of \textsc{Increase}$(V, E, b)$, so the degree of $u$ in the current graph is at least $b$.
    
    Suppose there exists a shortest path in $\mathcal{C}_H(v)$ with at least $b^6$ vertices on it. 
    Consider a shortest path $P$ in $\mathcal{C}_H(v)$ starting at $v$ such that $|P| \ge b^6 / 2$. 
    Such a path $P$ must exist, otherwise any vertex $u \in \mathcal{C}_H(v)$ can travel to $v$ then travel to any vertex $w \in \mathcal{C}_H(v)$ (by connectivity), giving a path of less than $b^6$ vertices, a contradiction. 
    Recall that for each integer $i \in [0, |P|-1]$, $P(i)$ is the $i$-th vertex on $P$. 
    For each integer $i \in [3, |P|-1]$, let path $P^{(i)}$ be the subpath from $P(0)$ to $P(i)$ (inclusively) using edges on $P$, and let $P^{(i)}_{\text{final}}$ be defined in Lemma~\ref{lem:truncate_path_length} given $P^{(i)}$ as the initial shortest path in $H$.
    By Lemma~\ref{lem:truncate_path_length}, $|P^{(i)}_{\text{final}}| \le 2$ w.p. $1 - 1/(\log n)^8$ for any shortest $P$ and any integer $i \in [3, |P|-1]$.
    
    Consider the (multi-)set of paths $\mathcal{P} \coloneqq \{P^{(i)}_{\text{final}} \mid i \in [3, |P|-1]\}$. 
    There are two cases. 
    \begin{enumerate}
        \item Suppose the set $\mathcal{P}$ contains at most $b^2$ distinct elements (paths). Note that each path in $\mathcal{P}$ has form $(P(0).p^{(2R+1)}, P(i).p^{(2R+1)})$ for some $i \in [3, |P|-1]$ by Definition~\ref{def:iterative_parent} and the fact that all such paths have length at most $1$. 
        So, the vertex set 
        $$V_P \coloneqq \{P(0).p^{(2R+1)}\} \cup \{P(i).p^{(2R+1)}) \mid i \in [3, |P|-1]\}$$ 
        contains at most $b^2 + 1$ distinct elements. 
        Observe that all vertices originally on path $P$ are distinct since $P$ is a shortest path. 
        Therefore, there exists a vertex $v^* \in V_P$ such that at least
        $$ \frac{|P| - 3}{b^2 + 1} \ge \frac{b^6 / 2 -3}{b^2 + 1} \ge b^2$$
        different $v'$ on $P$ satisfies $v'.p^{(2R+1)} = v^*$. 
        In Step~\ref{alg:increase:s4} of \textsc{Increase}$(V, E, b)$, each such vertex $v'$ writes itself to the $h(v')$-th cell of $\mathcal{H}'(v^*)$. 
        The probability that at most $2b$ vertices are in $\mathcal{H}'(v^*)$ is at most 
        \begin{equation*}
            \binom{b^9}{2b} \cdot \left( \frac{2b}{b^9} \right)^{b^2} \le b^{18 b} \cdot b^{-7 b^2} \le b^{-b} \le n^{-9} 
        \end{equation*}
        by $b \ge (\log n)^{100}$. 
        Therefore, by a union bound over all vertices, w.p. $1 - n^{-8}$, any of such vertex $v^*$ is marked as a head in  Step~\ref{alg:increase:s5} of \textsc{Increase}$(V, E, b)$. 
        
        If $v \notin V_P$ then there exists a vertex $u \ne v$ such that $u \in V_P$ and $v.p^{(2R+1)} = u$, which means $v$ is a non-root w.p. $1 - 1/(\log n)^7$ by Lemma~\ref{lem:corresponded_non_roots_final}, and the lemma holds for this case. 
        Now suppose $v \in V_P$, we have that $v = P(0).p^{(2R+1)}$ by the fact that $v$ is the first vertex on $P$. 
        There are two cases.
        \begin{enumerate}
            \item If $v = v^*$ then the number of vertices $w \in \mathcal{C}_H(v)$ with $w.p^{(2R+1)} = v$ is at least $b^2$, and all such $w$'s update their parents to $v$ in Step~\ref{alg:increase:s4} of \textsc{Increase}$(V, E, b)$. 
            Note that for any vertex $w$, $w.p^{(2R+1)}$ is unique, so the number of children of $v$ after (parallel execution of) Step~\ref{alg:increase:s4} is at least $b^2$. 
            The parent updates in Step~\ref{alg:increase:s6} can only increase the number of children of $v$ or make $v$ a non-root, because w.p. $1 - 1/(\log n)^7$, the edge in $E_{\text{close}}$ are only adjacent to roots before Step~\ref{alg:increase:s6} by the fourth paragraph in the proof of Lemma~\ref{lem:H_small_component_done}. 
            Suppose $v$ is a root after Step~\ref{alg:increase:s6}. 
            Step~\ref{alg:increase:s7} can only increase the number of children of $v$ since $v$ is a root.
            Step~\ref{alg:increase:s8} can only increase the number of children of $v$ or make $v$ a non-root, because Step~\ref{alg:increase:s7} makes all vertex in $V$ roots or children of roots such that Step~\ref{alg:increase:s8} can only update the parents of roots. 
            By the proof of Lemma~\ref{lem:corresponded_non_roots_final}, w.p. $1 - 1/(\log n)^7$, a non-root can never become a root, 
            thus if $v$ is a root at the end of \textsc{Increase}$(V, E, b)$, then the degree of $v$ is at least $b^2 /2 \ge b$ by the same argument in the second paragraph of this proof.
            
            \item If $v \ne v^*$ then the path from $v$ to $v^*$ is an element in $\mathcal{P}$ with path length $1$. 
            Therefore, $(v, v^*)$ is an edge in $E_{\text{close}}$ after Step~\ref{alg:increase:s2} of \textsc{Increase}$(V, E, b)$ and remains to be in $E_{\text{close}}$ right before Step~\ref{alg:increase:s6}. 
            If $v$ is a head then the number of distinct vertices in $\mathcal{H}'(v)$ is at least $2b$ by Step~\ref{alg:increase:s5}, giving at least $2b$ vertices whose parent is $v$ (a root) by Step~\ref{alg:increase:s4}, which continues to hold till Step~\ref{alg:increase:s7} because Step~\ref{alg:increase:s6} cannot make $v$ a non-root nor remove any child from $v$ (since any child of $v$ has no adjacent edge). Step~\ref{alg:increase:s8} can only increase the number of children of $v$ or make $v$ a non-root by the same argument in the previous case. 
            As a result, if $v$ is a root at the end of \textsc{Increase}$(V, E, b)$ then the degree of $v$ is at least $b$ by the same argument in the second paragraph of this proof.
            Else if $v$ is a non-head, since $v^*$ is marked as a head w.p. $1 - n^{-8}$, $v$ becomes a non-root in Step~\ref{alg:increase:s6} dues to edge $(v, v^*) \in E_{\text{close}}$ and remains so till the end of the algorithm, and the lemma holds for this case. \label{case:big_component_proof}
        \end{enumerate}
        
        \item Suppose the set $\mathcal{P}$ contains at least $b^2$ distinct elements (paths). 
        By definition, all paths (edges) in $\mathcal{P}$ starts with $P(0).p^{(2R+1)}$.
        Excluding the possible loop $(P(0).p^{(2R+1)}, P(0).p^{(2R+1)}) \in \mathcal{P}$, there are at least $b^2-1$ paths with length $1$ in $\mathcal{P}$, which correspond to at least $b^2-1$ different non-loop edges in $E_{\text{close}}$ at the end of Step~\ref{alg:increase:s2} of \textsc{Increase}$(V, E, b)$. 
        If $v$ is a non-root before Step~\ref{alg:increase:s8}, then w.p. $1 - 1/(\log n)^7$, $v$ is a non-root at the end of \textsc{Increase}$(V, E, b)$ by the proof of Lemma~\ref{lem:corresponded_non_roots_final}, and the lemma holds for this case. 
        Suppose $v$ is a root before Step~\ref{alg:increase:s8}, which gives $v= P(0).p^{(2R+1)}$ by the argument in the previous case. 
        There are two cases.
        \begin{enumerate}
            \item Suppose $v$ is sampled as a non-leader in Step~\ref{alg:increase:s8}. The number of different non-loops edges $(v, w)$ adjacent to $v$ in $E_{\text{close}}$ is at least $b^2 - 1$. So, by a Chernoff bound and $b \ge (\log n)^{100}$, w.p. $1 - n^{-8}$, there exists a leader $w$. Since $v.p = p$, Step~\ref{alg:increase:s8} updates the parent of the non-leader $v$ to $w.p$. 
            If $v \ne w.p$ then $v$ is a non-root at the end of \textsc{Increase}$(V, E, b)$ and the lemma holds for this case. 
            Else if $v = w.p$ before Step~\ref{alg:increase:s8}, then since $w$ has an adjacent edge $(v, w) \in E_{\text{close}}$, $w$ must be a root before Step~\ref{alg:increase:s6} w.p. $1 - 1/(\log n)^7$ (the fourth paragraph in the proof of Lemma~\ref{lem:H_small_component_done}). 
            Since $w \ne v$ and $w$ updates its parent to $v$ in Step~\ref{alg:increase:s6}, we obtain that $v$ is a head, and the lemma holds in this case by the proof in case~\ref{case:big_component_proof}.
            
            \item Suppose $v$ is sampled as a leader in Step~\ref{alg:increase:s8}. The number of different non-loops edges $(v, w)$ adjacent to $v$ in $E_{\text{close}}$ is at least $b^2 - 1$. By a Chernoff bound and $b \ge (\log n)^{100}$, w.p. $1 - n^{-8}$, there exist $2b$ non-leaders $w$ such that $(v,w) \in E_{\text{close}}$ and $w \ne v$. 
            So, after Step~\ref{alg:increase:s8}, the number of different vertices $w$ with $w.p.p = v.p = v$ is at least $2b$. 
            After the \textsc{Shortcut}$(V)$ in Step~\ref{alg:increase:s9}, the number of children of $v$ is at least $2b$ since $v$ is a root. 
            Since each vertex in $\mathcal{C}_H(v)$ has at least one adjacent edge (otherwise it is a singleton and cannot be presented in a component of size at least $b^6$), it must have at least one adjacent edge in $G$, so the number of edges in $G$ adjacent to these (at least) $2b$ different vertices (whose parent is $v$) is at least $2b / 2 = b$. 
            By the proof of Lemma~\ref{lem:increase_is_contraction}, all these edges are altered to be adjacent to root $v$ at the end of \textsc{Increase}$(V, E, b)$, so the degree of $v$ in the current graph is at least $b$.
        \end{enumerate}
    \end{enumerate}
    By a union bound over all the bad events, the lemma follows.
\end{proof}

\section{Stage 3: Connectivity on the Sampled Graph} \label{sec:stage3}

After Stage $2$, by Lemma~\ref{lem:H_big_component_large_degree}, w.p. $1 - 1/(\log n)^6$, we can assume $\deg(G) \ge b$, where $G$ is the current graph at the end of \textsc{Increase}$(V, E, b)$ and $b = (\log n)^{100}$. 

We also assume that the minimum spectral gap $\lambda$ over all connected components 
(henceforth also called the {\it component-wise spectral gap})
of the original input graph satisfies
\begin{equation*}
    \lambda \ge b^{-0.1} = (\log n)^{-10} .
\end{equation*}
This assumption will be removed in \S{\ref{sec:assumption}} to give a connected components algorithm for any graphs with any spectral gap and the algorithm does not need any prior knowledge on the spectral gap.

The following lemma states that Stages 1 and 2 of the algorithm does not decrease the component-wise spectral gap of the original graph.
\begin{lemma} \label{lem:spectral_gap_after_contraction}
    Let $\lambda$ be the component-wise spectral gap of the original input graph. 
    The component-wise spectral gap $\lambda_G$ of the current graph $G$ at the end of Stage $2$ satisfies that $\lambda_G \ge \lambda$.
\elaine{for every connected component?}
\end{lemma}
\begin{proof}
    By Lemma~\ref{lem:alg:reduce_correctness}, the algorithm in Stage $1$ (\textsc{Reduce}$(V, E, k)$) is a contraction algorithm since all vertices $v$ are contracted to vertex $v.p$, so the component-wise spectral gap of the graph at the end of Stage $1$ is at least $\lambda$ since a contraction can only increase the spectral gap of any graph (\cite{chung1997spectral} Chapter 1.4, Lemma 1.15). 
    
    The algorithm in Stage $2$ (\textsc{Increase}$(V, E, b)$) might create trees of height $2$ dues to the existence of non-roots created in Stage $1$. 
    By Lemma~\ref{lem:increase_is_contraction}, w.p. $1 - 1/(\log n)^8$, all edges in the current graph are only adjacent on roots and the algorithm never deletes edges nor add edges (but only replace them by replacing ends with their parents), so all vertices in a tree rooted at $v$ are contracted to $v$, giving a contraction algorithm \textsc{Increase}$(V, E, b)$ for Stage $2$ w.p. $1 - 1/(\log n)^8$. 
    As a result, the component-wise spectral gap $\lambda_G$ of the current graph $G$ at the end of Stage $2$ is at least $\lambda$ w.p. $1 - 1/(\log n)^8$ (this probability will be boosted to high-probability in \S{\ref{sec:boosting}} \Cliff{change citation}).\footnote{As an alternative proof, one can apply Cheeger's inequality for non-simple graph (\cite{chung1996laplacians}) and note that a contraction can only increase the graph conductance since it removes candidate cuts. This would give $\lambda_G \ge \lambda^3$ that suffices for our use.}
\end{proof}

The algorithm below, called
\textsc{SampleSolve}$(G = (V, E))$, computes the connected components of a graph given a promised component-wise spectral gap.
Essentially, the algorithm runs Stages 1 and 2, at the end of which all vertices have sufficiently large degree (except for the ones in tiny components which we can ignore). At this moment, 
we downsample the number of edges
by a $1/\poly\log n$ factor, and then run \cite{liu2020connected}.
In this section, we prove that the algorithm achieves $O(\log \log n)$ time  and linear expected work with good probability. We remove the known spectral gap assumption in \S~\ref{sec:assumption} and boost the guarantees to high probability in \S\ref{sec:boosting}.

\begin{framed}
\noindent \textsc{SampleSolve}$(G = (V, E))$:
\begin{enumerate}
    \item If $|V| \le n^{0.1}$ then removes parallel edges and loops from $E$ and compute the connected components of graph $(V, E)$ by Theorem~\ref{thm:ltz_main} and goto Step~\ref{alg:samplesolve:s4}. \label{alg:samplesolve:s1}
    \item Sample each edge in $E$ w.p. $1/(\log n)^7$ to obtain graph $G'$. \label{alg:samplesolve:s2}
    \item Compute the connected components of graph $G'$ by Theorem~\ref{thm:ltz_main}. \label{alg:samplesolve:s3}
    \item For each vertex $v$ in the original input graph: $v.p = v.p.p.p$. \label{alg:samplesolve:s4}
\end{enumerate}
\end{framed}

\begin{theorem} \label{thm:connectivity_with_assumption}
    There is an ARBITRARY CRCW PRAM algorithm that computes the connected components of any given graph with component-wise spectral gap $\lambda \ge (\log n)^{-10}$ in $O(\log \log n)$ time and $O(m) + \overline{O}(n)$ work w.p. $1 - 1/(\log n)$.
\end{theorem}
\begin{proof}
    We run algorithm \textsc{SampleSolve}$(G = (V, E))$ where $G$ is the current graph at the end of Stage $2$. 
    
    By Lemma~\ref{lem:alg:reduce_correctness}, all vertices at the end of Stage $1$ are roots or children of roots, as required by Stage $2$. 
    By Lemma~\ref{lem:increase_is_contraction}, w.p. $1 - 1/(\log n)^8$, all vertices in $V$ must be roots or children of roots, where $V$ is the set of roots at the end of Stage $1$. Therefore, all vertices in the input graph are in trees of height at most $2$. 
    Within \textsc{SampleSolve}$(G)$ in Stage $3$, either in its Step~\ref{alg:samplesolve:s1} or Step~\ref{alg:samplesolve:s3}, the contraction algorithm in Theorem~\ref{thm:ltz_main} makes each root at the beginning of Stage $3$ a root or a child of root at the end of Step~\ref{alg:samplesolve:s3}, all vertices in the input graph are in trees of height at most $3$. 
    Step~\ref{alg:samplesolve:s4} of \textsc{SampleSolve}$(G)$ makes all trees flat at the end of Stage $3$. 
    
    By Lemma~\ref{lem:polylog_shrink_reduce_all}, Stage $1$ runs in $O(\log \log n)$ time and $O(m) + \overline{O}(n)$ work w.p. $1 - n^{-2}$, and $|V| \le n/b^{10}$ as required by Stage $2$. 
    By Lemma~\ref{lem:increase_work_time}, w.p. $1 - 1/(\log n)^2$, Stage $2$ takes $O(\log b) = O(\log \log n)$ time and $O(m + n)$ work. 
    So, the first two stages take $O(\log \log n)$ time and $O(m) + \overline{O}(n)$ work w.p. $1 - 2/(\log n)^2$. 
    
    If $|V| \le n^{0.1}$ at the beginning of Stage $3$ (which can be detected in $O(\log^* n)$ time and $O(n)$ work w.p. $1 - n^{-9}$ by approximate compaction), then in Step~\ref{alg:samplesolve:s1} of \textsc{SampleSolve}$(G)$ we remove all loops from $E$ in $O(1)$ time and $O(m)$ work; 
    next, we remove all parallel edges from $E$ by PRAM perfect hashing on all edges of $E$ in $O(m)$ work and $O(\log^* n)$ time w.p. $1 - n^{-9}$ (\cite{DBLP:conf/focs/GilMV91}). 
    This results in a simple graph with at most $|V|^2 \le n^{0.2}$ edges, so the work of running the algorithm in Theorem~\ref{thm:ltz_main} on this graph is at most $O(n^{0.2} \cdot \log n) \le O(n)$ w.p. $1 - 1/(\log n)^9$, giving $O(m + n)$ total work for Step~\ref{alg:samplesolve:s1} w.p. $1 - 1/(\log n)^8$. 
    By Lemma~\ref{lem:spectral_gap_after_contraction} and the assumption on $\lambda$, we have $\lambda_G \ge (\log n)^{-10}$ at the beginning of Step~\ref{alg:samplesolve:s1}. 
    Let $d$ be the maximum diameter of any component in $G$, then $d \le O(\log n / \lambda) \le (\log n)^{12}$. 
    Removing loops and parallel edges from $G$ does not affect the diameter of $G$, nor its components. 
    Therefore, the algorithm in Theorem~\ref{thm:ltz_main} computes all connected components of $G$ in $O(\log d + \log\log n) \le O(\log \log n)$ time w.p. $1 - 1/(\log n)^9$. 
    Step~\ref{alg:samplesolve:s4} takes $O(n)$ work and $O(1)$ time. 
    As a result, if $|V| \le n^{0.1}$ at the beginning of Stage $3$, then Stage $3$ takes $O(m + n)$ work and $O(\log\log n)$ time w.p. $1 - 1/(\log n)^7$. 
    
    Suppose $|V| > n^{0.1}$ at the beginning of Stage $3$. 
    By Lemma~\ref{lem:H_big_component_large_degree}, w.p. $1 - 1/(\log n)^5$, we have $\deg(G) \ge b$ at the beginning of Step~\ref{alg:samplesolve:s2} of \textsc{SampleSolve}$(G)$. 
    By $|V| > n^{0.1}$ and Corollary~\ref{cor:sample_preserve} \Cliff{TBD}, w.p. $1 - n^{-0.1c} \ge 1 - n^{-9}$ (choosing $c$ large enough), the spectral gap $\lambda$ of any component $\mathcal{C}$ in $G$ and the spectral gap $\lambda'$ of this component $\mathcal{C}'$ in $G'$ after sampling satisfy
    \begin{equation*}
        \left| \lambda - \lambda' \right| \le C' \cdot \sqrt{\frac{\ln n}{(1/(\log n)^7) \cdot \deg(G)}} \le b^{-0.4}
    \end{equation*}
    by $1 / (\log n)^7 \cdot \deg(G) \ge b / (\log n)^7 \ge C \ln n$ and $b \ge (\log n)^{100}$. 
    So $\lambda' \ge \lambda - b^{-0.4} \ge b^{-0.2}$ by $\lambda \ge b^{-0.1}$ from the previous paragraph. 
    Since $\lambda' > 0$, the component $\mathcal{C}$ stays connected after sampling, and thus $V(\mathcal{C}) = V(\mathcal{C'})$. 
    Therefore, it remains to prove that each such component $\mathcal{C}'$ contracts to $1$ vertex in Step~\ref{alg:samplesolve:s3} of \textsc{SampleSolve}$(G)$. 
    Let $d'$ be the diameter of component $\mathcal{C'}$. 
    We have $d' \le O(\log n / \lambda') \le O(b^{0.2} \log n) \le (\log n)^{30}$ by $b = (\log n)^{100}$. 
    Therefore, the algorithm in Theorem~\ref{thm:ltz_main} computes all connected components of $G'$ in $O(\log d' + \log\log n) \le O(\log \log n)$ time w.p. $1 - 1/(\log n)^9$, giving all connected components of $G$. 
    Now we compute the work. If $|E| \le n^{0.2}$ then Step~\ref{alg:samplesolve:s3} takes $O(n)$ work w.p. $1 - 1/(\log n)^8$ as discussed in the previous paragraph. 
    Now suppose $|E| \ge n^{0.2}$. 
    By a Chernoff bound over all edges, after Step~\ref{alg:samplesolve:s2} of \textsc{SampleSolve}$(G)$, the number of edges in $G'$ is at most $m / (\log n)^6$ w.p. $1 - n^{-9}$. So the algorithm in Theorem~\ref{thm:ltz_main} takes $O(m + n) / (\log n)$ work in Step~\ref{alg:samplesolve:s3} w.p. $1 - 1/(\log n)^9$. 
    As a result, Stage $3$ takes $O(m + n)$ work and $O(\log\log n)$ time w.p. $1 - 1/(\log n)^4$ for this case. 
    
    Summing up all works and running time over all $3$ stages and by a union bound, the theorem follows.
\end{proof}

The total work of the algorithm will be improved to $O(m+n)$, and 
the success probability will be boosted to $1 - n^{-c}$ for arbitrary constant $c > 0$ in \S{\ref{sec:boosting}}. 
The assumption on the component-wise spectral gap of the input graph will be removed in \S{\ref{sec:assumption}}, resulting in a connected components algorithm for any given graph with any component-wise spectral gap $\lambda$ over all components that runs in $O(\log (1 / \lambda) + \log \log n)$ time and $O(m) + \overline{O}(n)$ work w.p. $1 - n^{-c}$ for arbitrary constant $c > 0$.

\section{The Overall Algorithm: Removing the Assumption on Spectral Gap} \label{sec:assumption}

In this section, we remove the assumption that the component-wise spectral gap $\lambda$ over all components of the input graph is at least $(\log n)^{-10}$, giving an algorithm that runs in $O(\log(1 / \lambda) + \log\log n)$ time.

\subsection{Algorithmic Framework} \label{subsec:overll_AF}

The overall algorithm starts with running \textsc{Reduce}$(V, E, k)$ with $k = 10^6 \log\log n$ on the input graph $(V, E)$ for preprocessing (same as Stage $1$) such that the number of vertices in the current graph is at most $n/b^{10}$ w.h.p., where $b = (\log n)^{100}$ initially. 
Next, we interweave the algorithms \textsc{Filter} from Stage $1$, \textsc{Increase} from Stage $2$, and \textsc{SampleSolve} from Stage $3$, such that in each \emph{phase} $i$ (defined below in the pseudocode):
\begin{itemize}
    \item If the component-wise spectral gap $\lambda$ (unknown to the algorithm) is at least $b^{-0.1}$, then \textsc{Increase} makes $\deg(G) \ge b$ in $O(\log b)$ time, and \textsc{SampleSolve} (with other subroutines) solves the problem in $O(\log b)$ time. Otherwise, $\lambda < b^{-0.1}$. 
    \item Revert the state of algorithm \textsc{Filter} to the previous phase. 
    Update $b \leftarrow b^{1.1}$. 
    Run $O(1.1^i \log\log n)$ rounds of \textsc{Filter} (and algorithm \textsc{Matching} with other subroutines) to reduce the number of vertices in the current graph to $n/b^{10}$.
\end{itemize}

\begin{framed}
\noindent \textsc{Connectivity}$(G = (V, E))$:
\begin{enumerate}
    \item Initialize the labeled digraph: for each vertex $v \in V$: $v.p = v$. \label{alg:connectivity:s1}
    \item Run \textsc{Reduce}$(V, E, k = 10^6 \log\log n)$ to obtain the current graph $G'$. \label{alg:connectivity:s2}
    \item Sample each edge of $G'$ w.p. $1 / (\log n)^7$ independently to get random subgraph $H_1$. Sample each edge of $G'$ w.p. $1 / (\log n)^7$ independently to get random subgraph $H_2$. The sampled graphs are copied out so that $G'$ is unchanged. \label{alg:connectivity:s3}
    \item Copy the edges of $G'$ to create an edge set $E_{\text{filter}}$. \label{alg:connectivity:s4}
    \item For \emph{phase} $i$ from $0$ to $10 \log\log n$: $E_{\text{filter}} = \textsc{Interweave}(G', H_1, H_2, E_{\text{filter}}, i)$, if $E_{\text{filter}} = \emptyset$ then goto Step~\ref{alg:connectivity:s6}.\label{alg:connectivity:s5}
    %\item For each vertex $v \in V$: $v.p = v.p.p$. 
    \item \textsc{Shortcut}$(V)$. \label{alg:connectivity:s6}
\end{enumerate}
\end{framed}

In the above, $H_1$ is the edges 
sampled to create the sparsified graph before we run  
Theorem~\ref{thm:ltz_main}~\cite{liu2020connected},
and $H_2$ 
is an independently sampled set of edges to help us create the skeleton
graph more work-efficiently in Stage $2$ of the algorithm. 
In particular, now that we need to 
run the algorithm for each assumption of the component-wise $\lambda$,  
we can no longer afford to spend $O(m)$ work for each assumption --
for this reason, we sparsify the graph before creating
the skeleton graphs such that we can 
create the skeleton graph in 
$O(m/\log^2 n)$ work.
See \S~\ref{subsec:reduce_work_increase}
for more details.

We now describe the \textsc{Interweave} algorithm which essentially searches for the right spectral gap by squaring the guess in every iteration.

\begin{framed}
\noindent \textsc{Interweave}$(G', H_1, H_2, E_{\text{filter}}, i)$:
\begin{enumerate}
    \item Set $b = (\log n)^{100 \cdot 1.1^i}$. \label{alg:interweave:s1}
    \item $\textsc{Increase}(G', H_1, H_2, b)$ ---  
see \S~\ref{subsec:reduce_work_increase} for more details. 
\label{alg:interweave:s2}
    \item Run $20 \log b$ rounds of $H_1 = \textsc{Expand-Maxlink}(H_1)$ then run $10^4 \log\log n$ rounds of the algorithm in Theorem~\ref{thm:ltz_main} on $H_1$. \textsc{Alter}$(H_1)$. \label{alg:interweave:s3}
    \item If all edges in $E(H_1)$ are loops then call \textsc{Remain}$(G', H_1)$ and return $\emptyset$. \label{alg:interweave:s4}
    \item Revert the labeled digraph and graph $H_1$ to their states in Step~\ref{alg:interweave:s1}. \label{alg:interweave:s5}
    \item Repeat for $10^6 \cdot 1.1^i \log\log n$ rounds: \textsc{Matching}$(E_{\text{filter}})$, \textsc{Alter}$(E_{\text{filter}})$, delete each edge from $E_{\text{filter}}$ w.p. $10^{-4}$. \label{alg:interweave:s6}
    \item Repeat for $i + 2 \log\log n$ times: \textsc{Shortcut}$(V(G'))$. \label{alg:interweave:s7}
    \item Edge set $E' = \{(v, w) \in E(G') \mid v.p \in V(G') \backslash V(E_{\text{filter}})\}$. \textsc{Alter}$(E')$. \label{alg:interweave:s8} 
    \item Repeat for $10^6 \cdot 1.1^i \log\log n$ rounds: \textsc{Matching}$(E')$, \textsc{Shortcut}$(V(G'))$, \textsc{Alter}$(E')$. \label{alg:interweave:s9}
    \item \textsc{Reverse}$(V(E_{\text{filter}}), E(H_2))$. \label{alg:interweave:s10}
    \item Return $E_{\text{filter}}$.
\end{enumerate}
\end{framed}

\begin{framed}
    \noindent \textsc{Remain}$(G', H_1)$:
    \begin{enumerate}
        \item \textsc{Alter}$(E(G'))$. \label{alg:remain:s1}
        \item Edge set $E_{\text{remain}} = E(G') \backslash E(H_1)$. \label{alg:remain:s2}
        \item Remove loops and parallel edges from $E_{\text{remain}}$. \label{alg:remain:s3}
        \item Compute the connected components of graph $(V(E_{\text{remain}}), E_{\text{remain}})$ by Theorem~\ref{thm:ltz_main}. \label{alg:remain:s4}
    \end{enumerate}
\end{framed}

At the beginning of each phase $i$, the algorithm assumes that $\lambda \ge b^{-0.1}$ where $b = (\log n)^{100 \cdot 1.1^i}$. 
If the assumption is correct, then the connected components of the input graph is correctly computed in $O(\log b) = O(\log(1/\lambda))$ time. 
Otherwise, $b$ is updated to $b^{1.1}$ and the number of vertices is reduces to $n/b^{10}$ for the next phase under the assumption that $\lambda \ge b^{-0.1}$. 

The algorithm eventually succeeds in a phase under the correct assumption. 
Since the running time $O(\log b)$ in each phase is geometrically increasing, the total running time is asymptotically the running time of the phase under the correct assumption, which is $O(\log(1/\lambda))$ as desired. 

An additional technique issue is to reduce the work in each phase from $O(m + n)$ to $O(m+n) /\log n$ such that the total work over all $O(\log\log n)$ phases are $O(m+n)$, which will be presented later in this section: 
In \S{\ref{subsec:reduce_work_increase}} we reduce the work of \textsc{Increase}$(G', H_1, H_2, b)$ by modifying the algorithm \textsc{Increase}$(V, E, b)$ in \S{\ref{subsec:increase}};
in \S{\ref{subsec:reduce_work_edge_set}} we reduce the work of computing the edge set $E'$ in Step~\ref{alg:interweave:s8} of \textsc{Interweave}$(G', H_1, H_2, E_{\text{filter}}, i)$.

\subsection{Basic Properties of \textsc{Connectivity}} \label{subse:basic_connectivity_overlall}

We give some basic properties of the algorithm \textsc{Connectivity}, which will be used later.

\begin{lemma} \label{lem:interweave_flat_each_phase}
    During the execution of \textsc{Connectivity}$(G)$, at the beginning of each phase in Step~\ref{alg:connectivity:s5}, all vertices in $V(G')$ are roots or children of roots.
\end{lemma}
\begin{proof}
    By Lemma~\ref{lem:alg:reduce_correctness}, after Step~\ref{alg:connectivity:s2}, all trees are flat and edges in $E(G')$ are only adjacent to roots, which also holds for edge set $E_{\text{filter}}$ (Step~\ref{alg:connectivity:s4}). 
    By an induction on phase $i$ from $0$ to $10 \log\log n$, it suffices to prove that for any phase, if at the beginning of the phase all vertices in $V(G')$ are roots or children of roots, then this still holds at the beginning of the next phase. 
    We do not consider the case of $E_{\text{filter}} = \emptyset$ in Step~\ref{alg:connectivity:s5} because there will be no next phase. 
    Consider the \textsc{Interweave}$(G', H_1, H_2, E_{\text{filter}}, i)$ in this phase. 
    Step~\ref{alg:interweave:s2} to Step~\ref{alg:interweave:s5} can be ignored because \textsc{Interweave}$(G', H_1, H_2, E_{\text{filter}}, i)$ does not return $\emptyset$ and thus the labeled digraph (as well as $H_1$) are the same as in Step~\ref{alg:interweave:s1} dues to Step~\ref{alg:interweave:s5}. 
    By an induction on rounds in Step~\ref{alg:interweave:s6} and Lemma~\ref{lem:constant_shrink_correctness}, the edges in $E_{\text{filter}}$ are only adjacent to roots at the beginning of each round of Step~\ref{alg:interweave:s6}, and any vertex $v \in V(G')$ has distance at most $10^6 \cdot 1.1^i \log\log n$ from the root of the tree containing $v$ because such distance increases by at most $1$ in each round. 
    Each application of \textsc{Shortcut}$(V(G'))$ in Step~\ref{alg:interweave:s7} reduces such distance by a factor of at least $3/2$ if the distance is at least $2$. 
    Therefore, after $i + 2\log\log n \ge \log_{3/2} (10^6 \cdot 1.1^i \log\log n)$ applications of \textsc{Shortcut}$(V(G'))$, all vertices in $V(G')$ are roots or children of roots after Step~\ref{alg:interweave:s7}. 
    The \textsc{Alter}$(E')$ in Step~\ref{alg:interweave:s8} moves all edges only adjacent to roots since $E' \subseteq E(G')$. 
    Finally, by an induction on rounds in Step~\ref{alg:interweave:s9} and Lemma~\ref{lem:constant_shrink_correctness}, the edges in $E'$ are only adjacent to roots at the beginning of each round of Step~\ref{alg:interweave:s9} and all vertices in $V(G')$ are roots or children of roots at the beginning of each round of Step~\ref{alg:interweave:s9} dues to the \textsc{Shortcut}$(V(G'))$ in each round. 
    This finishes the total induction and gives the lemma.
\end{proof}

\begin{definition} \label{def:upper_degree}
    Let graph $G'$ be computed in Step~\ref{alg:connectivity:s2} of \textsc{Connectivity}$(G)$. 
    Given a time $t$ in the execution of \textsc{Connectivity}$(G)$, 
    for each vertex $v \in V(G')$, 
    define $\overline{\deg}_{G'}(v) \coloneqq \left|\left\{(u, w) \in E(G') \mid u.p = v \right\}\right|$, 
    and define graph $\overline{G'}$ as the graph $G'$ after executing \textsc{Alter}$(E(G'))$, where the parent points are given in the labeled digraph at time $t$.
\end{definition}

In our later analysis, we will frequently use $\overline{G'}$, because we cannot afford executing \textsc{Alter}$(E(G'))$ as it would take $\Theta(m)$ work. 
Recall from \S{\ref{sec:stage1}} that an active root is a vertex that is a root with adjacent edges to another root in the current graph. 
Our current graph would be $G'$ if we executed \textsc{Alter}$(E(G'))$ at the end of each phase, which is unaffordable, but we can still define the \emph{current graph} to be $\overline{G'}$ to be consistent.\footnote{We will show how to detect active roots in desired work later in the proof of Lemma~\ref{lem:new_increase_work}.}
As a result, a root $u$ that is non-active must be in a tree containing all the vertices in the component of $u$ in $G'$, and thus $u$ and all its descendants can be ignored in the rest of the execution of the algorithm because the computation of this component is finished.

\begin{lemma} \label{lem:interweave_roots_shrink_each_phase}
    During the execution of \textsc{Connectivity}$(G)$, w.p. $1 - n^{-7}$, for any integer $i \in [0, 10\log\log n)$, at the beginning of each phase $i$ in Step~\ref{alg:connectivity:s5}, the number of active roots in $V(G')$ is at most $n/b^{10}$, where $b = (\log n)^{100 \cdot 1.1^i}$.
\end{lemma}
\begin{proof}
    The proof is by an induction on phase $i \in [0, 10\log\log n)$. 
    At the beginning of phase $0$, $|V(G')| \le n/(\log n)^{1000}$ guaranteed by Stage $1$. 
    Suppose the number of active roots is at most $n /b^{10}$ at the beginning of phase $i$ where $b = (\log n)^{100 \cdot 1.1^i}$, we will show that the number of active roots decreases to $n/b^{11}$ at the end of phase $i$, giving the lemma.
    
    Consider the \textsc{Interweave}$(G', H_1, H_2, E_{\text{filter}},i)$ in phase $i$. 
    We shall ignore Step~\ref{alg:interweave:s2} to Step~\ref{alg:interweave:s4} in the analysis, because the labeled digraph is reverted to Step~\ref{alg:interweave:s1} (by Step~\ref{alg:interweave:s4}) if \textsc{Interweave}$(G', H_1, H_2, E_{\text{filter}},i)$ does not return in Step~\ref{alg:interweave:s3}.
    
    By the proof of Lemma~\ref{lem:extract_time_work_helper_V}, each round in Step~\ref{alg:interweave:s6} decreases the number of roots in $V(E_{\text{filter}})$ by a factor of $0.001$ w.p. $1 - n^{-9}$. 
    There are at most $(\log n)^3$ rounds in Step~\ref{alg:interweave:s6} by $i \le 10 \log\log n$.
    So by a union bound, w.p. $1 - n^{-8}$, at the end of Step~\ref{alg:interweave:s6} we have the number of roots in $V(E_{\text{filter}})$ is at most
    \begin{equation*}
        n / b^{10} \cdot 0.999^{10^6 \cdot 1.1^i \log\log n} \le n/b^{12} .
    \end{equation*}
    
    After Step~\ref{alg:interweave:s8}, the edge set $E'$ contains all edges adjacent to all roots in $V(G') \backslash V(E_{\text{filter}})$, because all vertices in $V(G')$ are roots or children of roots by Lemma~\ref{lem:interweave_flat_each_phase} and \textsc{Alter}$(E')$ moves all edges to roots. 
    By the proof of Lemma~\ref{lem:polylog_shrink_reduce}, the number of roots such that the root has an adjacent edge to another root in $V(E')$ decreases by a factor of $0.001$ w.p. $1 - n^{-9}$ in each round of Step~\ref{alg:interweave:s9}. 
    By a union bound over all the at most $(\log n)^3$ rounds in Step~\ref{alg:interweave:s9}, the number of such roots (with an adjacent edge to another root) in $V(E')$ at the end of Step~\ref{alg:interweave:s9} is at most $n/b^{12}$ w.p. $1 - n^{-8}$. 
    By the same proof as in the proof of Lemma~\ref{lem:polylog_shrink_reduce}, 
    the \textsc{Reverse}$(V(E_{\text{filter}}), E(H_2))$ in Step~\ref{alg:interweave:s10} makes the number of active roots at most $2n/b^{12} \le n/b^{11}$ w.p. $1 - n^{-7}$.
\end{proof}

\subsection{Reducing the Work of \textsc{Increase}} \label{subsec:reduce_work_increase}

In this section, we modify Step~\ref{alg:increase:s1} and Step~\ref{alg:increase:s10} in the \textsc{Increase}$(V, E, b)$ in \S{\ref{subsec:increase}} such that the new \textsc{Increase}$(G', H_1, H_2, b)$ takes $(m + n) / (\log n)^2$ work w.p. $1 - 2/(\log n)^2$.

\begin{framed}
\noindent \textsc{Increase}$(G', H_1, H_2, b)$:
\begin{enumerate}
    \item Edge set $E_H = \textsc{SparseBuild}(G', H_2, b)$, vertex set $V = V(G')$, graph $H = (V(E_H), E_H)$. \label{alg:new_increase:s1}
    \item Run Step~\ref{alg:increase:s2} to Step~\ref{alg:increase:s9} in the \textsc{Increase}$(V, E, b)$ in \S{\ref{subsec:increase}}. \label{alg:new_increase:s2}
    \item \textsc{Alter}$(E(H_1))$. \label{alg:new_increase:s3}
\end{enumerate}
\end{framed}

\subsubsection{Building Sparse Skeleton Graph} \label{subsubsec:sparse_build}

Recall that the \textsc{Build}$(V, E, b)$ in Step~\ref{alg:increase:s1} of the \textsc{Increase}$(V, E, b)$ in \S{\ref{subsec:increase}} takes $O(m+n)$ work (Lemma~\ref{lem:build_time_work}), because the sparse subgraph $H$ is computed based on the input graph $(V, E)$ with $n$ vertices and $m$ edges. 
The computation of $H$ relies on classifying vertices as high-degree and low-degree. 
One of the key observation is that such classification is still accurate w.h.p. even if we sample the graph $\overline{G'}$ to get random subgraph $H_2$ and classify vertices using edges of $H_2$: a high-degree vertex in $\overline{G'}$ is still high-degree w.h.p. in $H_2$ (with different classification threshold). 
The final subgraph $H$ should also contain all edges adjacent to low-degree vertices. 
Another key observation is that the number of edges adjacent to low-degree vertices is bounded, and we present a new technique to extract all these edges from the original input graph in $n / \poly(\log n)$ work in \S{\ref{subsec:reduce_work_edge_set}}.

\begin{framed}
\noindent \textsc{SparseBuild}$(G', H_2, b)$:
\begin{enumerate}
    \item For each active root $v \in V(G')$: assign a block of $b^9$ processors, which will be used as an indexed hash table $\mathcal{H}(v)$ of size $b^9$. \label{alg:sparsebuild:s1}
    %\item Choose a random hash function $h: [m] \to [b^9]$. For each edge $(u, v) \in E(H_2)$: write $(u,v)$ into the $h(u)$-th cell of $\mathcal{H}(v)$. \label{alg:sparsebuild:s2}
    \item For each edge $(u, v) \in E(H_2)$: write $(u,v)$ into a random indexed cell of $\mathcal{H}(v)$. \label{alg:sparsebuild:s2}
    \item For each active root $v \in V(G')$: if the number of items in $\mathcal{H}(v)$ is more than $b^8$ then mark $v$ as \emph{high}, else mark $v$ as \emph{low}. \label{alg:sparsebuild:s3}
    \item $E' = \{(u, v) \in E(G') \mid u.p \text{~is~a~low~vertex}\}$. \textsc{Alter}$(E')$. \label{alg:sparsebuild:s4}
    \item Return $E' \cup E(H_2)$.  \label{alg:sparsebuild:s5}
\end{enumerate}
\end{framed}

Note that in Step~\ref{alg:sparsebuild:s2} of the \textsc{SparseBuild}$(G', H_2, b)$, we do not use hashing by vertex index as in \textsc{Build}$(V, E, b)$. 
This is because we cannot afford to removing parallel edges and loops in sublinear work as in \textsc{Build}$(V, E, b)$. 
As a result, we only obtain bounds on the vertex degrees (including parallel edges and loops) instead of bounds on the number of neighbors for each vertex as in \S{\ref{subsec:build}}. 
We will show that under this case, \textsc{Increase}$(G', H_1, H_2, b)$ still makes every root high-degree.

\begin{lemma} \label{lem:sparse_build:high_vertex}
    During the execution of \textsc{Connectivity}$(G)$,
    w.p. $1 - n^{-9}$, for any vertex $v \in V(G')$ that is marked as high in any call to \textsc{SparseBuild}$(G', H_2, b)$, $\overline{\deg}_{G'}(v) \ge b^8$ in that call;
    and for any vertex $v \in V(G')$ that is marked as low, $\overline{\deg}_{G'}(v) \le b^9$ in that call.
\end{lemma}
\begin{proof}
    For any high vertex $v$, it must have degree at least $b^8$ in graph $H_2$ by Step~\ref{alg:sparsebuild:s2} and Step~\ref{alg:sparsebuild:s3}. Thus, the degree of $v$ in $\overline{G'}$ (Definition~\ref{def:upper_degree}) is at least $b^8$ since $H_2$ is a subgraph of $\overline{G'}$ dues to the \textsc{Alter}$(E(H_2))$ within \textsc{Reverse}$(V(E_{\text{filter}}), E(H_2))$ in Step~\ref{alg:interweave:s9} of \textsc{Connectivity}$(G)$ and Lemma~\ref{lem:interweave_flat_each_phase}, giving the first part of the lemma by Definition~\ref{def:upper_degree}. 
    
    A vertex $v$ is low if at most $b^8$ cells in $\mathcal{H}(v)$ are occupied (Step~\ref{alg:sparsebuild:s3}). 
    If $\overline{\deg}_{G'}(v) \ge b^9$, then since each edge of $\overline{G'}$ is sampled into $H_2$ w.p. $1/(\log n)^7$ (by Step~\ref{alg:connectivity:s3} of \textsc{Connectivity}$(G)$, Definition~\ref{def:upper_degree}, and the fact that $\overline{G'}$ performs the same \textsc{Alter} as $H_2$ by Step~\ref{alg:interweave:s9} of \textsc{Connectivity}$(G)$ and Lemma~\ref{lem:interweave_flat_each_phase}), by a Chernoff bound, we get $\deg_{H_2}(v) \ge b^9 / (\log n)^8 \ge b^{8.5}$ w.p. $1 - n^{-11}$ by $b \ge (\log n)^{100}$.
    Then $v$ is low w.p. at most
    \begin{equation*}
        \binom{b^{8.5}}{b^8} \cdot \left( \frac{b^8}{b^{8.5}} \right)^{b^9} \le b^{8.5 b^8} \cdot b^{-0.5 b^9} \le b^{-b^8} \le n^{-11} ,
    \end{equation*}
    where we used $b \ge (\log n)^{100}$. So, w.p. $1 - n^{-10}$ all such vertices are marked as high in one call to \textsc{SparseBuild}$(G', H_2, b)$ by a union bound. During the execution of \textsc{Connectivity}$(G)$, \textsc{SparseBuild}$(G', H_2, b)$ is called for at most $10 \log\log n$ times (Step~\ref{alg:connectivity:s5} of \textsc{Connectivity}$(G)$). 
    So, all such vertices are marked as high w.p. $1 - n^{-9}$ in any call to \textsc{SparseBuild}$(G', H_2, b)$ by a union bound, and we shall condition on this happening. 
    Since any vertex $v$ with $\overline{\deg}_{G'}(v) \ge b^9$ is marked as high, we have that any low vertex $v$ must have $\overline{\deg}_{G'}(v) \le b^9$. This gives the second part of the lemma.
\end{proof}

\begin{lemma} \label{lem:sparse_H_component_size}
    Let graph $H$ be computed in Step~\ref{alg:new_increase:s1} of \textsc{Increase}$(G', H_1, H_2, b)$. 
    W.p. $1 - n^{-8}$, for any active root $v \in V(G')$, either $\mathcal{C}_{\overline{G'}}(v) = \mathcal{C}_{H}(v)$ or $\sum_{w \in V(\mathcal{C}_{H}(v))} \deg_{H}(w) \ge b^6$.
\end{lemma}
\begin{proof}
    The first part of the proof is an analog of the proof of Lemma~\ref{lem:H_component_size}. 
    Since the edge set $E_H$ of graph $H$ contains $E'$, the set of edges adjacent to low vertices (Step~\ref{alg:sparsebuild:s4} and Step~\ref{alg:sparsebuild:s5} of \textsc{SparseBuild}$(G', H_2, b)$), we can run breath-first searches on $\mathcal{C}_{\overline{G'}}(v)$ and $\mathcal{C}_{H}(v)$ simultaneously to get that if there is no high vertex in $\mathcal{C}_{H}(v)$, then $\mathcal{C}_{H}(v) = \mathcal{C}_{\overline{G'}}(v)$. 
    
    On the other hand, if there is a high vertex $w$ in $\mathcal{C}_{H}(v)$, then by Lemma~\ref{lem:sparse_build:high_vertex}, w.p. $1 - n^{-9}$, $\overline{\deg}_{G'}(w) \ge b^8$. 
    By a Chernoff bound, w.p. $1 - n^{-9}$, $\deg_{H_2}(w) \ge b^6$ because $H_2$ is a random subgraph of $\overline{G'}$ with edge sampling probability $1/(\log n)^7$ by Step~\ref{alg:connectivity:s3} of \textsc{Connectivity}$(G)$, Definition~\ref{def:upper_degree}, and the fact that $\overline{G'}$ performs the same \textsc{Alter} as $H_2$ by Step~\ref{alg:interweave:s9} of \textsc{Connectivity}$(G)$ and Lemma~\ref{lem:interweave_flat_each_phase}. 
    So, w.p. $1 - n^{-8}$, $\deg_{H}(w) \ge b^6$ since $E(H_2) \subseteq E_H$ (Step~\ref{alg:sparsebuild:s5} of \textsc{SparseBuild}$(G', H_2, b)$), 
    giving the second part of the lemma.
\end{proof}

\subsubsection{Increasing Vertex Degrees by Sparse Skeleton Graph} \label{subsubsec:increase_sparse_H}

\begin{lemma} \label{lem:sparse_H_small_component_done}
    Let graph $H$ be computed in Step~\ref{alg:new_increase:s1} of \textsc{Increase}$(G', H_1, H_2, b)$. 
    W.p. $1 - 1/(\log n)^6$, 
    for any component $\mathcal{C}$ in $H$ such that $\sum_{w \in V(\mathcal{C})} \deg_{H}(w) < b^6$, at the end of \textsc{Increase}$(G', H_1, H_2, b)$, all vertices in $\mathcal{C}$ must be roots or children of roots and all edges in $E(H_1)$ adjacent to these vertices must be loops.
\end{lemma}
\begin{proof}
    By Lemma~\ref{lem:sparse_H_component_size}, w.p. $1 - n^{-8}$, $V(\mathcal{C}) = V(\mathcal{C}_{\overline{G'}})$, where $\mathcal{C}_{\overline{G'}}$ is the component in $\overline{G'}$ induced on vertex set $V(\mathcal{C})$.
    Next, we prove that the number of vertices in $\mathcal{C}$ is less than $b^6$. 
    If $\mathcal{C}$ contains only $1$ vertex then the claim holds. 
    Otherwise, since each $v \in V(\mathcal{C})$ has at least $1$ adjacent edge in $H$ (otherwise $v$ is a singleton), we get $|\mathcal{C}| < b^6$ by $\sum_{w \in V(\mathcal{C})} \deg_{H}(w) < b^6$. 
    Therefore, for any shortest path $P$ in $\mathcal{C}$, there are at most $b^6$ vertices on $P$. 
    
    Using the same proof as of Lemma~\ref{lem:H_small_component_done}, w.p. $1 - 1/(\log n)^7$, at the end of \textsc{Increase}$(G', H_1, H_2, b)$, we obtain that all vertices in $\mathcal{C}$ must be roots or children of roots and they must be in the same tree. 
    The \textsc{Alter}$(E(H_1))$ in Step~\ref{alg:new_increase:s3} of \textsc{Increase}$(G', H_1, H_2, b)$ moves all edges adjacent to vertices in $\mathcal{C}$ to the root, which must be loops.
\end{proof}

By the same discussion after Lemma~\ref{lem:H_small_component_done}, such small component (with total degree less than $b^6$) in $H$ can be ignored in the rest of the entire algorithm as the connectivity computation on them is already finished. 

\begin{lemma} \label{lem:sparse_H_big_component_large_degree}
    During the execution of \textsc{Connectivity}$(G)$, w.p. $1 - 1/(\log n)^5$, for any integer $i \in [0, 10 \log\log n)$ and any active root $v \in V(G')$ at the end of Step~\ref{alg:interweave:s2} within \textsc{Interweave}$(G', H_1, H_2, E_{\text{filter}}, i)$, it must be 
    $\overline{\deg}_{G'}(v) \ge b = (\log n)^{100 \cdot 1.1^i}$.
    %\begin{equation} \label{eq:lem:sparse_H_big_component_large_degree}
    %    \left|\left\{(u, w) \in E(G') \mid u.p = v \right\}\right| \ge b = (\log n)^{100 \cdot 1.1^i} . 
    %\end{equation}
\end{lemma}
\begin{proof}
    By Lemma~\ref{lem:sparse_H_small_component_done}, w.p. $1 - 1/(\log n)^6$, any vertex in a component of $H$ with total degree less than $b^6$ is ignored in the rest of the algorithm since the computation of its component is finished. 
    Therefore, we assume that any vertex $v \in V$ is in a component with total degree at least $b^6$. 
    
    If any shortest path in $\mathcal{C}_H(v)$ has less than $b^6$ vertices on it, then since Step~\ref{alg:new_increase:s2} of \textsc{Increase}$(G', H_1, H_2, b)$ runs Step~\ref{alg:increase:s2} to Step~\ref{alg:increase:s9} in the \textsc{Increase}$(V, E, b)$ in \S{\ref{subsec:increase}}, we can use the same proof in the second paragraph in the proof of Lemma~\ref{lem:H_big_component_large_degree} to get that w.p. $1 - 1/(\log n)^6$, all vertices in $\mathcal{C}_H(v)$ contracts to $v$. 
    So, the degree of any root $v \in V(G')$ has $\overline{\deg}_{G'}(v) \ge b^6 / 2$ because the total degree of $\mathcal{C}_H(v)$ is at least $b^6$ and all edges in $\mathcal{C}_H(v)$ would be adjacent to $v$ if we executed \textsc{Alter}$(E(G'))$ (Definition~\ref{def:upper_degree}). 
    
    If there exists a shortest path in $\mathcal{C}_H(v)$ with at least $b^6$ vertices on it, then using the same proof in the proof of Lemma~\ref{lem:H_big_component_large_degree} for this case, we get that w.p. $1 - 1/(\log n)^6$, the degree of $v$ in $G'$ would be at least $b$ if we executed \textsc{Alter}$(E(G'))$. 
    Using the same argument in the previous paragraph, $\overline{\deg}_{G'}(v) \ge b$ holds w.p. $1 - 1/(\log n)^6$ for this case.
    
    By a union bound over all $10 \log\log n$ phases, $\overline{\deg}_{G'}(v) \ge b$ holds w.p. $1 - 1/(\log n)^5$ for all integers $i \in [0, 10 \log\log n)$ and all roots $v \in V(G')$.
\end{proof}

\begin{lemma} \label{lem:new_increase_work}
    If the edge set $E'$ in Step~\ref{alg:sparsebuild:s4} of \textsc{SparseBuild}$(G', H_2, b)$ can be computed in $n / (\log n)^7$ work, then \textsc{Increase}$(G', H_1, H_2, b)$ takes $(m + n) / (\log n)^2$ work w.p. $1 - 1/(\log n)^2$.
\end{lemma}
\begin{proof}
    By a Chernoff bound, w.p. $1 - n^{-9}$, we have $|E(H_1)| \le m/(\log n)^6$ and $|E(H_2)| \le m/(\log n)^6$ by Step~\ref{alg:connectivity:s3} of \textsc{Connectivity}$(G)$, and we shall condition on this happening. 
    
    In Step~\ref{alg:new_increase:s1} of \textsc{Increase}$(G', H_1, H_2, b)$, consider the work of \textsc{SparseBuild}$(G', H_2, b)$. 
    Detecting whether a vertex $v \in V(G')$ is an active root can be done by checking whether all of the following $3$ conditions hold:
    (\romannumeral1) $v = v.p$, 
    (\romannumeral2) in the previous phase, at the end of Step~\ref{alg:interweave:s6} in \textsc{Interweave}$(G', H_1, H_2, E_{\text{filter}},i)$, any edge in $E_{\text{filter}}$ does not connect $v$ to another root,
    (\romannumeral3) in the previous phase, at the end of Step~\ref{alg:interweave:s9} in \textsc{Interweave}$(G', H_1, H_2, E_{\text{filter}},i)$, any edge in $E'$ does not connect $v$ to another root. 
    Condition~(\romannumeral1) can be checked by all vertices in $V(G')$ in $O(1)$ time, taking $|V(G')| \le n/(\log n)^{10}$ work.
    Condition~(\romannumeral2) can be checked by adding $O(|E_{\text{filter}}|)$ work into Step~\ref{alg:interweave:s6} in \textsc{Interweave}$(G', H_1, H_2, E_{\text{filter}},i)$ then write a flag into the private memory of $v$, which does not change the asymptotic work of that step.
    Condition~(\romannumeral3) can be checked by adding $O(|E'|)$ work into Step~\ref{alg:interweave:s9} in \textsc{Interweave}$(G', H_1, H_2, E_{\text{filter}},i)$ then write a flag into the private memory of $v$, which does not change the asymptotic work of that step.
    All in all, the total work (attributes to \textsc{SparseBuild}$(G', H_2, b)$) of letting each vertex know whether itself is an active root is $O(n/(\log n)^{10})$.
    
    By Lemma~\ref{lem:interweave_roots_shrink_each_phase}, w.p. $1 - n^{-7}$, the number of active roots is at most $n/b^{10}$, 
    so Step~\ref{alg:sparsebuild:s1} of \textsc{SparseBuild}$(G', H_2, b)$ takes $O(n/b) \le n/(\log n)^9$ work (here we assume all active roots are in an index array of length at most $2n/b^{10}$, which is guaranteed by approximate compaction at the end of each phase, see Lemma~\ref{lem:each_phase_running_time} for details).
    Step~\ref{alg:sparsebuild:s2} takes $O(|E(H_2)|) \le O(m/(\log n)^6)$ work. 
    Step~\ref{alg:sparsebuild:s3} takes $O(n/b^{10} \cdot b^9 \cdot O(\log b^9)) \le n/(\log n)^9$ work by a binary tree counting argument.
    Computing edge set $E'$ in Step~\ref{alg:sparsebuild:s4} takes $n / (\log n)^7$ work by the precondition of the lemma. 
    By Lemma~\ref{lem:sparse_build:high_vertex}, w.p. $1 - n^{-9}$, for any vertex $v \in V(G')$ that is marked as low, $\overline{\deg}_{G'}(v) \le b^9$, giving at most $b^9$ edges in $E'$ adjacent on $v$ by Definition~\ref{def:upper_degree}. 
    Since non-active roots can be ignored by the discussion after Definition~\ref{def:upper_degree}, and the number of active roots is at most $n/b^{10}$, we have that $|E'| \le n/b$ w.p. $1 - n^{-6}$. 
    So, the following \textsc{Alter}$(E')$ takes $O(|E'|) \le O(n/b)$ work. 
    Summing up, Step~\ref{alg:new_increase:s1} of \textsc{Increase}$(G', H_1, H_2, b)$ takes $O(m+n)/(\log n)^6$ work.
    
    Consider Step~\ref{alg:new_increase:s2} of \textsc{Increase}$(G', H_1, H_2, b)$. 
    Firstly, we show that $|E(H)| \le (m + n) / (\log n)^5$ w.p. $1 - n^{-8}$. 
    By the previous paragraph, $|E'| \le n/b$ w.p. $1 - n^{-6}$.
    Since $|E(H_2)| \le m/(\log n)^6$, we have that $|E_H| = |E(H)| \le (m + n) / (\log n)^5$ w.p. $1 - n^{-5}$ by Step~\ref{alg:sparsebuild:s5} of \textsc{SparseBuild}$(G', H_2, b)$. 
    Secondly, the vertex set $V$ in Step~\ref{alg:increase:s2} to Step~\ref{alg:increase:s9} of \textsc{Increase}$(V, E, b)$ is replaced by the set of active roots, which is already described in the second paragraph of this proof and by running approximate compaction on all active roots at the end of each phase. By Lemma~\ref{lem:interweave_roots_shrink_each_phase}, the number of active roots is at most $n/b^{10}$. 
    Since the requirements on $H$ and $|V|$ in \S{\ref{subsec:increase}} are satisfied, (the proof of) Lemma~\ref{lem:increase_work_time} still holds. 
    In the \textsc{Increase}$(V, E, b)$ in \S{\ref{subsec:increase}}, its Step~\ref{alg:increase:s2} to Step~\ref{alg:increase:s9} perform $(m + n) / (\log n)^{2.5}$ work w.p. $1 - 2/(\log n)^3$ (see footnote~\ref{fn:increase_partial_work_sublinear}).
    As a result, Step~\ref{alg:new_increase:s2} of \textsc{Increase}$(G', H_1, H_2, b)$ takes $(m + n) / (\log n)^{2.5}$ work w.p. $1 - 2/(\log n)^3$.
    
    Step~\ref{alg:new_increase:s3} of \textsc{Increase}$(G', H_1, H_2, b)$ takes $O(|E(H_1)|) \le O(m/(\log n)^6)$ work. 
    Summing up and by a union bound, the lemma follows.
\end{proof}

The computation of the edge set $E'$ in Step~\ref{alg:sparsebuild:s4} of \textsc{SparseBuild}$(G', H_2, b)$ will be given in \S\ref{subsec:reduce_work_edge_set}.

\subsection{Reducing the Work of Computing Edge Set $E'$} \label{subsec:reduce_work_edge_set}

In this section, we present a new technique that work-efficiently computes 
\begin{itemize}
    \item The edge set $E'$ in Step~\ref{alg:sparsebuild:s4} of \textsc{SparseBuild}$(G', H_2, b)$.
    \item The edge set $E'$ in Step~\ref{alg:interweave:s8} of \textsc{Interweave}$(G', H_1, H_2, E_{\text{filter}}, i)$.
\end{itemize}

\subsubsection{Building the Auxiliary Array} \label{subsubsec:auxiliary_array}

The computation of edges set $E'$ mentioned above relies on an auxiliary array that is computed and stored in the public RAM at the end of Stage $1$ where the current graph is $G'$ (Step~\ref{alg:connectivity:s2} of \textsc{Connectivity}$(G)$). 
We need the following tool.

\begin{lemma}[Section 1.3 in \cite{hagerup1992waste}] \label{lem:padded_sort}
    The \emph{padded sort} is to sort an array of $m$ items such that after sorting, the items are in sorted order in an array of length at most $2m$. 
    If the item values to be sorted are integers in $[m]$, then $m$ items can be padded sorted in $O(\log \log m)$ time and $O(m)$ work w.p. $1 - m^{-20}$. 
\end{lemma}

In our application, we padded sort all $m$ edges of $G'$ by their first ends (see below). 
If $m \le n^{0.3}$, then the non-singletons in $G'$ are at most $2n^{0.3}$, which can be detected in $O(n)$ work. 
Since we only need to compute the connected components of the non-singletons, one can call Theorem~\ref{thm:ltz_main} for $\log n$ times in parallel to compute the connected components of $G'$ in $O(\log d + \log\log n) \le O(\log(1/\lambda) + \log\log n)$ time and $O(m + 2n^{0.3}) \cdot \log n \cdot O(\log n) \le O(n)$ work w.p. $1 - n^{-9}$. 
Therefore, we assume $m \ge n^{0.3}$, which means the padded sort succeeds w.p. at least $1 - n^{-6}$. 
Meanwhile, the padded sort takes $O(\log\log n)$ time by $m \le n^c$.

\begin{framed}
\noindent \textsc{BuildAuxiliary}$(G')$:
\begin{enumerate}
    \item Padded sort all edges of $G'$ by their first ends and store the output array $\mathcal{A}$. \label{alg:buildauxiliary:s1}
    \item For each item $(v, w) \in \mathcal{A}$ in parallel: if the first end of the predecessor of $(v, w)$ in $\mathcal{A}$ is not $v$ then write the index of $(v, w)$ in $\mathcal{A}'$ into $v.l$ in the private memory of (the processor corresponding to) $v$. \label{alg:buildauxiliary:s2}
    \item For each item $(v, w) \in \mathcal{A}$ in parallel: if the first end of the successor of $(v, w)$ in $\mathcal{A}$ is not $v$ then write the index of $(v, w)$ in $\mathcal{A}$ into $v.r$ in the private memory of $v$. \label{alg:buildauxiliary:s3}
    \item For each vertex $v \in V(G')$: if $v.r - v.l \ge (\log n)^{90}$ then (approximate) compact the subarray from index $v.l$ to index $v.r$ (inclusive) in $\mathcal{A}$. \label{alg:buildauxiliary:s4}
\end{enumerate}
\end{framed}

Padded sorting (Step~\ref{alg:buildauxiliary:s1}) create empty cells between items. 
In Step~\ref{alg:buildauxiliary:s4}, if the subarray from index $v.l$ to $v.r$ satisfies $v.r - v.l \ge (\log n)^{90}$, then the non-empty cells are moved to the beginning of this subarray while the original subarray is erased. 
The array $\mathcal{A}$ is called the \emph{auxiliary array} that is stored in the public RAM and will be used to efficiently compute the edge set $E'$ during the algorithm. 
Each cell of $\mathcal{A}$ that contains an edge $e$ also stores the processor id corresponding to $e$ for the purpose of mapping each edge in $\mathcal{A}$ back to the corresponding edge processor, which can be done by copying the edge set and storing the corresponding processor id in the item (edge) to be sorted before padded sorting. 

\begin{lemma} \label{lem:build_auxiliary_time_work}
    \textsc{BuildAuxiliary}$(G')$ takes $O(\log\log n)$ time and $O(m)$ work w.p. $1 - n^{-5}$.
\end{lemma}
\begin{proof}
    By Lemma~\ref{lem:padded_sort} and the discussion after it, Step~\ref{alg:buildauxiliary:s1} takes $O(m)$ work and $O(\log\log n)$ time w.p. $1 - n^{-6}$. Since $\mathcal{A}$ is indexed, we assign a processor for each of its cell.
    Step~\ref{alg:buildauxiliary:s2} and Step~\ref{alg:buildauxiliary:s3} take $O(m)$ work and $O(1)$ time by $|\mathcal{A}| \le 2m$.
    Step~\ref{alg:buildauxiliary:s4} takes $O(\log^* n)$ time and $O(v.r - v.l)$ work w.p. $1 - 1/2^{(\log n)^{90/25}} \ge 1 - n^{-9}$ for each $v \in V(G')$ by Lemma~\ref{lem:apx_compaction}.
    W.p. $1 - n^{-8}$, the total work of Step~\ref{alg:buildauxiliary:s4} is at most $\sum_{v \in V(G')} O(v.r - v.l) = O(|\mathcal{A}|) \le O(m)$ by telescoping. 
    The lemma follows.
\end{proof}

We will use the following property of the auxiliary array $\mathcal{A}$.
\begin{lemma} \label{lem:build_auxiliary_property}
    Let $\mathcal{A}$ be constructed in \textsc{BuildAuxiliary}$(G')$. 
    W.p. $1 - n^{-4}$, for any vertex $v \in V(G')$, all edges adjacent on $v$ in $E(G')$ are stored in the subarray of $\mathcal{A}$ that starts from index $v.l$ and ends at index $v.l + v.s$, where $v.s \le \max\{(\log n)^{90}, 2 \deg_{G'}(v)\}$.
\end{lemma}
\begin{proof}
    By the proof of Lemma~\ref{lem:build_auxiliary_time_work}, w.p. $1 - n^{-5}$, the padded sorting in Step~\ref{alg:buildauxiliary:s1} and all approximate compactions in Step~\ref{alg:buildauxiliary:s4} succeed. 
    For any vertex $v \in V(G')$, if $v.r - v.l < (\log n)^{90}$, then all adjacent edges are stored in the subarray of $\mathcal{A}$ that starts from index $v.l$ and ends at index $v.r \le v.l + (\log n)^{90}$, so the lemma holds for this case.
    Else if $v.r - v.l \ge (\log n)^{90}$, the all the $\deg_{G'}(v)$ edges adjacent to $v$ are compacted to the subarray of $\mathcal{A}$ that starts from index $v.l$ and ends at index $v.r \le v.l + v.s \le v.l + 2 \deg_{G'}(v)$ by Definition~\ref{def:apx_compaction}, and the lemma holds for this case.
\end{proof}

\subsubsection{Computing the Edge Set in \textsc{SparseBuild}} \label{subsubsec:edge_set_sparse_build}

In this section we show that the edge set $E'$ in Step~\ref{alg:sparsebuild:s4} of \textsc{SparseBuild}$(G', H_2, b)$ can be computed in $n / (\log n)^7$ work and $O(\log b)$ time, giving the precondition of Lemma~\ref{lem:new_increase_work}.

\begin{lemma} \label{lem:parent_low_its_low}
    At the beginning of Step~\ref{alg:sparsebuild:s4} of \textsc{SparseBuild}$(G', H_2, b)$, w.p. $1 - n^{-9}$, for any vertex $u \in V(G')$, if $u.p$ is a low vertex, then $\deg_{G'}(u) \le b^9$.
\end{lemma}
\begin{proof}
    By Lemma~\ref{lem:sparse_build:high_vertex}, w.p. $1 - n^{-9}$, $\overline{\deg}_{G'}(u.p) \le b^9$, and we shall condition on this happening. 
    Assume for contradiction that $\deg_{G'}(u) > b^9$, then the degree of $u.p$ is more than $b^9$ if we perform \textsc{Alter}$(E(G'))$, giving $\overline{\deg}_{G'}(u.p) > b^9$, a contradiction.
\end{proof}

\begin{lemma} \label{lem:edge_set_sparse_build_time_work}
    W.p. $1 - n^{-3}$, the edge set $E'$ in Step~\ref{alg:sparsebuild:s4} of \textsc{SparseBuild}$(G', H_2, b)$ can be computed in $O(\log b)$ time and $n / (\log n)^7$ work.
\end{lemma}
\begin{proof}
    In Step~\ref{alg:sparsebuild:s4} of \textsc{SparseBuild}$(G', H_2, b)$, each vertex $u \in V(G')$ checks whether $u.p$ is a low vertex using $O(|V(G')|) \le O(n / (\log n)^{10})$ work in total. 
    If $u.p$ is low, then we do the following. 
    The processor corresponding to $u$ awakens the processor corresponding to the cell indexed at $u.l$ of the auxiliary array $\mathcal{A}$ in round $0$, then in the next round (round $1$) the processor corresponding to the cell at $u.l$ awakens the cells at $u.l$ and $u.l + 1$, then in round $2$, the processor at $u.l$ awakens $u.l$ and $u.l+1$, and the processor at $u.l+1$ awakens $u.l+2$ and $u.l+3$, and so on. That is, in round $j$, each of the $2^{j-1}$ awakened processor at $u.l + k$ awakens the processors at $u.l + 2k$ and $u.l + 2k+1$. 
    We continue this process until the $(9 \log b + 1)$-th round or a cell containing an edge not starting at $u$ is awakened. 
    Since in each round, a processor only needs to awaken $2$ processors, the running time for this part is $O(\log b)$. 
    After $9 \log b + 1$ rounds, all cells indexed from $u.l$ to $u.l + 2^{9 \log b + 1} = u.l + 2 b^9$ such that the cell contains an edge starting at $u$ are awakened by $u$ (if a cell containing an edge not starting at $u$ is awakened by $u$, we mark this cell as not awakened by $u$ to avoid the case that a cell is awakened by more than $1$ vertices). 
    
    By Lemma~\ref{lem:parent_low_its_low}, w.p. $1 - n^{-9}$, $\deg_{G'}(u) \le b^9$ if $u.p$ is a low vertex. 
    By Lemma~\ref{lem:build_auxiliary_property}, w.p. $1 - n^{-4}$, all edges adjacent on $u$ in $E(G')$ are stored in the subarray of $\mathcal{A}$ that starts from index $u.l$ and ends at index $u.l + u.s$, where 
    $$u.s \le \max\{(\log n)^{90}, 2 \deg_{G'}(v)\} \le 2 b^9$$ 
    by the previous sentence and $b \ge (\log n)^{100}$. 
    Therefore, if $u.p$ is a low vertex, then w.p. $1 - n^{-3}$, any cell that stores an edge in $E(G')$ adjacent to $u$ is awakened by $u$. 
    The processor corresponding to each cell then notifies the stored edge (processor) that it is in $E'$, giving the edge set $E'$. 
    
    The total running time is $O(\log b)$. 
    The total work is at most $O(|V(G')| + O(|E'|)) \le n / (\log n)^7$ because at most $2|E'|$ cells are awakened in total and $|E'| \le n/b$ w.p. $1 - n^{-6}$ by the third paragraph in the proof of Lemma~\ref{lem:new_increase_work}. 
    The lemma follows.
\end{proof}

\subsubsection{Computing the Edge Set in \textsc{Interweave}} \label{subsubsec:edge_set_interweave}

In this section we show that the edge set $E'$ in Step~\ref{alg:interweave:s8} of \textsc{Interweave}$(G', H_1, H_2, E_{\text{filter}}, i)$ can be computed in $n / (\log n)^7$ work and $O(\log b)$ time.

\begin{lemma} \label{lem:high_degree_has_edge_in_filter}
    During the execution of \textsc{Connectivity}$(G)$, w.p. $1 - n^{-7}$, for any phase $i \in [0, 10 \log\log n)$ and for any vertex $v \in V(G')$, if $\deg_{G'}(v) \ge (\log n)^{10^6 \cdot 1.1^i}$, then $v.p \notin V(G') \backslash V(E_{\text{filter}})$ in Step~\ref{alg:interweave:s8} of \textsc{Interweave}$(G', H_1, H_2, E_{\text{filter}}, i)$.
\end{lemma}
\begin{proof}
    Since the statement in the lemma is about the algorithm at Step~\ref{alg:interweave:s8} of \textsc{Interweave}$(G', H_1, H_2, E_{\text{filter}}, i)$, we can ignore Step~\ref{alg:interweave:s2} to Step~\ref{alg:interweave:s4} of \textsc{Interweave}$(G', H_1, H_2, E_{\text{filter}}, i)$ in this proof, because if the algorithm executes to Step~\ref{alg:interweave:s8}, then the labeled digraph is reverted to its state in Step~\ref{alg:interweave:s1} (by Step~\ref{alg:interweave:s4}).
    
    In Step~\ref{alg:connectivity:s4} of \textsc{Connectivity}$(G)$, the edge set of $G'$ is copied as $E_{\text{filter}}$. 
    In each phase, the edges in $E_{\text{filter}}$ are altered and deleted randomly in Step~\ref{alg:connectivity:s6} of \textsc{Interweave}$(G', H_1, H_2, E_{\text{filter}}, i)$. For each edge $e \in E_{\text{filter}}$ at a timestamp of the algorithm, let the \emph{corresponding edge} $\overline{e}$ of $e$ be the edge in $E(G')$ that is altered (in Step~\ref{alg:connectivity:s6} of \textsc{Interweave}$(G', H_1, H_2, E_{\text{filter}}, i)$) over phases to become $e$ at this timestamp, and let $e$ be the \emph{corresponded edge} of $\overline{e}$.
    
    Let $v$ be a vertex in $ V(G')$ with $\deg_{G'}(v) \ge (\log n)^{10^6 \cdot 1.1^i}$. At the end of Step~\ref{alg:interweave:s6} in phase $i$, consider the edge set $\mathcal{E}_v$ of corresponded edges of all edges in $E(G')$ adjacent to $v$. 
    For each edge $\overline{e} \in \mathcal{E}_v$, every time the algorithm executes Step~\ref{alg:interweave:s6} of \textsc{Interweave}$(G', H_1, H_2, E_{\text{filter}}, i)$ in a phase $j \le i$, the edge $\overline{e}$ is deleted w.p. $10^{-4}$. 
    In all the phases before and including phase $i$, the algorithm executes Step~\ref{alg:interweave:s6} for at most
    \begin{equation*}
        \sum_{j = 0}^{i} 10^6 \cdot 1.1^j \log\log n \le 10^8 \log\log n \cdot 1.1^i
    \end{equation*}
    times, so the probability that edge $\overline{e}$ ever gets deleted before phase $i$ is at most
    \begin{equation*}
        1 - (1 - 10^{-4})^{10^8 \log\log n \cdot 1.1^i} .
    \end{equation*}
    Since $|\mathcal{E}_v| = \deg_{G'}(v) \ge (\log n)^{10^6 \cdot 1.1^i}$, the probability that all edges in $\mathcal{E}_v$ are deleted before phase $i$ is at most
    \begin{equation*}
        \left( 1 - (1 - 10^{-4})^{10^8 \log\log n \cdot 1.1^i} \right)^{(\log n)^{10^6 \cdot 1.1^i}} \le \left(1-(\log n)^{-10^5 \cdot 1.1^i} \right)^{(\log n)^{10^6 \cdot 1.1^i}} \le \left(\frac{1}{e}\right)^{(\log n)^{9}} \le n^{-9} .
    \end{equation*}
    Therefore, w.p. $1 - n^{-9}$, there is an edge in $\mathcal{E}_v$ that is not deleted at the beginning of Step~\ref{alg:interweave:s8} in phase $i$. 
    By the definition of corresponded edges, there is an edge in $E_{\text{filter}}$ that is adjacent to $v.p$ at this timestamp, which means $v.p \in V(E_{\text{filter}})$ and thus $v.p \notin V(G') \backslash V(E_{\text{filter}})$. 
    The lemma follows by a union bound over all the at most $10 \log\log n$ phases and all the at most $n$ vertices.
\end{proof}

\begin{lemma} \label{lem:E_prime_bound_interweave}
    W.p. $1 - n^{-8}$, the edge set $E'$ in Step~\ref{alg:interweave:s8} of \textsc{Interweave}$(G', H_1, H_2, E_{\text{filter}}, i)$ in any phase $i \in [0, 10 \log\log n)$ of \textsc{Connectivity}$(G)$ satisfies $|E'| \le n/(\log n)^9$.
\end{lemma}
\begin{proof}
    Note that the definition of edge set $E'$ in Step~\ref{alg:interweave:s8} of \textsc{Interweave}$(G', H_1, H_2, E_{\text{filter}}, i)$ in phase $i$ is exactly the edge set $E'$ at the end of Step~\ref{alg:extract:s2} of \textsc{Extract}$(E, k)$ if we set $k = 10^6 \cdot 1.1^i \log\log n$, because $V(E_{\text{filter}})$ is exactly the vertex set $V'$ defined at the end of Step~\ref{alg:extract:s2} of \textsc{Extract}$(E, k)$. 
    
    Using the same technique used in the proof of Lemma~\ref{lem:extract_time_work_helper_E}, we get that the expected contribution to $|E'|$ from each vertex in $V(E_{\text{filter}})$ is at most $2$,
    %By the proof of Lemma~\ref{lem:interweave_flat_each_phase}, the ends of any edge in $E_{\text{filter}}$ are roots,
    %Since the number of roots in $V(G')$ is at most $n/(\log n)^{10}$ after Stage $1$, 
    and we get that $\E[|E'|] \le O(|V(E_{\text{filter}})|)$. 
    Using the technique in \S{\ref{sec:boosting}}, we get that w.p. $1 - n^{-9}$, $|E'| \le O(|V(E_{\text{filter}})|) \le O(|V(G')|) \le n / (\log n)^9$. 
    The lemma follows from a union bound over all phases.
\end{proof}

With the above two lemmas, we are ready to give the cost of compute the edge set $E'$.

\begin{lemma} \label{lem:edge_set_interweave_time_work}
    W.p. $1 - n^{-3}$, the edge set $E'$ in Step~\ref{alg:interweave:s8} of \textsc{Interweave}$(G', H_1, H_2, E_{\text{filter}}, i)$ in any phase $i \in [0, 10 \log\log n)$ of \textsc{Connectivity}$(G)$ can be computed in $O(\log b)$ time and $n / (\log n)^7$ work.
\end{lemma}
\begin{proof}
    We shall use the auxiliary array $\mathcal{A}$ built in \S{\ref{subsubsec:auxiliary_array}}, similar to the proof of Lemma~\ref{lem:edge_set_sparse_build_time_work}. 
    In Step~\ref{alg:interweave:s8} of \textsc{Interweave}$(G', H_1, H_2, E_{\text{filter}}, i)$, each vertex $v \in V(G')$ first checks whether $v.p \in V(G') \backslash V(E_{\text{filter}})$. Since each vertex knows whether itself its in $V(E_{\text{filter}})$ in $O(1)$ time by adding $O(|E_{\text{filter}}|)$ work to Step~\ref{alg:interweave:s6} (which does not affect the asymptotic work of Step~\ref{alg:interweave:s6}), the above computation can be done in $O(|V(G')|) \le n/(\log n)^{10}$ work and $O(1)$ time. 
    Next, each vertex $v$ such that $v.p \in V(G') \backslash V(E_{\text{filter}})$ awakens all its adjacent edges in $\mathcal{A}$ following the same manner in the proof of Lemma~\ref{lem:edge_set_sparse_build_time_work} for $10^4 \log b$ rounds or until a cell containing an edge not starting at $v$ is reached. 
    By Lemma~\ref{lem:high_degree_has_edge_in_filter}, w.p. $1 - n^{-7}$, any such vertex $v$ has $\deg_{G'}(v) \le  (\log n)^{10^6 \cdot 1.1^i} \le b^{10^4}$ by $b = (\log n)^{100 \cdot 1.1^i}$. 
    Therefore, by Lemma~\ref{lem:build_auxiliary_property}, w.p. $1 - n^{-4}$, all adjacent edges of any such vertex $v$ are awakened in $10^4 \log b$ rounds, and each awakened edge notifies the edge processor that it is in $E'$.
    
    The total running time is $O(\log b)$. 
    The total work is at most $n / (\log n)^7$ because at most $2|E'| \le 2n / (\log n)^9$ cells are awakened w.p. $1 - n^{-8}$ by Lemma~\ref{lem:E_prime_bound_interweave}.
\end{proof}

\subsection{The Work of \textsc{Connectivity}} \label{subsec:work_connectivity}

In this section, we give the total work of \textsc{Connectivity}$(G)$. 

\begin{lemma} \label{lem:connectivity_total_work}
    For any graph $G$, w.p. $1 - 1/\log n$, \textsc{Connectivity}$(G)$ takes $O(m) + \overline{O}(n)$ work.
\end{lemma}
\begin{proof}
    Step~\ref{alg:connectivity:s1} takes $O(n)$ work.
    By Lemma~\ref{lem:polylog_shrink_reduce_all}, w.p. $1 - n^{-2}$, Step~\ref{alg:connectivity:s2} takes $O(m) + \overline{O}(n)$ work because the edge set $E'$ in Step~\ref{alg:reduce:s4} of \textsc{Reduce}$(V, E, k)$ satisfies $\E[|E'|] \le \overline{O}(n)$ by the same proof of Lemma~\ref{lem:extract_time_work_helper_E}. 
    Step~\ref{alg:connectivity:s3} and Step~\ref{alg:connectivity:s4} take $O(m)$ work. 
    Step~\ref{alg:connectivity:s6} takes $O(n)$ work.
    
    Consider Step~\ref{alg:connectivity:s5}. 
    Consider the \textsc{Interweave}$(G', H_1, H_2, E_{\text{filter}}, i)$ in each phase $i \in [0, 10\log\log n)$.
    By Lemma~\ref{lem:new_increase_work} and Lemma~\ref{lem:edge_set_sparse_build_time_work}, 
    the \textsc{Increase}$(G', H_1, H_2, b)$ in Step~\ref{alg:interweave:s2} takes $(m + n) / (\log n)^2$ work w.p. $1 - 2/(\log n)^2$. 
    By the same proof for Lemma~\ref{lem:truncate_time_work}, w.p. $1 - 1/(\log n)^3$, Step~\ref{alg:interweave:s3} takes $(m + n) / (\log n)^3$ work, because $|E(H_1)|$ is bounded same as $|E(H)|$ required in \textsc{Densify}$(H, b)$. 
    The following \textsc{Alter}$(H_1)$ takes $O(|E(H_1)|) \le m /(\log n)^6$ work w.p. $1 - n^{-9}$. 
    Step~\ref{alg:interweave:s4} without calling \textsc{Remain}$(G', H_1)$ takes the same work as in the previous step.
    To implement Step~\ref{alg:interweave:s5}, we make a copy of $H_1$ at the beginning of Step~\ref{alg:interweave:s1}, which takes $O(|E(H_1)|) \le m /(\log n)^6$ work w.p. $1 - n^{-9}$. To revert the labeled digraph, we copy all parent pointers for each vertex $v \in V(G')$ in Step~\ref{alg:interweave:s1} which takes $O(|V(G')|) \le n/(\log n)^9$ work, because the vertices in $V(G) \backslash V(G')$ are non-roots after Step~\ref{alg:connectivity:s2} of \textsc{Connectivity}$(G)$ which do not change their parent until Step~\ref{alg:connectivity:s6} of \textsc{Connectivity}$(G)$. 
    Running approximate compaction on all active roots (required by $\textsc{Increase}(G', H_1, H_2, b)$ when assigning blocks to active roots) at the end of each phase takes at most $n / (\log n)^7$ work.
    As a result, the total work of all steps but Step~\ref{alg:interweave:s6} to Step~\ref{alg:interweave:s10} of \textsc{Interweave}$(G', H_1, H_2, E_{\text{filter}}, i)$ in all phases except the call to \textsc{Remain}$(G', H_1)$ is at most $2(m + n) / (\log n)^2 \cdot 10 \log\log n \le (m + n) / \log n$ w.p. $1 - 1/(\log n)^{1.5}$.
    
    Next, we show that the total work of Step~\ref{alg:interweave:s6} to Step~\ref{alg:interweave:s10} of \textsc{Interweave}$(G', H_1, H_2, E_{\text{filter}}, i)$ over all phases is at most $O(m + n)$ w.p. $1 - n^{-2}$. 
    To total work of Step~\ref{alg:interweave:s6} over all phases is $O(m)$ w.p. $1 - n^{-8}$. 
    Because the algorithm deletes each edge from $E_{\text{filter}}$ with constant probability in each round of Step~\ref{alg:interweave:s6} over phases so $|E_{\text{filter}}|$ decreases by a constant fraction in each round w.p. $1 - n^{-9}$, and the claim follows from a union bound over all the at most $(\log n)^4$ rounds. 
    To total work of Step~\ref{alg:interweave:s7} over all phases is $O(n)$ by $|V(G')| \le n / (\log n)^{1000}$. 
    By Lemma~\ref{lem:E_prime_bound_interweave} and Lemma~\ref{lem:edge_set_interweave_time_work}, the edge set $E'$ in Step~\ref{alg:interweave:s8} in any phase $i \in [0, 10 \log\log n)$ satisfies $|E'| \le n/(\log n)^9$ and Step~\ref{alg:interweave:s8} takes $n/ (\log n)^7$ work w.p. $1 - n^{-3}$, so the total work of Step~\ref{alg:interweave:s8} and Step~\ref{alg:interweave:s9} over all phases is at most $O(1.1^{10\log\log n} \log\log n) \cdot n/(\log n)^9 + n/(\log n)^7 \le O(n)$ w.p. $1 - n^{-3}$. 
    Step~\ref{alg:interweave:s10} in each phase takes $O(|V(E_{\text{filter}})| + |E(H_2)|)$ work, which is at most $O(|V(G')| + n/(\log n)^6) \le n / (\log n)$ w.p. $1 - n^{-8}$, so the total work of Step~\ref{alg:interweave:s10} over all phases is $O(n)$ w.p. $1 - n^{-7}$.
    
    Finally, we show that the call to \textsc{Remain}$(G', H_1)$ in Step~\ref{alg:interweave:s4} takes $O(m + n)$ work w.p. $1 - 1/(\log n)^8$, which completes the proof because it is called for $1$ time during the execution of \textsc{Connectivity}$(G)$. 
    Step~\ref{alg:remain:s1} and Step~\ref{alg:remain:s2} of \textsc{Remain}$(G', H_1)$ take $O(|E(G')|) \le O(m)$ work. 
    Step~\ref{alg:remain:s3} take $O(m)$ work w.p. $1 - n^{-9}$ by PRAM perfect hashing \cite{DBLP:conf/focs/GilMV91}. 
    We claim that at the beginning of Step~\ref{alg:remain:s4}, $|E_{\text{remain}}| \le n /(\log n)^6$ w.p. $1 - n^{-9}$. 
    Since the algorithm runs \textsc{Alter}$(E(G'))$ in Step~\ref{alg:remain:s1}, the graph $H_1$ is an (unbiased) random subgraph of $G'$ by sampling each edge w.p. $1 / (\log n)^7$ independently. 
    The key observation is that, although $H_1$ has been used (in Step~\ref{alg:interweave:s2} to Step~\ref{alg:interweave:s5} of \textsc{Interweave}$(G', H_1, H_2, E_{\text{filter}}, i)$) before the current phase, it is reverted to its initial state (in Step~\ref{alg:interweave:s5} of the previous phase), so the execution of the algorithm before phase $i$ (excluding Step~\ref{alg:interweave:s2} to Step~\ref{alg:interweave:s5} in each previous phase, which is irrelevant to the analysis of the current phase) is independent of the randomness in generating $H_1$.\footnote{Without this observation, it could be the case that the execution of the algorithm before the current phase is dependent on the generation of $H_1$, such that $H_1$ is not an unbiased random subgraph of $\overline{G'}$ given the fact that the algorithm does not terminate before the current phase. \label{fn:H_1_observation}}
    By the sampling lemma of Karger, Klein and Tarjan (\cite{karger1995randomized}), the number of edges in $G'$ connecting different connected components of its random subgraph $H_1$ is at most $O(|V(G')| / p)$ w.p. $1 - n^{-9}$ where the edge sampling probability $p = (1/(\log n)^7)$.\footnote{\cite{karger1995randomized} states a more general version for the application in minimum spanning tree and the bound on  the number of inter-component edges is in expectation. A high-probability result is a straightforward corollary (e.g., see  \cite{DBLP:journals/jcss/HalperinZ96}, Lemma 9.1).}
    Therefore, $|E_{\text{remain}}| \le n /(\log n)^6$ w.p. $1 - n^{-9}$ at the beginning of Step~\ref{alg:remain:s4}. 
    By Theorem~\ref{thm:ltz_main}, Step~\ref{alg:remain:s4} takes at most $n /(\log n)^6 \cdot O(\log n)$ work w.p. $1 - 1/(\log n)^9$.
\end{proof}

\subsection{The Running Time and Correctness of \textsc{Connectivity}} \label{subsec:time_connectivity}

In this section, we give the running time of \textsc{Connectivity}$(G)$ and prove that it correctly computes the connected components of the input graph $G$.

\begin{lemma} \label{lem:each_phase_running_time}
    During the execution of \textsc{Connectivity}$(G)$, w.p. $1 - 1/(\log n)^2$, for any integer $i \in [0, 10 \log\log n)$, phase $i$ in Step~\ref{alg:connectivity:s5} takes $O(\log b)$ time unless $E_{\text{filter}} = \emptyset$, where $b = (\log n)^{100 \cdot 1.1^i}$. 
\end{lemma}
\begin{proof}
    Consider the \textsc{Interweave}$(G', H_1, H_2, E_{\text{filter}}, i)$ in phase $i$. 
    Consider the $\textsc{Increase}(G', H_1, H_2, b)$ in Step~\ref{alg:interweave:s2} of \textsc{Interweave}$(G', H_1, H_2, E_{\text{filter}}, i)$. 
    In Step~\ref{alg:new_increase:s1} of $\textsc{Increase}(G', H_1, H_2, b)$, we calculate the running time of $\textsc{SparseBuild}(G', H_2, b)$. 
    We already described how to let each vertex in $V(G')$ knows whether itself is an active root in the proof of Lemma~\ref{lem:new_increase_work}, so at the end of each phase the algorithm runs approximate compaction on all active roots to compact them in an index array of length $O(n / b^{10})$ in $O(\log^* n)$ time w.p. $1 - n^{-6}$ by Lemma~\ref{lem:interweave_roots_shrink_each_phase}.
    Step~\ref{alg:sparsebuild:s1} of $\textsc{SparseBuild}(G', H_2, b)$ takes $O(1)$ time since the active roots are indexed. 
    Step~\ref{alg:sparsebuild:s2} takes $O(1)$ time. 
    Step~\ref{alg:sparsebuild:s3} takes $O(\log b)$ time by the binary tree counting argument. 
    Step~\ref{alg:sparsebuild:s4} takes $O(\log b)$ time w.p. $1 - n^{-3}$ by Lemma~\ref{lem:edge_set_sparse_build_time_work}. 
    Therefore, Step~\ref{alg:new_increase:s1} of $\textsc{Increase}(G', H_1, H_2, b)$ takes $O(\log b)$ time w.p. $1 - 2n^{-3}$.
    By the fourth paragraph in the proof of Lemma~\ref{lem:new_increase_work} and Lemma~\ref{lem:increase_work_time}, Step~\ref{alg:new_increase:s2} of $\textsc{Increase}(G', H_1, H_2, b)$ takes $O(\log b)$ time w.p. $1 - 2/(\log n)^3$. 
    (This probability will be boosted to $1 - n^{-c}$ in \S{\ref{sec:boosting}}.)
    Summing up, $\textsc{Increase}(G', H_1, H_2, b)$ takes $O(\log b)$ time w.p. $1 - 3/(\log n)^3$.
    
    Consider Step~\ref{alg:interweave:s3} of \textsc{Interweave}$(G', H_1, H_2, E_{\text{filter}}, i)$. 
    By Lemma~\ref{lem:expandm_constant_time}, each $\textsc{Expand-Maxlink}(H_1)$ takes $O(1)$ time w.p. $1 - 1/(\log n)^6$. 
    %This probability will be boosted to $1 - n^{-9}$ in \S{\ref{sec:boosting}} \Cliff{change citation}. 
    So all $20 \log b$ rounds of Step~\ref{alg:interweave:s3} take $O(\log b)$ time w.p. $1 - 1/(\log n)^5$. 
    By Theorem~\ref{thm:ltz_main}, the following call takes $O(\log\log n) \le O(\log b)$ time w.p. $1 - 1/(\log n)^8$ by $b \ge (\log n)^{100}$. %(This probability will be boosted to $1 - n^{-9}$ in \S{\ref{sec:boosting}} \Cliff{change citation}.) 
    The \textsc{Alter}$(E(H_1))$ takes $O(1)$ time. 
    Summing up, Step~\ref{alg:interweave:s3} of \textsc{Interweave}$(G', H_1, H_2, E_{\text{filter}}, i)$ takes $O(\log b)$ time w.p. $1 - 1/(\log n)^7$.
    
    In Step~\ref{alg:interweave:s4} of \textsc{Interweave}$(G', H_1, H_2, E_{\text{filter}}, i)$ the algorithm only needs to check whether all edges in $E(H_1)$ are loops without calling \textsc{Remain}$(G', H_1)$ because the algorithm does not return $\emptyset$, which can be done in $O(1)$ time. 
    Step~\ref{alg:interweave:s5} can be done in $O(1)$ time by the second paragraph in the proof of Lemma~\ref{lem:connectivity_total_work}.
    Each round of Step~\ref{alg:interweave:s6} of \textsc{Interweave}$(G', H_1, H_2, E_{\text{filter}}, i)$ can be done in $O(1)$ time by Lemma~\ref{lem:constant_shrink_work_time}. 
    So Step~\ref{alg:interweave:s6} takes $O(10^6 \cdot 1.1^i \log\log n) = O(\log b)$ time by $b = (\log n)^{100 \cdot 1.1^i}$. 
    Step~\ref{alg:interweave:s7} takes $O(i + 2\log\log n) \le O(\log \log n) \le O(\log b)$ time. 
    Step~\ref{alg:interweave:s8} takes $O(\log b)$ time w.p. $1 - n^{-3}$ by Lemma~\ref{lem:edge_set_interweave_time_work}. 
    Step~\ref{alg:interweave:s9} takes $O(\log b)$ time, same as Step~\ref{alg:interweave:s6}.
    Step~\ref{alg:interweave:s10} takes $O(1)$ time. 
    
    Summing up and by a union bound, the lemma follows.
\end{proof}

\begin{lemma} \label{lem:remain_correct}
    During the execution of \textsc{Connectivity}$(G)$, w.p. $1 - 1/(\log n)^7$, for any integer $i \in [0, 10 \log\log n)$, if the \textsc{Interweave}$(G', H_1, H_2, E_{\text{filter}}, i)$ in phase $i$ returns $\emptyset$, then the \textsc{Remain}$(G', H_1)$ in Step~\ref{alg:interweave:s3} of \textsc{Interweave}$(G', H_1, H_2, E_{\text{filter}}, i)$ takes $O(\log d + \log\log n)$ time where $d$ is the maximum diameter of $G$, and \textsc{Connectivity}$(G)$ computes all the connected components of $G$.
\end{lemma}
\begin{proof}
    Let $i \in [0, 10 \log\log n)$ be the last phase of \textsc{Connectivity}$(G)$. 
    By Step~\ref{alg:connectivity:s1} of \textsc{Connectivity}$(G)$ and Lemma~\ref{lem:alg:reduce_correctness}, Step~\ref{alg:connectivity:s2} is a contraction algorithm. 
    Note that in each phase $j < i$, Step~\ref{alg:interweave:s6} to Step~\ref{alg:interweave:s9} of \textsc{Interweave}$(G', H_1, H_2, E_{\text{filter}}, i)$ is a contraction algorithm on graph $G'$ by the same proof for Lemma~\ref{lem:alg:reduce_correctness}. 
    Vertices in $V(G) \backslash V(G')$ are non-roots created at the end of Step~\ref{alg:connectivity:s2} of \textsc{Connectivity}$(G)$.
    Therefore, the execution of \textsc{Connectivity}$(G)$ before phase $i$ is a contraction algorithm if we execute \textsc{Shortcut}$(V(G))$ (which will be done in Step~\ref{alg:connectivity:s6} of \textsc{Connectivity}$(G)$). 
    
    Now consider phase $i$. 
    A proof similar to the proof of Lemma~\ref{lem:increase_is_contraction} shows that Step~\ref{alg:interweave:s2} of \textsc{Interweave}$(G', H_1, H_2, E_{\text{filter}}, i)$ is a contraction algorithm on $G'$ if we execute \textsc{Alter}$(E(G'))$, so Step~\ref{alg:interweave:s2} is a contraction algorithm on $H_1$. 
    By the same proof of Lemma~\ref{lem:truncate:flat}, w.p. $1 - 1/(\log n)^8$, Step~\ref{alg:interweave:s3} of \textsc{Interweave}$(G', H_1, H_2, E_{\text{filter}}, i)$ is a contraction algorithm on $H_1$. 
    In Step~\ref{alg:interweave:s4} of \textsc{Interweave}$(G', H_1, H_2, E_{\text{filter}}, i)$, since all edges in $E(H_1)$ are loops, all connected components of $H_1$ are computed because Step~\ref{alg:interweave:s2} and Step~\ref{alg:interweave:s3} are contraction algorithms on $H_1$. 
    Now consider \textsc{Remain}$(G', H_1)$ in Step~\ref{alg:interweave:s4}.\footnote{The algorithm needs to run \textsc{Remain}$(G', H_1)$ to maintain correctness because if the assumption $\lambda \ge b^{-0.1}$ is wrong in this phase, then the edge sampling on $G'$ (to obtain $H_1$) disconnects some components of $G'$, but Step~\ref{alg:interweave:s3} of \textsc{Interweave}$(G', H_1, H_2, E_{\text{filter}}, i)$ contracts all components of $H_1$ either because the contraction algorithm on $H_1$ takes less running time than its upper bound or the disconnection does not create any component with spectral gap less than $\lambda$. In both cases, the computation of connected components of $G'$ is not finished, which will be done in \textsc{Remain}$(G', H_1)$ by adding all the remaining edges unsampled into $H_1$. \S{\ref{sec:stage3}} is correct without \textsc{Remain}$(G', H_1)$ because $\lambda \ge b^{-0.1}$ holds.}
    Step~\ref{alg:remain:s1} of \textsc{Remain}$(G', H_1)$ takes $O(1)$ time. 
    Step~\ref{alg:remain:s2} takes $O(1)$ time to let each edge of $E(G')$ check whether itself is in $H_1$. 
    Step~\ref{alg:remain:s3} take $O(\log^* n)$ time w.p. $1 - n^{-9}$ by PRAM perfect hashing \cite{DBLP:conf/focs/GilMV91}. 
    Note that the diameter of graph $(V(E_{\text{remain}}), E_{\text{remain}})$ is at most $d$ because computing all connected components of $H_1$ and executing \textsc{Alter}$(E(G'))$ is a contraction algorithm on $G'$ and any contraction does not increase the diameter; moreover, removing loops and parallel edges from $E_{\text{remain}}$ does not influence the diameter of the graph.
    Therefore, Step~\ref{alg:remain:s4} takes $O(\log d + \log\log n)$ time w.p. $1 - 1/(\log n)^9$ by Theorem~\ref{thm:ltz_main}, and it computes all connected components of $G'$.
    
    By the first paragraph, the algorithm computes all connected components of $G$ after Step~\ref{alg:connectivity:s6} of \textsc{Connectivity}$(G)$. The lemma follows from a union bound.
\end{proof}

\begin{lemma} \label{lem:spectral_gap_b_i}
    For any graph $G$ with component-wise spectral gap $\lambda$, w.p. $1 - 1/(\log n)^4$, \textsc{Connectivity}$(G)$ computes all the connected components of $G$ before or in phase $i$, where $i \in [0, 10 \log\log n)$ is the minimum integer that satisfies $(\log n)^{-10 \cdot 1.1^i} \le \lambda$.
\end{lemma}
\begin{proof}
    First of all, such an integer $i$ must exist because $(\log n)^{-10 \cdot 1.1^{9c \log\log n}} \le n^{-9c}$, and $\lambda \ge \phi^2 / 2 \ge m^{-2} / 2 \ge n^{-3c}$ by Cheeger's inequality and $m \le n^c$ where constant $c \ge 2$.
    
    By Lemma~\ref{lem:remain_correct}, w.p. $1 - 1/(\log n)^8$, if \textsc{Connectivity}$(G)$ terminates in a phase $j \in [0, 10 \log\log n)$, then it correctly computes all connected components of $G$. 
    Suppose \textsc{Connectivity}$(G)$ does not terminate before phase $i$. 
    Consider phase $i$ and let $b = (\log n)^{100 \cdot 1.1^i}$. 
    By Lemma~\ref{lem:sparse_H_big_component_large_degree}, w.p. $1 - 1/(\log n)^5$, it must be $\overline{\deg}_{G'}(v) \ge b$ for any active root $v \in V(G')$ at the end of Step~\ref{alg:interweave:s2} within \textsc{Interweave}$(G', H_1, H_2, E_{\text{filter}}, i)$. 
    By Definition~\ref{def:upper_degree} and Step~\ref{alg:connectivity:s3} of \textsc{Connectivity}$(G)$, at the end of Step~\ref{alg:interweave:s2} of \textsc{Interweave}$(G', H_1, H_2, E_{\text{filter}}, i)$, the graph $H_1$ is an (unbiased) random subgraph of $\overline{G'}$ by sampling each edge w.p. $1 / (\log n)^7$ independently (see the observation in last paragraph in the proof of Lemma~\ref{lem:connectivity_total_work} and footnote~\ref{fn:H_1_observation}).
    Note that the component-wise spectral gap of $\overline{G'}$ is at least $\lambda$ since the algorithm only performs contractions on $G$ (same as in \S{\ref{sec:stage3}}). 
    Now since $\deg(\overline{G'}) \ge b$ and the sampling probability $p = 1/(\log n)^7$, we have that $p \cdot \deg(\overline{G'}) \ge C \ln n$ and thus Corollary~\ref{cor:sample_preserve} is applicable. We have that w.p. $1 - n^{-9}$, the component-wise spectral gap $\lambda'$ of $H_1$ satisfies
    \begin{equation*}
        \lambda' \ge \lambda - C' \cdot \sqrt{\frac{\ln n}{1 / (\log n)^7 \cdot b}} \ge b^{-0.1} - b^{-0.4} \ge b^{-0.2} 
    \end{equation*}
    by $\lambda \ge (\log n)^{-10 \cdot 1.1^i} = b^{-0.1}$, where $C'$ is the constant in Corollary~\ref{cor:sample_preserve}. 
    So the maximum diameter $d'$ over all components of $H_1$ is at most $O(\log n / \lambda') \le b^{0.3}$. 
    Since any shortest path of $H$ has at most $b^{0.3} + 1 \le b^7$ vertices, we can apply the same proof of Lemma~\ref{lem:truncate_path_length} to get that $20 \log b$ rounds of $H_1 = \textsc{Expand-Maxlink}(H_1)$ (Step~\ref{alg:interweave:s3} of \textsc{Interweave}$(G', H_1, H_2, E_{\text{filter}}, i)$) contracts $H_1$ such that its diameter becomes at most $1$ w.p. $1 - 1/(\log n)^8$, and the following call to Theorem~\ref{thm:ltz_main} computes all connected components of $H_1$ in $O(\log\log n)$ time w.p. $1 - 1/(\log n)^9$ (the same, more detailed proof appears in \cite{liu2020connected}). 
    As a result, the following \textsc{Alter}$(H_1)$ makes all edges in $E(H_1)$ loops and Step~\ref{alg:interweave:s4} of \textsc{Interweave}$(G', H_1, H_2, E_{\text{filter}}, i)$ returns $\emptyset$. 
    So phase $i$ is the last phase and \textsc{Connectivity}$(G)$ correctly computes all connected components of $G$ by the second paragraph of this proof. 
    The lemma follows from a union bound. 
\end{proof}

\begin{lemma} \label{lem:connectivity_time_correct_final}
    For any graph $G$ with component-wise spectral gap $\lambda$, w.p. $1 - 1/\log n$, \textsc{Connectivity}$(G)$ computes all the connected components of $G$ in $O(\log(1/\lambda) + \log\log n)$ time.
\end{lemma}
\begin{proof}
    By Lemma~\ref{lem:polylog_shrink_reduce_all}, w.p. $1-n^{-2}$, Step~\ref{alg:connectivity:s2} of \textsc{Connectivity}$(G)$ takes $O(\log\log n)$ time. 
    By Lemma~\ref{lem:build_auxiliary_time_work}, the \textsc{BuildAuxiliary}$(G')$ at the end of Step~\ref{alg:connectivity:s2} takes $O(\log\log n)$ time w.p. $1 - n^{-5}$ (required for reducing the work, see \S{\ref{subsubsec:auxiliary_array}}). 
    Step~\ref{alg:connectivity:s3}, Step~\ref{alg:connectivity:s4}, and Step~\ref{alg:connectivity:s6} take $O(1)$ time. 
    W.p. $1 - 2/(\log n)^2$, assuming Lemma~\ref{lem:each_phase_running_time}, Lemma~\ref{lem:remain_correct}, and Lemma~\ref{lem:spectral_gap_b_i} hold.
    
    Let phase $i$ be the last phase of \textsc{Connectivity}$(G)$. 
    If $i = 0$, then in \textsc{Interweave}$(G', H_1, H_2, E_{\text{filter}}, i)$ of phase $0$, before calling \textsc{Remain}$(G', H_1)$ in its Step~\ref{alg:interweave:s4}, the algorithm takes $O(\log ((\log n)^{100})) = O(\log\log n)$ time by the proof of Lemma~\ref{lem:each_phase_running_time}.
    Now suppose $i \ge 1$. 
    By Lemma~\ref{lem:spectral_gap_b_i}, it must be
    \begin{equation} \label{eq:lamba_phase_i}
        \lambda \le (\log n)^{-10 \cdot 1.1^{i+1}} .
    \end{equation}
    By Lemma~\ref{lem:each_phase_running_time}, the total running time in Step~\ref{alg:connectivity:s5} of \textsc{Connectivity}$(G)$ before calling \textsc{Remain}$(G', H_1)$ in phase $i$ is at most
    \begin{equation*}
        \sum_{j < i} O\left(\log \left(\left(\log n\right)^{100 \cdot 1.1^j}\right)\right) \le \sum_{j < i} O\left(1.1^j \log\log n \right) \le O\left(1.1^i \log\log n\right) \le O\left(\log\left(1/\lambda\right)\right) ,
    \end{equation*}
    where the last inequality is by Equation~(\ref{eq:lamba_phase_i}). 
    By Lemma~\ref{lem:remain_correct}, the \textsc{Remain}$(G', H_1)$ in phase $i$ takes $O(\log d + \log\log n) \le O(\log(1/\lambda) + \log\log n)$ time by $d \le O(\log n / \lambda)$. 
    Summing up and by a union bound, the lemma follows.
\end{proof}

The success probability $1 - 1/\log n$ will be boosted to $1 - n^{-c}$ for arbitrary constant $c > 0$ in \S{\ref{sec:boosting}}.

\section{Boosting the Work and Success Probability} \label{sec:boosting}

In this section we improve the total work of our algorithm to $O(m + n)$ and the success probability to high probability in $n$.

\subsection{Bounding the Size of Edge Set with High Probability} \label{subsec:bound_edge_set_whp}

The main step towards getting a high probability bound is to show that in the execution of  \textsc{Extract}$(E, k)$, the size of each $E_i$ is bounded by $2 \cdot 0.999^{i-1}\cdot n$ with high probability for $i \in [1,k]$. As we discussed in the proof of Lemma \ref{lem:extract_time_work_helper_E} in order for some edge $(v,u)$  to contribute to $E_{i+1}$ through a vertex $v$, all  of the edges originally in $E_{i,v}$ must be deleted by the end of the \textsc{Filter}$(E',k)$. We show that the excepted number of edges that contribute to $E_{i+1}$ is concentrated around its mean. The main challenge to provide a concentration bound is that the events that two different edges $e, e' \in E_i$ contribute to $E_{i+1}$ are not necessarily independent from each other. For example, suppose that $e=(v,u)$ and $e'=(v,u')$. In this case $e$ and $e'$ share a common endpoint and if we condition on the event that $e$ contributes to $E_{i+1}$, then it implies that all of the edges in $E_{i,v}$ or all edges in $E_{i,u}$ are deleted within  \textsc{Filter}. This intuitively increases the probability that $e'$ contributes to the $E_{i+1}$ since if all edges in $E_{i,v}$ are deleted, then $e'$ also has a good chance contributes to the matching. The reason that we cannot say that $e'$ contributes to $E_{i+1}$ with the probability $1$, is that $e'$ might have moved to some other vertex and that vertex might have some edges that are not deleted by the end of \textsc{Filter}. This however, still shows  a possible positive correlation between the contribution of vertices, and we can't simply use Chernoff-like concentration bounds as there are only valid for random variables with negative association. %Consider a vertex $v \in V'_{i-1}$ such that $v$ has more than $1000 \log n$ incident edges in $E_{i-1}$. The probability that each of these edges gets deleted during $k+1$ rounds of \textsc{Filter}$(E',k)$ is $(1-(1-10^{-4})^{k+1})$. We then have the following bound for this probability for $k=1000 \log \log n \log n$.

In order to prove the concentration bound, we take two steps. First we show that for all rounds $i$ and vertices $v \in V'_i$ if $|E_{i,v}| \ge 1000 \log n$, then with the high probability no edges contribute to $E_{i+1}$ through these vertices. Therefore, we can assume that for every vertex $v$ such that some edge contributes to $E_{i+1}$ through $v$ we have $|E_{i,v}| < 1000 \log n$, thus the total contribution of $v$ is relatively small. We then use this observation, and apply the concentration bound for self-bounding functions to show that $E_i$ is bounded by $2 \cdot 0.999^{i-1}\cdot n$ with high probability.

Consider any edge $e \in E_i$, the probability that this edge gets deleted by the end of \textsc{Filter}$(E',k)$ is $(1-(1-10^{-4})^{k+1})$. Since we run the algorithm for $k=1000 \log \log \log n$, then we have  the following bound for the probability that $e$ gets deleted by the end of the \textsc{Filter}$(E',k)$.

\begin{align*}
(1-(1-10^{-4})^{k+1}) = (1-(1-10^{-4})^{1000 \log \log \log n+1}) \le 1- \frac{1}{\log \log n} .
\end{align*}
Now suppose that for a vertex $v \in V'_i$ we have $|E_{i,v}| \ge 1000  \log n$, then the probability that all of the edges in $|E_{i,v}|$ gets deleted is at most
\begin{align*}
(1- \frac{1}{\log \log n})^ {1000 \log n} \le n^{-10} . 
\end{align*}

Therefore, we can say that if for a vertex $v$, we have $|E_{i,v}| \ge 1000  \log n$, then no edges will contribute to $E_{i+1}$ through $v$ with the probability of $(1-n^{-10})$. By taking the union bound over all  possible $n$ different vertices $v$ and $k$ rounds of \textsc{Extract}$(E, k)$, we can say that with the probability of at least $(1-n^{-9} k) \le (1-n^{-8})$, for all rounds $i \in [0,k-1]$ if a vertex $v \in V'_i$ has more than $1000 \log n$ edges in $E_{i,v}$, then none of these edges contribute to $E_{i+1}$ through $v$.

For a round $i$  of \textsc{Extract}$(E, k)$, let $V''_i$ be the set of vertices in $V_i$ such that for each $v \in V''_i$ we have $|E_{i,v}| < 1000  \log n$. As we discussed, all of the contribution to $E_{i+1}$ are through vertices in $V''_i$ with high probability in $n$.  This means that with high probability in $n$, the total contribution of all vertices to $E_{i+1}$  is at most $1000 \log n$. We then use this property, and apply the concentration bounded for self-bounding functions. First we define the self-bounding functions.

\begin{definition}
Let $X_1, X_2, \cdots, X_n$ be $n$  independent random variables taking values in some measure space $\mathcal{X}$. Also, let $f: \mathcal{X}^n \rightarrow  \mathbb{R}$, then $f$ is $(a)$-self-bounding if there exists functions $f_1, f_2, \cdots, f_n : \mathcal{X}^{n-1} \rightarrow \mathbb{R}$ such that
\begin{align*}
0 \le f(x) - f_j(x^{(j)}) \le 1 \,.
\end{align*} 
And,
\begin{align*}
\sum_{i=1}^n \big( f(x) - f_j(x^{(j)}) \big) \le a \cdot f(x) \,.
\end{align*}

For all inputs $x=(x_1, x_2, \cdots, x_n) \in \mathcal{X}^n$ where $x^{(j)}=(x_1, \cdots, x_{j-1}, x_{j+1}, \cdots, x_n)$ is the $x$ when $j^{\text{th}}$ element is dropped.
\end{definition}

The following concentration bound is known for self-bounding functions

\begin{lemma} [\cite{boucheron2009concentration}]
Let $Z= f(X_1, X_2, \cdots, X_n)$, then if $f(.)$ is $(a)$-self-bounding function  we have
\begin{align*}
\Pr [ Z \ge \E[Z]+t] \le \exp\bigg(\frac{-t^2}{2a \E[Z] + a \cdot t}\bigg) \,.
\end{align*}
\end{lemma}

We now use the concentration bound above, to provide a concentration bound for $|E_{i+1}|$. Let us fix an ordering on the edges in $E_i$ and assume that $|E_i|=\{e_1, e_2, \cdots, e_{\ell}\}$ where $\ell = |E_i|$. Also, let $X_{e_1}, X_{e_2}, \cdots, X_{e_{\ell}}$ be $\ell$ random variables where $X_{e_j}$ indicates whether $e_j$ is deleted by the end of the \textsc{Filter}$(E',k)$. Specifically, $X_{e_j} = 1$ if $e_j$ is deleted by \textsc{Filter}, and it is $0$ otherwise. It is then clear that these random variables are independent as \textsc{Filter} deletes each edge independent of others. Now we define function $f_v : \{0,1\}^{\ell} \rightarrow \mathbb{R}$ as follows for every $v \in V''_i$.
$$
f_v(x) =
\begin{cases}
|E_{i,v}| , & \text{If for all edges $e \in E_{i,v}$ we have $x_e=1$}
\\ & \text{, i.e., all edges in $E_{i,v}$ are deleted based on the input $x$} \\
0, & \text{Otherwise}
\end{cases}
$$
and we define function $f$ as follows $f(x) = \frac{\sum_{v \in V''_i} f_v(x)}{2000 \log n}$. We then claim that $f$ is $(1000 \log n)$-self-bounding.

\begin{claim}
Let $f$ be the function defined as above, then $f$ is $(1000 \log n)$-self-bounding.
\end{claim}
\begin{proof}
For each $e_j \in E_i$, we define function $f_{e_j}(x^{(e_j)})$ as below.
\begin{align*}
f_{e_j}(x^{(e_j)}) = f(x_{e_1}, \cdots, x_{e_{i-1}},0, x_{e_{i+1}}, \cdots, x_{e_\ell}) \,.
\end{align*}
In other words, $f_{e_j}(x^{(e_j)})$ is the same as function $f$ when we change that $e_j$ is not deleted. For each $e_j$, we use $x'_j$ to denote the input sequence $(x_{e_1}, \cdots, x_{e_{i-1}},0, x_{e_{i+1}}, \cdots, x_{e_\ell})$. Thus, it is equivalent to define $f_{e_j}(x^{(e_j)})= f(x'_j)$. Also, for an input sequence $x$, we use $\eta(x)$ to denote the set of vertices in $V''_i$ such that for each $v \in \eta(x)$ all edges in $E_{i,v}$ are deleted according to $x$. Based on this definition, we can re-write function $f$ as follows.
\begin{align*}
f(x) = \frac{\sum_{v \in \eta(x)} |E_{i,v}|}{2000 \log n}
\end{align*}

Now we are ready to show these function $f$ satisfies the properties of self-bounding functions. Considering the first property we have to show that for every edge $e_j$ we have
\begin{align*}
0 \le f(x) - f_{e_j}(x^{(e_j)}) \le 1 
\end{align*}

Let $e_j=(v',u')$. It is then clear that $\eta({x'_j}) \subseteq \eta(x)$ since in $x'$ we are changing the input such that $e_j$ is not deleted, and it is only changed for one edge. Then by the definition of functions $f(.)$ and $f_{e_j}(.)$ we have
\begin{align}
 f(x) - f_{e_j}(x^{(e_j)}) &= f(x) - f(x'_j) \nonumber \\
 \label{eq:self_bound_p1}
&= \frac{\sum_{v \in \eta(x), v \notin \eta({x'_j})} |E_{i,v}|}{2000 \log n} 
\\ & \le \frac{|E_{i,v'}|+ |E_{i,u'}|}{2000 \log n} \nonumber \\
& (\text{$\eta({x'_j})$ is different from $\eta(x)$ for at most vertices adjacent to $e_j$}) \nonumber
\\ & \le 1 \,. \nonumber
\\ & (\text{Since $|E_{i,v}| \le 1000 \log n$ for every vertex $v \in V''_i$ }) \nonumber
\end{align}
 Thus, function $f$ satisfies the first property of self-bounding functions.
 
 For the second property we have
 \begin{align*}
 &\sum_{e_j \in E_i} f(x) - f_{e_j}(x^{(e_j)}) \\
 & = \sum_{e_j \in E_i} \bigg(\frac{\sum_{v \in \eta(x), v \notin \eta({x'_j})} |E_{i,v}|}{2000 \log n}\bigg) \\ 
&\le \sum_{v \in \eta(x)} \bigg(\frac{1000 \log n \cdot |E_{i,v}|}{2000 \log n}\bigg)  \\
& (\text{$\eta({x'_j})$ is different from $\eta(x)$ for at most vertices adjacent to $e_j$} \\
& \text{and each $v \in V''_i$ is adjacent to at most $1000 \log n$ edges}) \nonumber
\\ &= \sum_{v \in \eta(x)} \bigg(\frac{|E_{i,v}|}{2}\bigg) \\
& = 1000 \log n \cdot f(x) \,.
 \end{align*}
Thus, function $f$ is $(1000 \log n)$-self-bounded and it proof the claim.
\end{proof}

%Therefore, to bound the size of $|E_i|$ we can assume w.h.p. that the degree of every vertex in $|E_{i-1}|$ is at most $100 \log n$ as vertices with higher degrees do not contribute to $E_i$.

We now show the concentration bound for $E_{i+1}$. Recall that $X_{e_1}, X_{e_2}, \cdots, X_{e_\ell}$ are the random variables where $X_{e_i}$ indicates whether $e_i$ is deleted by the end of the \textsc{Filter}. Also, let $\vec{X}=(X_{e_1}, X_{e_2}, \cdots, X_{e_\ell})$. First we claim that with the probability of $(1-n^{-8})$ we have $|E_{i+1}| \le \sum_{v \in V''_i} f_v(\vec{X})$. The reason is that in order an edge $(v,u)$ contribute to $E_{i+1}$, for at least one of its endpoints, let's say $v$, all the edges in $E_{i,v}$ must be deleted. In this case $f_v(\vec{X})$ is $E_{i,v}$ which is an upper bound on the number of edges that contribute to $E_{i+1}$ through $v$. Also with the probability of at least $(1-n^{-8})$ all of the contribution to $E_{i+1}$ are through vertices in $V''_i$. This implies that $|E_{i+1}| \le \sum_{v \in V''_i} f_v(\vec{X})$, and $\frac{|E_{i+1}|}{2000 \log n} \le f(\vec{X})$.

Also, as we discussed in the proof of Lemma \ref{lem:extract_time_work_helper_E}, the expected number of edges contributed through each vertex $v \in V'_i$ is upper bounded by $1$. Using the similar argument we can say that $\E[f_v(\vec{X})] \le 1$ for each $v \in V''_i$. Thus, $\E[f(\vec{X})] \le \frac{|V''_i|}{2000 \log n}$. Let $Z=f(\vec{X})$. Since $f$ is $(1000 \log n)$-self-bounded function, using the concentration bound for self-bounding functions we can say

\begin{align*}
\Pr[ Z \ge \E[Z] + n^{0.6}] &\le \exp\bigg(\frac{-n^{1.2}}{2000 \log n \E[Z]+ 1000 \log n \cdot n^{0.6}} \bigg) \\
\\& \le \exp\bigg(\frac{-n^{1.2}}{|V''_i|+ 1000 \log n \cdot n^{0.6}}\bigg)  \\
& (\text{Since $\E[Z] \le \frac{|V''_i|}{2000 \log n}$}) \\
& \le \exp\bigg(\frac{-n^{1.2}}{n+ 1000 \log n \cdot n^{0.6}}\bigg) \le n^{-10} \\ \,
& (\text{Since $|V''_i| \le n$}) . \\
\end{align*}

Recall that with the probability of at least $(1-n^{-8})$ we had $\frac{|E_{i+1}|}{2000 \log n} \le f(\vec{X})$. By combining this with the concentration bound above, we can say with the probability of at least $(1-n^{-7})$ we have
\begin{align*}
&\frac{|E_{i+1}|}{2000 \log n} \le \E[f(\vec{X})] + n^{0.6}
\end{align*}
Thus,
\begin{align*}
|E_{i+1}| &\le 2000 \log n \cdot E[f(\vec{X})]+ 2000 \log n \cdot n^{0.6}  \\
&= \sum_{v \in V''_i} E[f_v(\vec{X})] + 2000 \log n \cdot n^{0.6} \\
&\le |V''_i| + 2000 \log n \cdot n^{0.6} \,.
\end{align*}

By Lemma \ref{lem:extract_time_work_helper_Vprime}, with the probability of at least $1-kn^{-9} \ge 1-n^{-8}$, we have $|V''_i| \le 0.999^i\cdot n$. Thus with the probability of at least $1-n^{-6}$ we have
\begin{align*}
|E_{i+1}| \le |V''_i| + 2000 \log n \cdot n^{0.6} \le  2 \cdot 0.999^i\cdot n \,,
\end{align*}
which proves the following lemma.
\begin{lemma}
\label{lem:concentrationExtract}
For each round $i \in [1,k]$ of \textsc{Extract}, with the probability of at least $1-n^{-6}$ we have $|E_i| \le 2 \cdot 0.999^{i-1}\cdot n$.
\end{lemma}

We can use the same technique to provide the high probability bound for the size of $E'$ in the \textsc{Interweave}$(G', H_1, H_2, E_{\text{filter}}, i)$.
\begin{corollary}
Let $E'$ be the set of edges by end of Step \ref{alg:interweave:s8} of \textsc{Interweave}$(G', H_1, H_2, E_{\text{filter}}, i)$. Then, $|E'|\le 10^9 \cdot n$ with the probability of at least $1-n^{-6}$.
\end{corollary}
\begin{proof}
    The set $E'$ is computed in the same way that $E_{i+1}$ is computed in round $i$ of \textsc{Extract}. The only difference is that in the proof above, we are assuming that each edge gets deleted with the probability of $(1-10^{-4})$ for $1000 \log \log \log n$ rounds. However, \textsc{Interweave}$(G', H_1, H_2, E_{\text{filter}}, i)$ this happens for every edge for $10^6 \cdot 1.1^i \log\log n$ rounds. Also, it is guaranteed that $10^6 \cdot 1.1^i \log\log n$ is bounded by $10^6 \cdot (\log n)^3$. We show that we can still claim that in order for a vertex to have a contribution to $E'$, it should have at most $10^8 \cdot (\log n)^4$ edges. Consider a vertex with at least  $10^8 \cdot (\log n)^4$ edges. Then, the probability that all of these edges get deleted is
    \begin{align*}
        (1-(1-10^{-4})^{10^6 \cdot 1.1^i \log\log n})^{10^8 \cdot (\log n)^4} \le  (1-(1-10^{-4})^{10^6 \cdot (\log n)^3})^{10^8 \cdot (\log n)^4} \le n^{-10} \,.
    \end{align*}
    Thus, with high probability in $n$, vertices with more than $10^8 \cdot (\log n)^4$ edges do not contribute to $E'$.  Then we define the function $f= \frac{\sum_{v \in V''} f_v(x)}{2 \cdot 10^8 \cdot (\log n)^4}$ where $V''$ is the set of vertices with at most $10^8 \cdot (\log n)^4$ edges. Then, it is straightforward to verify that $f$ is $(10^8 \cdot (\log n)^4)$-self bounding and get the concentration bound stated in this corollary.
\end{proof}

We further leverage this technique to get the concentration bound for the set of edges $E'$ in Step \ref{alg:reduce:s4} of \textsc{Reduce}. The proof immediately follows from what we have discussed so far as it calls \textsc{Filter} with $k=10^6 \log \log n$.
\begin{corollary}
Let $E'$ be the set of edges by end of Step \ref{alg:reduce:s4} of \textsc{Reduce}. Then, $|E'|\le 10^7 \cdot n$ with the probability of at least $1-n^{-6}$. 
\end{corollary}
\subsection{Boosting to Linear Work with High Probability}

We show that \textsc{Connectivity}$(G)$ can be implemented to take $O(m + n)$ work with high probability.

\begin{lemma} \label{lem:connectivity_total_work_high_prob}
    For any graph $G$, w.p. $1 - 1/n$, \textsc{Connectivity}$(G)$ takes $O(m + n)$ work.
\end{lemma}
\begin{proof}
    Consider the proof of Lemma~\ref{lem:connectivity_total_work}. 
    Step~\ref{alg:connectivity:s2} of \textsc{Connectivity}$(G)$ takes $O(m) + \overline{O}(n)$ work because the edge set $E'$ in Step~\ref{alg:reduce:s4} of \textsc{Reduce}$(V, E, k)$ has $\E[|E'|] \le \overline{O}(n)$.
    Using the technique in \S{\ref{subsec:bound_edge_set_whp}}, we can show that $|E'| \le O(n)$ w.p. $1 - n^{-6}$, giving $O(m + n)$ total work for Step~\ref{alg:connectivity:s2} of \textsc{Connectivity}$(G)$ w.p. $1 - n^{-6}$. 

    It remains to improve the total work of Step~\ref{alg:connectivity:s5} of \textsc{Connectivity}$(G)$ to $O(m + n)$ with high probability.
    Consider the \textsc{Interweave}$(G', H_1, H_2, E_{\text{filter}}, i)$ in each phase $i \in [0, 10\log\log n)$. 
    In the proof of Lemma~\ref{lem:connectivity_total_work}, we already showed that the total work of Step~\ref{alg:interweave:s6} to Step~\ref{alg:interweave:s10} of \textsc{Interweave}$(G', H_1, H_2, E_{\text{filter}}, i)$ over all phases is at most $O(m + n)$ w.p. $1 - n^{-2}$. 

    Observe that Step~\ref{alg:interweave:s2} and Step~\ref{alg:interweave:s3} take $2(m+n)/(\log n)^2$ work w.p. $1 - 1/\log n$. 
    The idea is to run $\log n$ instances of Step~\ref{alg:interweave:s2} and Step~\ref{alg:interweave:s3} in parallel independently, and halt the execution of an instance when the work of that execution reaches $2 (m + n / (\log n)^2)$, so the total work would be $2(m+n) / (\log n)^2 \cdot \log n$ w.p. $1 - n^{-9}$. 
    To implement $\log n$ instances of parallel execution, note that $|E(H_1)|, |E(H_2)| \le m/(\log n)^6$ w.p. $1 - n^{-8}$. 
    At the end of Step~\ref{alg:connectivity:s3} of \textsc{Connectivity}$(G)$, we approximate compact the arrays of edges of $H_1$ and $H_2$ into arrays of length at most $2m / (\log n)^6$ in $O(\log^* n)$ time and $O(m)$ work w.p. $1 - n^{-9}$. 
    Since the edge of $H_1, H_2$ are indexed, at the beginning of each phase we can use $4m / (\log n)^5$ processors to make $\log n$ copies of $H_1, H_2$. 
    Observe from the proof of Lemma~\ref{lem:new_increase_work}, Lemma~\ref{lem:edge_set_sparse_build_time_work}, and Lemma~\ref{lem:truncate_time_work}, that each atomic step of Step~\ref{alg:interweave:s2} and Step~\ref{alg:interweave:s3} in \textsc{Interweave}$(G', H_1, H_2, E_{\text{filter}}, i)$ uses $O(|V(G')| + |E(H_1)| + |E(H_2)|) \le (m+n)/(\log n)^5$ processors w.p. $1 - 1/(\log n)^2$, so we can use $(m+n)/(\log n)^4$ processors for the $\log n$ independent parallel executions. 
    The same observation implies that Step~\ref{alg:interweave:s2} and Step~\ref{alg:interweave:s3} perform more than $2 (m + n / (\log n)^2)$ work only if the algorithm wants to use more than $(m+n)/(\log n)^5$ processors in an atomic step (e.g., when assigning block of processors to a vertex with level more than $100 \log\log n$), which can be detected and we halt the execution of that instance. 
    %To implement the halting (when an instance reaches $2 (m + n / (\log n)^2)$ work), we use $O(m+n)/(\log n)^4$ processors for each instance to approximate compact the number of working processors in each atomic step in $O(1)$ time to calculate the total work of the execution of that instance in $1$ atomic step. (If there are $\Theta(\log n)$ processors per cell in the array to be compacted then approximate compaction can be computed in $O(1)$ time w.p. $1 - n^{-9}$, see \cite{DBLP:conf/focs/Goodrich91} and \cite{liu2020connected} for details.)
    %Therefore, we can halt the execution of an instance when the accumulated work 
    Finally, for all instances of Step~\ref{alg:interweave:s2} and Step~\ref{alg:interweave:s3} that finishes its execution without halting, pick an arbitrary one to continue the execution of \textsc{Interweave}$(G', H_1, H_2, E_{\text{filter}}, i)$ and discard the others. 
    As a result, Step~\ref{alg:interweave:s2} and Step~\ref{alg:interweave:s3} over all the $10\log\log n$ phases takes $O(m+ n)$ work w.p. $1 - n^{-8}$.

    To bound the total work of Step~\ref{alg:interweave:s4}, we only need to show that Step~\ref{alg:remain:s4} of \textsc{Remain}$(G', H_1)$ takes $O(n)$ work w.p. $1 - n^{-8}$ by the last paragraph in the proof of Lemma~\ref{lem:connectivity_total_work}. 
    Since $|E_{\text{remain}}| \le n /(\log n)^6$ w.p. $1 - n^{-9}$, we run $\log n$ instances of the algorithm in Theorem~\ref{thm:ltz_main} in parallel and halt the execution of the instance if it reaches $n / (\log n)^5$ using the technique described in the previous paragraph. 
    As a result, Step~\ref{alg:interweave:s4} of \textsc{Interweave}$(G', H_1, H_2, E_{\text{filter}}, i)$ takes $O(n)$ work w.p. $1-n^{-8}$. 
    The lemma follows from a union bound.
\end{proof}

Next, we show that our algorithm computes the connected components in desired time with high probability using the implementation in the proof of Lemma~\ref{lem:connectivity_total_work_high_prob}.

\begin{lemma} \label{lem:connectivity_time_correct_final_high_prob}
    For any graph $G$ with component-wise spectral gap $\lambda$, w.p. $1 - 1/n$, \textsc{Connectivity}$(G)$ computes all the connected components of $G$ in $O(\log(1/\lambda) + \log\log n)$ time.
\end{lemma}
\begin{proof}
    In the proofs of Lemma~\ref{lem:connectivity_time_correct_final} and Lemma~\ref{lem:connectivity_total_work_high_prob}, we already showed that Step~\ref{alg:connectivity:s1}, Step~\ref{alg:connectivity:s2}, Step~\ref{alg:connectivity:s3}, Step~\ref{alg:connectivity:s4}, and Step~\ref{alg:connectivity:s6} can be implemented to take $O(\log\log n)$ time w.p. $1 - 2/n^{-2}$. 
    It remains to prove for Step~\ref{alg:connectivity:s5} of \textsc{Connectivity}$(G)$.
    
    In the proof of Lemma~\ref{lem:each_phase_running_time}, we showed that Step~\ref{alg:interweave:s2} and Step~\ref{alg:interweave:s3} of \textsc{Interweave}$(G', H_1, H_2, E_{\text{filter}}, i)$ take $O(\log b)$ time w.p. $1 - 1/(\log n)^2$. 
    In the proof of Lemma~\ref{lem:connectivity_total_work_high_prob}, we run $\log n$ parallel instances of Step~\ref{alg:interweave:s2} and Step~\ref{alg:interweave:s3} of \textsc{Interweave}$(G', H_1, H_2, E_{\text{filter}}, i)$. 
    So we can halt the execution of an instance if it reaches $c \log n$ running time where $c > 0$ is an absolute constant such that Step~\ref{alg:interweave:s2} and Step~\ref{alg:interweave:s3} take at most $c \log b$ time. 
    Therefore, the statement in Lemma~\ref{lem:each_phase_running_time} holds w.p. $1 - n^{-2}$.
    
    In the proof of Lemma~\ref{lem:remain_correct}, we used the fact that w.p. $1 - 1/(\log n)^8$, Step~\ref{alg:interweave:s3} of the algorithm \textsc{Interweave}$(G', H_1, H_2, E_{\text{filter}}, i)$ is a contraction algorithm on $H_1$. 
    One can detect whether this step is a contraction algorithm on $H_1$ by checking whether all vertices in $V(H_1)$ are either roots or children of roots. 
    In the proof of Lemma~\ref{lem:connectivity_total_work_high_prob}, we run $\log n$ parallel instances of Step~\ref{alg:interweave:s2} and Step~\ref{alg:interweave:s3} of \textsc{Interweave}$(G', H_1, H_2, E_{\text{filter}}, i)$. 
    So we can discard an instance at the end of the execution if it is not a contraction algorithm on $H_1$. 
    Therefore, Step~\ref{alg:interweave:s3} is a contraction algorithm on $H_1$ w.p. $1 - n^{-9}$. 
    In \textsc{Remain}$(G', H_1)$ in Step~\ref{alg:interweave:s4}, its first $3$ steps take $O(\log^* n)$ time w.p. $1 - n^{-9}$. 
    In the proof of Lemma~\ref{lem:connectivity_total_work_high_prob}, 
    we run $\log n$ instances of the algorithm in Theorem~\ref{thm:ltz_main} in parallel.  
    So we halt the execution of the instance if it reaches $c (\log d + \log\log n)$ time where $c$ is an absolute constant, which happens w.p. at most $1/(\log n)^9$. 
    As a result, the statement in Lemma~\ref{lem:remain_correct} holds w.p. $1 - n^{-7}$.
    
    In the proof of Lemma~\ref{lem:spectral_gap_b_i}, we proved that in the last phase $i$, w.p. $1 - 1/(\log n)^4$, all edges of $H_1$ are loops at the end of Step~\ref{alg:interweave:s3} of \textsc{Interweave}$(G', H_1, H_2, E_{\text{filter}}, i)$, which can be detected using $|E(H_1)| \le m / (\log n)^6$ processors in $O(1)$ time w.p. $1 - n^{-8}$. 
    In the proof of Lemma~\ref{lem:connectivity_total_work_high_prob}, we run $\log n$ parallel instances of Step~\ref{alg:interweave:s2} and Step~\ref{alg:interweave:s3}, so w.p. $1 - n^{-9}$, at least one of the instance at the end of its execution satisfies that all edges of $H_1$ are loops and we choose this one to continue the execution of \textsc{Interweave}$(G', H_1, H_2, E_{\text{filter}}, i)$ and discard the others. 
    Therefore, the statement in Lemma~\ref{lem:spectral_gap_b_i} holds w.p. $1 - n^{-7}$. 
    
    Since Lemma~\ref{lem:each_phase_running_time}, Lemma~\ref{lem:remain_correct}, and Lemma~\ref{lem:spectral_gap_b_i} hold w.p. $1 - 1/n$, the rest of the proof is the same as the proof of Lemma~\ref{lem:connectivity_time_correct_final}.
\end{proof}

By Lemma~\ref{lem:connectivity_total_work_high_prob} and Lemma~\ref{lem:connectivity_time_correct_final_high_prob}, Theorem~\ref{thm:our_main_result} follows immediately.

\appendix

\section{Conditional Lower Bound} \label{sec:lowerbound}

In this section, we show that our time bound is optimal when the input graph has component-wise spectral gap at most $1/\poly(\log(n))$ conditioning on the widely believed 
\textsc{2-Cycle} Conjecture (see below). 
The approach is similar to that of \cite{DBLP:conf/focs/BehnezhadDELM19}, who gives an $\Omega(\log d)$ conditional lower bound, where $d \ge \poly(\log(n))$ is the diameter.

\begin{conjecture}[\textsc{2-Cycle} Conjecture \cite{DBLP:journals/jacm/RoughgardenVW18,DBLP:conf/podc/AssadiSW19,DBLP:conf/focs/Andoni,DBLP:conf/focs/BehnezhadDELM19}]
    Any MPC algorithm that uses $n^{1 - \Omega(1)}$ space per processor requries $\Omega(\log n)$ MPC rounds to distinguish between one cycle of size $n$ from two cycles of size $n/2$ w.h.p.
\end{conjecture}

For simplicity, we state the theorem in the MPC model so that our proof can be based on the MPC model. 
The conditional lower bound easily holds in any PRAM model with $O(m+n)$ total memory by a simple simulation \cite{DBLP:conf/podc/AssadiSW19, DBLP:conf/focs/Andoni}: if there exists a PRAM algorithm that distinguishes between one cycle from two cycles in $o(\log n)$ time, then an MPC with $n^{1 - \Omega(1)}$ space per processor can simulate this PRAM algorithm with $O(1)$ overhead, which violates the conjecture.

\begin{theorem} \label{thm:lower_bound}
    If the \textsc{2-Cycle} Conjecture is true, then there exists a constant $c > 20$ such that any MPC algorithm with $n^{1 - \Omega(1)}$ space per processor that w.h.p. computes connected components of any given $n$-vertex graph with component-wise spectral gap $\lambda'$ requires $\Omega(\log(1/\lambda'))$ rounds, where $\lambda' \le \log^{-c} n$ is a fixed parameter.
\end{theorem}

For the contradiction, we assume that there exists an MPC algorithm \textsf{ALG} with $n^{1 - \Omega(1)}$ space per processor that w.h.p. computes connected components of any given $n$-vertex graph with component-wise spectral gap $\lambda'$ in $o(\log(1/\lambda'))$ rounds, where $\lambda'$ is given in Theorem~\ref{thm:lower_bound}. 
%(The MPC algorithm \textsf{ALG} does \emph{not} need to know the component-wise spectral gap of the input graph to compute its connected components - it only needs to be \emph{correct} w.h.p. when the input graph has spectral gap at most $\lambda'$.)

Given a graph $G$ which consists a cycle of size $n$ or two cycles of size $n/2$. 
The algorithm proceeds in rounds. 
In each round, delete each edge of $G$ w.p. $\lambda'^{1/4}$ independently to get graph $G'$. 
Compute the connected components of $G'$ by \textsf{ALG}.
Contract each connected components of $G'$ to a single vertex in $G$ and add the edges deleted in this round to $G$, completing this round. 
Repeat for $R$ rounds until the number of edges in $G$ is $n^{o(1)}$ such that $G$ can be stored in one processor and solved where $R$ is a parameter to be tuned. 
%, or until $G$ has spectral gap at least $\lambda'$ such that \textsf{ALG} solves it in $o(\log(1/\lambda'))$ rounds.
We need the following claim.
\begin{lemma}\label{lem:lb_path_cycle}
    The graph $G$ at the beginning of each round and subgraph $G'$ during each round is a collection of paths and cycles.
\end{lemma}
\begin{proof}
    The statement is trivially true at the beginning of the algorithm. 
    By an induction, assuming the graph $G$ at the beginning of each round is a collection of paths and cycles. 
    During this round, $G'$ is a subgraph of $G$ after edge deletions, which is still a collection of paths and cycles (singleton vertex is considered as a path or a cycle of size $1$, where \emph{size} is defined as the number of vertices). 
    After contracting connected components of $G'$, the resulting $G$ is also a collection of paths and cycles, completing the induction.
\end{proof}
\begin{lemma} \label{lem:lb_round}
    Each round of the algorithm takes $o(\log(1/\lambda'))$ MPC rounds w.h.p.
\end{lemma}
\begin{proof}
    The edge deletion at the beginning of the round and the contraction at the end of the round can be implemented in $O(1)$ MPC rounds (e.g., see \cite{DBLP:conf/focs/Andoni}). 
    It remains to show that the graph $G'$ ($G$ after edge deletions) has spectral gap at least $\lambda'$ w.h.p. at the beginning of each round.\footnote{If graph $G'$ has component-wise spectral gap greater than $\lambda'$, the algorithm creates a \emph{dummy} component of spectral gap $\lambda'$ and adds it into $G'$ such that $G'$ has component-wise spectral gap exactly $\lambda'$ and thus \textsf{ALG} is applicable on $G'$. The dummy component has no edge with the original $G'$ so does not influence the computation of other components. Moreover, the dummy component has at most $m$ edges (e.g. a path of fixed length) and can be created and deleted at the end of this round in $O(1)$ MPC rounds.}
    
    By Lemma~\ref{lem:lb_path_cycle}, $G'$ is a collection of paths and cycles. 
    In the following, we show that w.h.p. the size of each path and cycle is at most $\lambda'^{-1/3}$, then a union bound gives the result since there are at most $n$ such paths and cycles.
    Fix a path/cycle of size $\lambda'^{-1/3} + 1$, the number of edges in it is at least $\lambda'^{-1/3}$. 
    The probability that none of these edges gets deleted in this round is at most
    $$ \left(1 - \lambda'^{1/4} \right)^{\lambda'^{-1/3}} \le \left(1 - \lambda'^{1/4} \right)^{\lambda'^{-1/4} \cdot 5 \log n} \le n^{-3} $$
    by $\lambda' \le \log^{-c} n$ and $c > 20$. 
    By a union bound, all paths and cycles in $G'$ have size at most $\lambda'^{-1/3}$ w.h.p.
    By Definition~\ref{def:conductance} and the fact that $G'$ is a collection of paths and cycles (Lemma~\ref{lem:lb_path_cycle}), the conductance $\phi_{G'}$ of $G'$ is at least $2 \lambda'^{1/3}$. 
    By Cheeger's inequality, w.h.p. the spectral gap of $G'$ is at least ${\phi_{G'}}^2 / 2 \ge \lambda'^{2/3} \ge \lambda'$.
    
    Finally, by our assumption, \textsf{ALG} computes the connected components of $G'$ in $o(\log(1/\lambda'))$ rounds w.h.p., giving the lemma.
\end{proof}

\begin{lemma}\label{lem:lb_shrink}
    If $G$ has $m$ edges at the beginning of a round, then w.h.p. $G$ has at most $\max\{3^{\sqrt{\log n}}, m \lambda'^{1/5}\}$ edges at the end of that round.
\end{lemma}
\begin{proof}
    Observe that after \textsf{ALG} contracts all edges of $G'$ in this round, the resulting edges in $G$ are exactly those deleted edges in this round. 
    The expected number of deleted edges is $m\lambda'^{1/4}$. 
    If the number of deleted edges is at most $3^{\sqrt{\log n}}$, then the lemma holds.
    Otherwise, by a Chernoff bound, the number of deleted edges is at most $2m \lambda'^{1/4}$ w.p. at least
    $$1 - \exp\left( -m\lambda'^{1/4} / 3\right) \ge 1 - \exp\left(- 3^{\sqrt{\log n}} / 6 \right) \ge 1 - n^{-9}, $$
    where the second inequality follows from the fact that $2m \lambda'^{1/4} \ge 3^{\sqrt{\log n}}$. 
    The lemma follows from $2m \lambda'^{1/4} \le m \lambda'^{1/5}$.
\end{proof}

Now set parameter $R = -10 \log n / \log \lambda'$. 
(Recall that $R$ is the number of rounds of the algorithm.)
Initially the graph $G$ has at most $n^2$ edges. 
Applying Lemma~\ref{lem:lb_shrink} for $R$ times, we have that the number of edges at the end of round $R$ is at most
$$ 3^{\sqrt{\log n}} +  n^2 \cdot (\lambda'^{1/5})^R \le 3^{\sqrt{\log n}} +  n^2 \cdot \lambda'^{-2 \log n / \log \lambda'} \le 3^{\sqrt{\log n}} + 1 = n^{o(1)} $$
w.h.p. by a union bound over all $R \le \poly\log(n)$ rounds. (We use the fact that the spectral gap $\lambda' \ge 1/\poly(n)$.)

By Lemma~\ref{lem:lb_round}, the total number of rounds are at most 
$$R \cdot o(\log(1/\lambda')) = o(\log(1/\lambda') \cdot (-10 \log n / \log \lambda')) = o(\log n) $$
w.h.p., giving an MPC algorithm that distinguishes one cycle of size $n$ with two cycles of size $n/2$ in $o(\log n)$ rounds, a contradiction.

\section{Edge Sampling Does Not Preserve Diameter} \label{sec:diameter}

%\begin{figure*}[h]
%\centering
%\includegraphics[width=0.95\textwidth]{example}
%\vspace{-5pt}
%\caption{Example: random sampling does not preserve diameter.}
%\end{figure*}

In this section we show that edge sampling does not preserve the diameter of graphs, which partially explains why an $O(\log d)$-time, linear-work PRAM algorithm is hard to obtain. 

We provide an undirected graph $G$ with a diameter of $O(\log n)$ such that if we sample every edge in $G$ with probability of $p=\frac{1}{\log n}$, the largest diameter in the connected components of $G$ would be at least $\frac{n}{\poly (\log n)}$ where $n$ is the number of vertices in the graph. 
This shows that edge sampling can dramatically increase the diameter of a graph. 

We construct the graph $G=(V,E)$ as follows. The vertex set of $G$ consists of  $V=(\{x\} \cup P \cup Z)$. Here, $Z$  consists of the vertices $\{z_1, z_2,  \cdots, z_k\}$ where $k$ is a parameter determined later. Also, $P= P_1 \cup P_2 \cup \cdots \cup P_k$ where each $P_i$ is a vertex-disjoint path of length $\log n$ between $x$ and $z_i$.  In the graph $G$ we further add an edge between any two vertices $z_i$ and $z_j$ such that $|i-j| \le (\log n)^3$. 
Note that in the mentioned construction, we have $n=|V|=1+k + k \log n$. This implies $k= \frac{n-1}{\log n +1}$.

It is clear that the graph $G$ has a diameter of at most $2 \log n$ since every vertex is within distance of at most $\log n$ from $x$. Now let $\widetilde G$ be a graph with the same vertex set as $G$, and we add every edge in $G$ to $\widetilde G$ independently with the probability $p= \frac{1}{\log n}$. We claim that w.h.p. the largest diameter in  $\widetilde G$ is $\Omega (\frac{n}{\poly (\log n)})$. First notice that w.h.p.  $x$ is not connected to any of the vertices $z_1, z_2, \cdots, z_k$. 
This is because there was a path $P_i$ with length $\log n$ between $x$ and each $z_i$ in $G$. The probability that we sample all of the edges of $P_i$ and include them in $\widetilde G$ is $(\frac{1}{\log n})^ {\log n}$ which is at most $1/{n^3}$ for large enough $n$. Thus, with probability at least $1- 1/n^3$, $x$ is not connected to $z_i$ via the path $P_i$. By taking the union bound, we can say that  with probability at least $1- 1/n^2$, $x$ is not connected to any of the vertex in $Z$. 

Now we show that w.h.p. the subgraph induced by the vertices of $Z$ remains connected. This immediately implies that the diameter in connected component of the vertices of $Z$ is at least $\Omega (\frac{n}{\poly (\log n)})$. Consider the shortest path between $z_1$ and $z_k$. Since none of the vertices $z_i$ are connected to $x$, the shortest way to travel from $z_1$ and $z_k$ is by using the vertices in $Z$. Recall that each $z_i$ is connected to $z_{j}$ in $G$ iff $|i-j| \le \log^2 n$. Thus, in order to travel from $z_1$ to $z_k$  we have to use at least $\Omega(\frac{k}{(\log n)^3}) = \Omega(\frac{n}{\poly(\log n)})$ edges.  Thus, the diameter in the connected component of the vertices of $Z$ is at least $\Omega (\frac{n}{\poly (\log n)})$. It is now sufficient to show that w.h.p. the subgraph induced by the vertices of $Z$ remains connected in $\widetilde G$.

To show the connectivity, first consider the vertices $z_1, \cdots, z_{(\log n)^3}$. We show that these vertices form a connected subgraph in $\widetilde G$. Notice that these vertices form a complete induced subgraph in $G$. Since we sample each edge with probability $p= 1/\log n$, we can say that in $\widetilde G$, the subgraph induced by these vertices has the same distribution as $G((\log n)^3, p)$. Here,  $G(n,p)$ is used to denote the Erdős–Rényi graph with $n$ vertices where each edge is included in it independently with probability $p$. 
It is known that $G(n,p)$ is connected with probability $e^{-e^{-c}}$ if $p \ge \frac{(1 + c) \ln n}{n}$ \cite{erdHos1960evolution}. 
Recall that in the random subgraph $G((\log n)^3, p)$, we have $p= 1/\log n$. 
Thus, $p \gg \frac{3\ln n \ln ((\log n)^3)}{(\log n)^3}$ for large enough $n$ which implies the connectivity of this random graph w.p. $1-1/n^2$. 

So far, we have showed that w.h.p. the induced subgraph of the first $(\log n)^3$ vertices of $Z$ is connected. 
Now consider a vertex $z_i$ where $i > (\log n)^3$, we show that w.h.p. $z_i$ has an edge in $\widetilde G$ to some other vertex $z_j$ where $j <i$. 
Recall that there is an edge between $z_i$ and each of the vertices $z_{i-1}, \cdots, z_{i- (\log n)^3}$ in $G$.
Thus, the probability that we include none of these edges in $\widetilde G$ is $(1- 1/\log n)^{(\log n)^3}$ which is at most $1/n^3$ for large enough $n$. Therefore, we can say that with probability at least $1-1/n^3$ there is an edge between $z_i$ and another vertex $z_j$ where $j <i$. By taking the union bound, we can say that w.h.p. for every vertex $z_i$ where $i > (\log n)^3$, there is an edge between $z_i$ and some other vertex $z_j$ where $j<i$. This is enough to show that the induced subgraph of vertices $Z$ is connected. 
To see that, we show w.h.p. there is a path from every vertex $z_i$ to $z_1$ in $\widetilde G$. Consider a vertex $z_i$, if $i > (\log  n)^3$, we know that there is an edge between $z_i$ and some other vertex $z_j$ where $j<i$. Thus, we can travel from $z_i$ to $z_j$ as long as $i > (\log n)^3$. By traveling to these neighbors repeatedly, we eventually reach a vertex $z_j$ such that $j \le (\log n)^3$. In this case we already know that the induced subgraph of the first $(\log n)^3$ vertices of $Z$ is connected. Thus, there is a path from $z_j$ to $z_1$ in $\widetilde G$ which implies that there is also a path from $z_i$ to $z_1$.

\newcommand{\R}{\ensuremath{\mathbb{R}}\xspace}

\section{Spectrum of Randomly Sampled Graphs}
\label{sec:spectral}
We prove a concentration bound for the spectrum of randomly sampled graphs.
Our proof is an adaptation of that of Oliveira~\cite{oliveira}. %\elaine{cite} 
Specifically, known concentration bounds 
for the spectrum of randomly sampled graphs~\cite{chunghorn,oliveira,chungradcliffe} 
%\elaine{cite more papers}
are stated for graphs without parallel edges.
In our application, we need essentially the same concentration
bounds for multi-graphs, i.e., allowing parallel edges and loops.
It is not too difficult to extend the proof of Oliveira~\cite{oliveira} %\elaine{cite}
to multi-graphs.  For completeness, we provide the 
modified proof in this section.

\subsection{Preliminaries}
Suppose $x \in \R^n$ is a vector over the reals, we use
the notation $\|x\|$ to denote the Euclidean norm of $x$.
\begin{definition}[Spectral radius norm]
Let $A \in \R^{n \times n}$ be a symmetric matrix over the reals.
We define the spectral radius norm of $A$, denoted $\| A\|$, 
to be the following:
\[
\|A\| = \sup_{x \in \R^n, \|x\| = 1} \|A x\|
\]
\end{definition}

Let $\lambda_0(A) \leq \lambda_1(A) \leq \ldots \leq \lambda_{n-1}(A)$ be the 
eigenvalues of $A$, which correspond to orthonormal eigenvectors 
$\psi_0, \ldots, \psi_{n-1}$.
The {\it spectrum} of $A$ is the set of all eigenvalues of $A$.
Since $A$ is symmetric, the following holds:
\[
\|A\| = \max_{0 \leq i \leq n-1} |\lambda_i(A)|
\]

\begin{theorem}[Matrix concentration bound~\cite{oliveira,chungradcliffe}] 
Let $X_1, \ldots, X_m$ be mean-zero independent $n \times n$ 
random matrices over the reals.
Suppose that there exists a $B > 0$ such that 
for all $1 \leq i \leq m$, 
$\| X_i\| \leq B$; %with probability \elaine{FILL}, 
and let $\sigma^2 = \lambda_{\rm max}\left(\sum_{i = 1}^m \E[X_i^2]\right)$.
Then, for all $t \geq 0$, 
\[
\Pr\left(\left\| \sum_{i = 1}^m X_i\right\| \geq t\right)
\leq  
2 n \cdot \exp\left(-\frac{t^2}{8 \sigma^2 + 4 B t}\right) 
\]
\label{thm:matrixconcentr}
\end{theorem}

\subsection{Spectrum of Randomly Sampled Graphs}

Consider a graph $\widetilde{\graph}$ that is down-sampled from 
some original graph $\graph$ through the following process:
\begin{itemize}[leftmargin=5mm,itemsep=1pt]
\item 
$\widetilde{\graph}$ has the same set of vertices as $\graph$;
\item 
for every edge (including loops and parallel edges)
in $\graph$, 
with probability $p \in (0, 1)$, preserve the edge in $\widetilde{\graph}$.
\end{itemize}

We want to show that the down-sampled graph  
$\widetilde{\graph}$ approximately preserves the spectrum 
of the original graph $\graph$, as 
stated in the following theorem:

\begin{theorem}[Spectrum of randomly sampled graph]
Let $\deg(\graph)$ be the minimum degree of any vertice in $\graph$,
let $n$ be the number of vertices in $\graph$, and let $N \geq n$.
For any constant $c > 0$, there exists 
a constant $C = C(c) > 0$, 
such that 
if $p \cdot \deg(\graph) \geq C \cdot \ln N$, then 
for $N^{-c} \leq \delta \leq 1/2$,
\[
\Pr\left(
\|\mcal{L}(\widetilde{\graph}) - \mcal{L}(\graph) \| \leq 
13\sqrt{\frac{\ln(\frac{4n}{\delta})}{p \cdot \deg(\graph)}}
\right) \geq 1-\delta
\]
\label{thm:spectrumsample}
\end{theorem}

\begin{proof}
The expected degree of every vertex $v$ in $\widetilde{\graph}$ is at least 
$\E[d_{\widetilde{\graph}}(v)] = p \cdot d_{\graph}(v) \geq p 
\cdot \deg(\graph)$.  %\geq C \cdot \ln n$.
By the standard Chernoff bound, 
for a sufficiently large constant $C$, we have 
\[
\forall \text{vertex } v: \Pr\left(\left| 
\frac{d_{\widetilde{\graph}}(v)}{ p \cdot d_{\graph}(v)} - 1
\right| > 
2 \sqrt{\frac{\ln \frac{4n}{\delta}}{p \cdot d_{\graph}(v)}}\right) 
\leq \frac{\delta}{2n}   
\]
By the union bound, 
with probability at least $1-\delta/2$, 
we have that 
\[
\forall \text{vertex } v \in [n]:
\left| \frac{d_{\widetilde{\graph}}(v)}{ p \cdot d_{\graph}(v)} - 1
\right| \leq 2 \sqrt{\frac{\ln \frac{4n}{\delta}}{p \cdot d_{\graph}(v)}}
\]
Since $\delta \geq N^{-c}$, 
for a sufficiently large constant $C$, the above implies that  
with probability at least $1-\delta/2$, 
\begin{equation}
\forall \text{vertex } v \in [n]:
\left| \frac{ p \cdot d_{\graph}(v)}{d_{\widetilde{\graph}}(v)} - 1
\right| \leq 4 \sqrt{\frac{\ln \frac{4n}{\delta}}{p \cdot d_{\graph}(v)}}
\leq 4 \sqrt{\frac{\ln \frac{4n}{\delta}}{p \cdot \deg(\graph)}}
\label{eqn:degconcentr}
\end{equation}

We will the above inequality to compare the matrices:
\[
\widetilde{T} = \text{diagonal with } d_{\widetilde{\graph}}(v)^{-1/2} \text{ at the $(v, v)$-th position}
\]
and 
\[
T = \text{diagonal with } (p \cdot d_{\graph}(v))^{-1/2} \text{ at the $(v, v)$-th position}
\]
For a sufficiently large constant $C$, the right-hand-side 
of Equation~(\ref{eqn:degconcentr}) is at most $3/4$.
Observe that for any $x \in [-3/4, 3/4]$, 
it must be 
$\left| \sqrt{1 + x} - 1\right| \leq x$.
Applying this fact to $x = \frac{p \cdot d_{\graph}(v)}{d_{\widetilde{\graph}}(v)} - 1$, 
we have  the following with probability 
at least $1-\delta/2$:
\begin{equation}
\|T \widetilde{T}^{-1} - I \| = \max_{1 \leq v \leq n}\left| 
\frac{\sqrt{p \cdot d_{\graph}(v)}}{\sqrt{d_{\widetilde{\graph}}(v)}} - 1\right|
\leq 4\sqrt{\frac{\ln(4n/\delta)}{p \cdot \deg(\graph)}}
\label{eqn:diagconcentr}
\end{equation}

Let $A = p \cdot A(\graph)$ where $A(\graph)$ is the adjacency matrix
of $\graph$, and let $\widetilde{A} = A(\widetilde{\graph})$  be the adjacency matrix
of the sampled graph $\widetilde{\graph}$.
Recall that for a multi-graph, the  
$(u, v)$-position in the adjacency matrix stores the 
number of edges between $u$ and $v$.
We now compare $\widetilde{\mcal{L}} = I - \widetilde{T} \widetilde{A} \widetilde{T}$
and $\mcal{L} = I - T A T$, where $\widetilde{\mcal{L}}$ 
is the normalized Laplacian of $\widetilde{\graph}$ and $\mcal{L}$ is the normalized
Laplacian of $\graph$.
To compare the two, we will introduce
an intermediate matrix 
\[
\mcal{M} = I - T \widetilde{A} T = 
I - (T \widetilde{T}^{-1}) ( I - \widetilde{\mcal{L}}) (T \widetilde{T}^{-1})
\]

We have that 
\begin{align*}
\| \mcal{M}- \widetilde{\mcal{L}}\| & = 
\| (T \widetilde{T}^{-1}) (I - \widetilde{\mcal{L}}) (T \widetilde{T}^{-1})
- (I - \widetilde{\mcal{L}}) \|\\
& \leq \| (T \widetilde{T}^{-1} - I) (I - \widetilde{\mcal{L}}) (T \widetilde{T}^{-1})
- (I - \widetilde{\mcal{L}}) (T \widetilde{T}^{-1} - I)  \|  \\
& \leq 
\| T \widetilde{T}^{-1} - I\| \cdot \|I - \widetilde{\mcal{L}}\| \cdot
\| T \widetilde{T}^{-1}\|
+ \|I - \widetilde{\mcal{L}}\| \cdot \|T \widetilde{T}^{-1} - I \|
%\hfill (\star)
\end{align*}

Now, since the spectrum of any normalized Laplacian matrix lies 
in $[0, 2]$, it must be that $\|I - \widetilde{\mcal{L}}\| \leq 1$.
Using this fact as well as 
Equation~(\ref{eqn:diagconcentr}), we have that 
\begin{align*}
\| \mcal{M}- \widetilde{\mcal{L}}\| \leq
4\sqrt{\frac{\ln(4n/\delta)}{p \cdot \deg(\graph)}}
\cdot \left(1 + 4\sqrt{\frac{\ln(4n/\delta)}{p \cdot \deg(\graph)}}\right)
+ 4\sqrt{\frac{\ln(4n/\delta)}{p \cdot \deg(\graph)}}
\leq 9 \sqrt{\frac{\ln(4n/\delta)}{p \cdot \deg(\graph)}}
\end{align*}
where the last inequality again relies 
on the fact that $C$ is a sufficiently large constant.

To finish the proof, it suffices
to show that $\|\mcal{M} - \mcal{L}\| \leq 
4\sqrt{\frac{\ln(4n/\delta)}{p \cdot \deg(\graph)}}$
with probability at least $1-\delta/2$.
Let $e = (u, v)$ be an edge in $\graph$,
we define the matrix $A_e$ as the adjacency matrix
for a single edge $e$:
\[
A_e = \begin{cases}
{\bf b}_u {\bf b}_v^T
+ {\bf b}_v {\bf b}_u^T & \text{if $u \neq v$}\\
{\bf b}_u {\bf b}_u^T & \text{if $u = v$}
\end{cases}
\]
where ${\bf b}_u$
denote the $u$-th vector in the standard basis.

Henceforth, 
for every edge $e \in \graph$, 
let $I_e$
be the indicator random variable that denotes 
whether the edge $e$ is preserved in the down-sampled graph $\widetilde{\graph}$;
define the random matrix $X_e = (I_e - p) \cdot A_e$.
We can rewrite $\mcal{L}-\mcal{M}$ as the following
where the summation over $e \in \graph$ iterates over all 
parallel edges (including loops) with the corresponding multiplicity:
\[
\mcal{L}-\mcal{M} = \sum_{e \in \graph} T X_e T 
\]
Observe that each $X_e$ is an independent mean-0 matrix.
Let $Y_e = T X_e T$, we have the following
where we use $u_e, v_e$ to iterate
over the two vertices 
incident to the edge $e$: 
\[
Y_e =  T \cdot (I_e - p) \cdot A_e \cdot T 
= (I_e - p) \cdot \frac{A_e}{p \cdot \sqrt{d_{\graph}(u_e) \cdot {d_{\graph}(v_e)}}}
\]
The possible eigenvalues of $Y_e$ 
are contained in the following set:
\[
\left\{
\frac{\pm\left(1-p\right)}{p \cdot \sqrt{d_{\graph}(u_e) \cdot {d_{\graph}(v_e)}}}, 
\frac{\pm p}{p \cdot \sqrt{d_{\graph}(u_e) \cdot {d_{\graph}(v_e)}}}, 0
\right\}
\]
Therefore, 
\[
\|Y_e\|
\leq \frac{1}{p \cdot \sqrt{d_{\graph}(u_e) \cdot {d_{\graph}(v_e)}}}
\leq \frac{1}{p \cdot \deg(\graph)}
\]

The sum of variance is the following
where the notation $\sum_{e = (u, v) \in \graph}$
iterates over all edges adjacent to some fixed $u$ in $\graph$:
\begin{align*}
\sum_{e \in \graph} \E[Y_e^2] & =  
\sum_{e \in \graph} \E[(I_e - p)^2] \cdot \left(\frac{A_e}{p \cdot \sqrt{d_{\graph}(u_e) \cdot {d_{\graph}(v_e)}}} \right)^2 \\
& = \sum_{(u, v) \in \graph, u \neq v} 
p(1-p) \cdot \frac{{\bf b}_u {\bf b}_u^T + {\bf b}_v {\bf b}_v^T}{p^2 \cdot 
d_{\graph}(u) \cdot d_{\graph}(v)}
+ \sum_{(u, u) \in \graph} p(1-p) \cdot 
 \frac{{\bf b}_u {\bf b}_u^T}{p^2 \cdot
d_{\graph}(u)^2} \\
& = \frac{1-p}{p} \cdot \sum_{u \in \graph} \frac{1}{d_{\graph}(u)}
\sum_{e = (u, v) \in \graph} \frac{{\bf b}_u {\bf b}_u^T}{d_{\graph}(v)}
\end{align*}
The above is a diagonal matrix, and its $(u, u)$-th entry 
is at most: 
\begin{align*}
\frac{1-p}{p} \cdot \frac{1}{d_{\graph}(u)} 
\cdot \sum_{e = (u, v) \in G} \frac{1}{d_{\graph}(v)}
& \leq 
\frac{1-p}{p} \cdot \frac{1}{d_{\graph}(u)} 
\cdot \sum_{e = (u, v) \in G} \frac{1}{\deg(\graph)}\\
& =  
\frac{1-p}{p} \cdot \frac{1}{d_{\graph}(u)} 
\cdot d_{\graph}(u) \cdot \frac{1}{\deg(\graph)}\\
& \leq \frac{1}{p \cdot \deg(\graph)}
\end{align*}
Now, we apply the matrix concentration theorem, i.e., 
Theorem~\ref{thm:matrixconcentr}, 
to $\sum_{e}Y_e$
with $B = \sigma^2 = \frac{1}{p \cdot \deg(\graph)}$.
We thus have 
\[
\Pr\left(\|\mcal{L} - \mcal{M}\| \geq t \right)
\leq 2n \cdot \exp\left(-\frac{t^2 \cdot p \cdot \deg(\graph)}{8 + 4t}\right)
\]
Plugging in $t = 4 \sqrt{\frac{\ln(4n/\delta)}{p \cdot \deg(\graph)}}\leq 3/4$, 
we have 
\[
\Pr\left(\|\mcal{L} - \mcal{M}\| \geq 4 \sqrt{\frac{\ln(4n/\delta)}{p \cdot \deg(\graph)}}\right)
\leq 
2n \cdot \exp\left(-\frac{16 \ln(4n/\delta)}{16}\right) \leq \delta/2  
\]
\end{proof}

\begin{corollary}[Random sampling approximately preserves spectral gap]  
Let $n$ be the number of vertices of $\graph$ and let $N \geq n$.
Let $c > 0$ be an arbitrary constant.
There exist constants $C$
and $C'$ such that as long as $p \cdot \deg(\graph) \geq C \ln N$,
then with probability at least $1-1/N^c$, 
\[
|\lambda_1(\mcal{L}) - \lambda_1(\widetilde{\mcal{L}})|
\leq C' \cdot \sqrt{\frac{\ln n}{p \cdot \deg(\graph)}} 
\]
where $\mcal{L} = \mcal{L}(\graph)$
and $\widetilde{\mcal{L}}= \mcal{L}(\widetilde{\graph})$
denote the normalized Laplacian matrices of the original  graph $\graph$ 
and the sampled graph $\widetilde{\graph}$, respectively.
\label{cor:samplespectral1}
\end{corollary} 
\begin{proof}
\begin{align*}
\lambda_1(\widetilde{\mcal{L}})  = 
\min_{{\bf 1}^T x = 0,  \|x \| = 1} x^T \widetilde{\mcal{L}}x 
& = 
\min_{{\bf 1}^T x = 0,  \|x \| = 1} x^T 
\left(\mcal{L} + (\widetilde{\mcal{L}} - \mcal{L}) \right) x \\
&  \geq
\min_{{\bf 1}^T y = 0,  \|y \| = 1} y^T \mcal{L} y
-  \| \widetilde{\mcal{L}} - \mcal{L} \| \\
& = 
\lambda_1({\mcal{L}})  -  \| \widetilde{\mcal{L}} - \mcal{L} \| \\
\end{align*}
Similarly, we can show the other direction, that is, 
$\lambda_1({\mcal{L}}) \geq \lambda_1(\widetilde{\mcal{L}}) 
- \| \widetilde{\mcal{L}} - \mcal{L} \| $.
The corollary now follows in a straightforward manner 
from Theorem~\ref{thm:spectrumsample}.
\end{proof}

\begin{corollary}[Random sampling approximately preserves the spectral gap of
every connected component]  
Suppose that the graph $\graph$ contains $k$ connected components
denoted $\graph_1, \ldots, \graph_k$.
Let $N$ be the number of vertices of $\graph$. %and let $N \geq n$.
Let $c > 0$ be an arbitrary constant.
There exist constants $C$
and $C'$ such that as long as $p \cdot \deg(\graph) \geq C \ln N$,
then with probability at least $1-1/N^c$, 
it must be that for all $i \in [k]$, 
\[
|\lambda_1(\mcal{L}(\graph_i)) - \lambda_1({\mcal{L}(\widetilde{\graph_i})})|
\leq C' \cdot \sqrt{\frac{\ln N}{p \cdot \deg(\graph)}} 
\]
where $\mcal{L}(\graph)$
and $\mcal{L}(\widetilde{\graph})$
denote the normalized Laplacian matrices of the original connected component $\graph_i$ 
and the connected component $\widetilde{\graph}_i$ after sampling, respectively.
\label{cor:sample_preserve}
\end{corollary}  
\begin{proof}
Follows in a straightforward fashion
from 
Corollary~\ref{cor:samplespectral1}
by taking a union bound over all the at most $N$ connected components.
\end{proof}

%\paragraph{Acknowledgments.} 
%We thank Robert E. Tarjan and Peilin Zhong for insightful discussions, and thank Noga Alon for helping us on the construction that edge sampling does not preserve diameter. 

\newpage
\addcontentsline{toc}{section}{Bibliography}

\bibliographystyle{alpha}
\bibliography{FLS22}

\end{document}